\documentclass[11pt, reqno]{amsart}
\usepackage[left=2.73cm, right=2.73cm, top=3.2cm, bottom=3.2cm]{geometry}
\usepackage{amsmath,amssymb,amsthm,mathrsfs,mathtools,nicefrac,cancel,graphicx,xcolor,tikz,comment,listings,hyperref,enumerate}
\usepackage[ruled]{algorithm2e}
\usepackage[T1]{fontenc}
\usepackage[utf8]{inputenc}
\usepackage[shortlabels]{enumitem}
\usepackage{fancyhdr}
\usepackage{dirtytalk}
\usepackage{bm}
\usepackage{halloweenmath}

\usetikzlibrary{cd,decorations.pathreplacing,matrix,arrows,positioning,automata,shapes,shadows,calc,fadings,decorations,through,intersections}
\usepackage{float}
\AtBeginDocument{%
   \def\MR#1{}
}

\newtheorem{theorem}{Theorem}
\newtheorem{corollary}[theorem]{Corollary}
\newtheorem{lemma}[theorem]{Lemma}
\newtheorem{proposition}[theorem]{Proposition}

\newtheorem*{axioms}{Axioms}
\theoremstyle{definition} 
\newtheorem{definition}[theorem]{Definition}
\let\olddefi\definition
\renewcommand{\definition}{\olddefi\normalfont}
\newtheorem{question}{Question}
\let\oldquestion\question
\renewcommand{\question}{\oldquestion\normalfont}
\newtheorem{example}[theorem]{Example}
\let\oldexample\example
\renewcommand{\example}{\oldexample\normalfont}
\newtheorem{remark}[theorem]{Remark}
\let\oldremark\remark
\renewcommand{\remark}{\oldremark\normalfont}
\newtheorem{claim}{\textit{Claim}}

\AtBeginEnvironment{example}{%
  \pushQED{\qed}%
}
\AtEndEnvironment{example}{\popQED\endexample}

\usepackage{todonotes}

\hypersetup{
    pdftitle={Prova},
    pdfauthor={},
    pdfmenubar=false,
    pdffitwindow=true,
    pdfstartview=FitH,
    colorlinks=true,
    linkcolor=blue,
    citecolor=green,
    urlcolor=cyan
}

\providecommand{\MR}[1]{}

\providecommand{\MR}{\relax\ifhmode\unskip\space\fi MR }

\newcommand{\tim}{{\mkern-0mu\times\mkern-0mu}}
\newcommand{\bb}{{\boldsymbol{\beta}}}
\newcommand{\ga}{{\boldsymbol{\gamma}}}
\newcommand{\ddd}{{\dd^\mu_\bb}}
\newcommand{\pp}{{\boldsymbol{\phi}}}
\newcommand{\R}{\mathbb{R}}

\newcommand{\V}{{\mathbf V_n}}
\newcommand{\C}{[n]}

\newcommand{\LL}{\mathop{}\!\mathbb{S}}

\newcommand{\dd}{\mathrm{d}}
\newcommand{\Id}{\mathrm{I}}

\newcommand{\1}{\mathbf 1}
\DeclareMathOperator*{\argmin}{arg\,min}
\DeclareMathOperator*{\argmax}{arg\,max}

\usepackage{tikz,tikz-3dplot}
\tdplotsetmaincoords{60}{150}    

\title[Distances
on Rankings Axioms, Aggregation and Approximation]{On the Weighted Top-Difference Distance\\ Axioms, Aggregation, and Approximation}

\begin{document}

\author[A.~Aveni]{Andrea Aveni}
\address[A.~Aveni]{Department of Statistics, Duke University, 415 Chapel Dr, Durham 27705, USA}
\email{andrea.aveni@duke.edu}

\author[L.~Crippa]{Ludovico Crippa}
\address[L.~Crippa]{Stanford Graduate School of Business, Stanford University, 655 Knight Way, Stanford 94305, USA}
\email{lcrippa@stanford.edu}

\author[G.~Principi]{Giulio Principi}
\address[G.~Principi]{Department of Economics, New York University, 19 West 4th St, New York 10012, USA}
\email{gp2187@nyu.edu}

\date{\today}


\begin{abstract}
We study a family of distance functions
on rankings that allow for asymmetric treatments of alternatives and consider the distinct relevance of the top and bottom positions for ordered lists. We provide a full axiomatic characterization of our distance.
In doing so, we retrieve new characterizations of existing axioms and show how to effectively weaken them for our purposes. This analysis highlights the generality of our distance as it embeds many (semi)metrics previously proposed in the literature. Subsequently, we show that, notwithstanding its level of generality, our distance is still readily applicable. We apply it to preference aggregation, studying the features of the associated median voting rule. It is shown how the derived preference function  satisfies many desirable features in the context of voting rules, ranging from fairness to majority and Pareto-related properties. We show how to compute consensus rankings exactly, and provide generalized Diaconis-Graham inequalities that can be leveraged to obtain approximation algorithms. Finally, we propose some truncation ideas for our distances inspired by \cite{lu2010unavailable}. These can be leveraged to devise a Polynomial-Time-Approximation Scheme for the corresponding rank aggregation problem.

\end{abstract}
\maketitle
\thispagestyle{empty}

\section{Introduction}

The problem of aggregating preferences of heterogeneous individuals into a representative ranking, or rank aggregation in short, has sparked substantial interest across different research areas, such as Statistics, Economics, Computer Science, and Political Science. For instance, in Social Choice theory individuals are voters who cast their preferences over a pool of candidates: The goal of the aggregation is to identify \emph{consensus preferences}. Statisticians also study ranking data aggregation providing probabilistic models based on weighted distances over rankings, prominent examples are the Mallows model and its generalizations (\cite{mallows}, \cite{boutilier_stat}, \cite{weighted_distance_stat_software}). More recent instances of rank aggregation include collaborative filtering and recommender systems (e.g., \cite{pennock2000social}), which are pervasive in modern digital platforms. Movies and shopping recommendations to users are often expressed as an ordered list of suggested films and items, produced by combining movie and item ratings supplied by other users. Further, rank aggregation problems are associated with metasearch and spam detection engines. The objective here is to combine rankings of web pages or hyperlinked documents produced by different ranking algorithms to improve the users' searching experience (e.g., \cite{cohen1997learning}, \cite{dwork2001rank}). Voting and more generally rank aggregation mechanisms have been proposed as instrument for planning and coordination in the field of 
Multiagent Systems (e.g., \cite{ephrati1993multi}).


Several methods of aggregation rely on distance functions that evaluate the proximity between two rankings. These functions quantify the level of cohesion among a collection of rankings, such as an electorate. Specularly, distance-based methods also allow to evaluate the degree of dispersion of opinions of a group of individuals. One of the most popular distances over rankings is the Kendall distance, \cite{Kendall48}, which is defined as the minimum number of swaps of adjacent elements that transform one ranking into the other. 

The Kendall distance has two main shortcomings. On the one hand, it does not take into account the relative importance of positions where the swaps occur. Yet, in many applications swaps at the top positions carry more dissimilarity compared to swaps at lower positions. 
For instance, when comparing alternative webpages' rankings, one is often interested in the resulting traffic flow, as measured by the \emph{click-through rates} (CTR). However, several studies (e.g., \cite{agarwal2011location}, \cite{richardson2007predicting}, \cite{wang2019serial}) have indicated a significant disparity in the average CTR between the top-ranked result and the second-ranked result, which is notably larger than the corresponding difference among the CTRs of lower-ranked webpages. In addition, the presence of screening costs and time constraints naturally favors top-ranked results when processing long pieces of information. 

On the other hand, the Kendall distance also treats homogeneously all the alternatives over which rankings are cast. This hinders its applicability to  settings where the aggregator may have preferences over the available alternatives. For example, in the presence of sponsored links or in the presence of different displaying costs for the alternatives involved, the aggregator's preferences might play a crucial role in determining an aggregate ranking.

In this work, we introduce a class of distances, named weighted top-difference distances, that overcome these shortcomings while extending and connecting to several distances proposed in the literature. Our distance evaluates the proximity of two rankings in terms of three aspects: the difference between the maximal elements in each subset of alternatives, the size of such menu, and the relative importance of each alternative. We provide an axiomatic characterization of our class of distances. Such axiomatic foundation is built over a thorough study of betwenness axioms without any neutrality and invariance requirements. The results we provide not only lead to the full characterization of the proposed class of distances, but also highlight the implications of betweenness axioms taken singularly. The complete characterization of our distance is obtained by combining some betweenness and separability axioms, and it allows to have a better understanding of the existing axiomatizations of the Kendall distance and its weighted versions.

In addition to its generality, our distance is shown to be profitably applicable to rank aggregation problems. 
Given a set of alternatives and a profile of rankings over them, we apply our distance to study the properties of the induced median aggregation correspondence. Among the most significant properties, we highlight fairness conditions, such as neutrality and various specifications of Pareto dominance, and majority-related aspects, such as the majority and Condorcet principles.

From a computational standpoint our contribution is threefold. We show how to compute in polynomial time all the distances considered in this work and provide a linear integer program for the ranking aggregation problem with polynomially many variables and constraints. We provide a set of generalized Diaconis-Graham inequalities and a unified framework to analyze them, thus obtaining approximation algorithms in the spirit of \cite{diaconis1977spearman} and \cite{dwork2001rank}. We also propose a Polynomial-Time Approximation Scheme inspired by the truncation techniques presented in~\cite{lu2010unavailable}.

\subsubsection*{Organization of the paper} The next two sections are dedicated to the literature review and mathematical notation. In §\ref{section:distance} we present our distance, its axiomatic characterization and main properties. §\ref{section:voting} is devoted to median rank aggregation and its social choice theoretic properties. §\ref{section:computational} is dedicated to algorithmic and computational properties of both our distance and the induced voting rules presented in §\ref{section:voting}. In the Appendix we report the proofs of the results presented in the main text.

\section{Literature Review}\label{lit rev}

\subsubsection*{Rankings and metrics} As it concerns distance functions over rankings seminal works in the literature are \cite{Kendall48}, \cite{KemenySnell62}, \cite{Bogart73}, and \cite{Bogart75}. In these papers, the authors provide different formulations and characterizations of the Kendall metric. More recent advancements involve providing weighted versions of the Kendall metric as in \cite{Hassandez} and \cite{kumar2010generalized}. In particular, \cite{Hassandez} proposed a weighted version of the Kendall metric that overcomes one of its main shortcomings, that is its indifference to the varying relevance of different positions in the rankings. Starting from the same challenge, \cite{dissimiOkNishi} proposed a different solution, that they called the top-difference approach. Their distance evaluates the dissimilarity between any two rankings by considering the (possibly) distinct choices induced by the two rankings over all possible menus. This approach is not restricted to the case of linear orders, but extends to the larger domain of acyclic binary relations. In \cite{dissimiOkNishi} the authors proved that when restricted to linear orders, the neutral version of their distance is a particular case of the weighted Kendall metric of \cite{Hassandez}.

The distances we study embed both the Kendall metric and the dissimilarity semimetrics proposed by \cite{dissimiOkNishi} and \cite{Hassandez}. Weighted version of the Kendall metric are also advanced in \cite{kumar2010generalized} with a main focus on their computability. Differently from them, we provide an axiomatic foundation for the class of distances proposed in this work. As argued in §\ref{Connection}, our axiomatization yields additional insights on their distances as well.
A different approach was taken by \cite{KlamlerDistChoice} where the author proposed a distance directly on choice correspondences rather than rankings. As long as such choice correspondences are rationalizable by binary relations this approach is analogous to the one of \cite{dissimiOkNishi}.

\subsubsection*{Social Choice Theory}
The use of distance functions for rank aggregation is well-established in social choice theory. This approach was advanced by \cite{KemenySnell62}, and was later adopted by \cite{Blin}, \cite{CookSeiford78}, and \cite{YoungLev78}. In particular, \cite{YoungLev78} provided a full characterization of the median preference function induced by the Kendall metric (Kemeny voting rule). They show that such preference function is the only one satisfying three properties: reinforcing, neutrality, and the Condorcet principle. In §\ref{section:voting} we strengthen some of the results of \cite{YoungLev78} showing under which conditions the median preference function induced by our metric satisfies reinforcing, neutrality, and the Condorcet principle. In more recent papers \cite{Hassandez} and \cite{GILBERT202227}, the authors provided some results on median rank aggregation showing that some properties of the median approach proposed by \cite{YoungLev78} extend to more general distances. The use of distance functions for the rank aggregation problem is also analyzed in \cite{baldiga2013assent} and \cite{lu2010unavailable}, respectively \emph{assent-maximizing social welfare functions} and the \emph{unavailable candidate model}. Both models are encompassed by the class of distances advanced in this work. For a survey on distance-based methods see \cite{comsoc} and \cite{Cook2006}.
Many voting rules (e.g., Borda, Dodgson, Plurality, and Kemeny) can be retrieved through the concept of \emph{distance rationalization} (see \cite{lerer1985some}, \cite{campbell1986social}, \cite{meskanen2008closeness}, and \cite{DisRatVR}) which is based on a consensus class (\cite{nitzan1981some} and \cite{baigent1987metric}), and a distance function over consensus elections. Specifically, in this framework, a candidate is deemed a winner if it is ranked first in one of the nearest (with respect to the given distance) consensus elections within the consensus class. 

A closely related line of research focuses on the \emph{Maximum Likelihood Estimator} (MLE) approach (\cite{de2014essai} and \cite{young_1988_book}). The underlying assumptions concern the existence of a correct ranking of the alternatives (\emph{ground truth}) of which we only have noisy realizations: the goal is to compute the MLE based on the available observations. Tight connections between the rank aggregation problem and the MLE approach have been discovered by \cite{young_1988_book}, and, more recently, strengthened by \cite{conitzer2012common} and \cite{conitzer2009preference}. The design of voting rules that are MLEs for some noise model has been an active area of research, see for example, \cite{elkind2010good}, \cite{procaccia2012maximum}, \cite{mao2013better}, and \cite{caragiannis2014modal}.  For additional connections among the rationalizability framework, the MLE approach, and the rank aggregation problem, we refer the reader to Chapter $8$ of \cite{comsoc} and references therein.

\subsubsection*{Computability}
It is known that computing Kemeny rankings is NP-hard \cite{bartholdi1989voting}, even with only four voters \cite{dwork2001rank}. Bounds used in search-based methods to compute \emph{optimal} Kemeny rankings have been proposed by \cite{davenport2004computational} and \cite{conitzer2006improved}. The computational complexity of the Kemeny rank aggregation problem led to the subsequent study and development of approximation algorithms. We summarize some of them in the coming lines while deferring the reader to \cite{comsoc} for a more detailed description.
 Randomly selecting a voter's preference provides a $2$-approximation in expectation, and this procedure can be de-randomized by preserving the same approximation factor \cite{ailon2008aggregating}. Sorting the alternatives in ascending order of Borda score results in a $5$-approximation to the Kemeny rank aggregation problem \cite{coppersmith2006ordering}. Another polynomial-time $2$-approximation algorithm can be obtained by replacing the Kemeny distance with the Spearman’s footrule distance, see the seminal works by \cite{diaconis1977spearman} and \cite{dwork2001rank}. Improved approximation factors for the Kemeny rank aggregation problem are obtained in \cite{dwork2001rank} and \cite{van2009deterministic}. A Polynomial-Time-Approximation Scheme appeared in \cite{kenyon2007rank}.

 Additional evidence of the hardness concerning distance-based rank aggregation problems is provided by \cite{GILBERT202227} when analyzing setwise versions of the Kendall distance. Since the Kendall and its setwise versions are a special case of the weighted top-difference distances presented in this paper, finding an aggregate ranking within the class of our distances is an NP-hard problem. 
 
\section{Mathematical preliminaries}\label{section:math_prelim}

We report here the basic notation we employ in the next sections. We consider a set $C$ of~$n \geq 2$ elements which can be interpreted as alternatives or candidates. Our investigation revolves around rankings over~$C$ represented as lists that arrange candidates based on preference, with the foremost candidate being the most favored and the ultimate candidate the least favored. We arbitrarily label the candidates with the numbers $\left[n\right]:=\left\lbrace 1,\ldots,n\right\rbrace$, so that a ranking is represented by a \textit{permutation}; we use the terms permutations and rankings interchangeably. By a permutation we mean a bijective function $\sigma:\left[n\right]\to \left[n\right]$, with $\sigma_i$ denoting the candidate in the $i^{\mathrm{th}}$-position according to $\sigma$. Therefore, the set of all rankings over $[n]$, is the symmetric group denoted by $\LL_n$. We, sometimes, write permutations as ordered lists e.g., $\sigma=(\sigma_1,\ldots,\sigma_n)$. For all permutations $\sigma\in\LL_n$ and alternatives~$i,j\in[n]$, we write~$i>_\sigma j$ if $\sigma_i^{-1}<\sigma_j^{-1}$ with $<_{\sigma}$ bearing the analogous meaning. The identity permutation~$(1,\ldots,n)$ will be denoted by $\Id$. Often we employ the group structure of $\LL_n$, and we adopt the convention that for two permutations~$\sigma,\pi\in \LL_n$, $\sigma\pi:i\mapsto \sigma_{\pi_{i}}$. For all $\sigma\in \LL_n$, we denote by~$\sigma^{-1}$ the inverse permutation according to this product. Particularly relevant for our work are \textit{transpositions}. For distinct $i,j\in \left[n\right]$, we denote by $t_{i,j}$ the $i,j$-transposition, that is the permutation that agrees with the identity $\Id$, but swaps~$i$ and $j$, formally $t_{i,j}=(1,\ldots,i-1,j,i+1,\ldots,j-1,i,j+1,\ldots,n)$. Whenever $|i-j|=1$, we refer to~$t_{i,j}$ as \textit{adjacent}. It is noteworthy that, in general, for $\sigma\in \LL_n$, we have $\sigma t_{i,j}\neq t_{i,j}\sigma$. In particular,~$\sigma t_{i,j}$ is the permutation obtained from $\sigma$ by exchanging $\sigma_i$ with~$\sigma_j$, while $t_{i,j}\sigma$ is obtained from $\sigma$ by exchanging $i$ with $j$. Also we observe that for any distinct $i,j\in[n]$ and~$\sigma\in\LL_n$, we have that~$\sigma t_{i,j}=t_{\sigma_i,\sigma_j}\sigma$ and analogously $t_{i,j}\sigma=\sigma t_{\sigma^{-1}_i,\sigma^{-1}_j}$.\\
To distinguish products (of permutations) on the left from products on the right we use the notation~$\coprod_{k=1}^n\sigma_{k}:=\prod_{k=1}^n\sigma_{n+1-k}$.\\
Given any two permutations $\sigma,\pi\in\LL_n$, we define their \textit{inversion set} as the collection of all couples~$(i,j)$ on which $\sigma$ and $\pi$ disagree, that is
$$I(\sigma,\pi):=\{(i,j)\in[n]^2:i>_\sigma j\textrm{ and }i<_\pi j\}.$$
The \textit{Kendall distance} $\dd_K$ between $\sigma$ and $\pi$ is defined as the cardinality of $I(\sigma,\pi)$. Given a ranking~$\sigma\in\LL_n$ and a candidate $j\in[n]$, we denote by $j^{\uparrow,\sigma}$ the collection of alternatives preferred to $j$ by $\sigma$, and~$j^{\downarrow,\sigma}$ is defined analogously. More precisely,
$$j^{\uparrow,\sigma}:=\left\{i\in[n]:i>_\sigma j\right\},\ \ \ \ \ \ \ \ \ j^{\downarrow,\sigma}:=\left\{i\in[n]:i<_\sigma j\right\}.$$
\noindent To ease the notation we denote $j^{\uparrow}:=j^{\uparrow,\Id}$ and $j^{\downarrow}:=j^{\downarrow,\Id}$.
Given a subset of alternatives $S\subseteq C$ and a ranking $\sigma\in\LL_n$, we denote by
$M(S,\sigma)$ and by $m(S,\sigma)$ the maximum and minimum elements in $S$ with respect to $\sigma$. Finally, for all $\sigma,\pi\in\LL_n$, $$\bigtriangleup_S(\sigma,\pi):=M(S,\sigma)\bigtriangleup M(S,\pi)$$
where $\bigtriangleup$ denotes the symmetric difference of sets. We denote by $\mathcal{P}_n$ the collection of subsets of~$\left[n\right]$ with cardinality larger or equal than $2$.

We denote $\R_+:=[0,\infty)$ and $\R_{++}:=(0,\infty)$. When considering (signed) measures $\mu:2^{[n]}\to\R$ on the finite set $[n]$, to ease notation we may write $\mu(j)$ or $\mu_j$ instead of $\mu(\{j\})$. As we work in a finite setting, with some abuse of notation, we will often identify measures as elements of $\R^n$. The counting measure will be denoted by $\lvert\cdot\rvert$.
We denote by~$\mathcal{M}^{+}_n$ the set of non-negative measures over $2^{[n]}$, this set is identified with $\R^n_+$. The set of fully supported measures in $\mathcal{M}^+_n$ is denoted by~$\mathcal{M}_n^{++}$, identified with $\R_{++}^n$. Moreover, we define also the following sets:
\[
\mathcal{M}_n^{\geq}:=\left\lbrace \mu\in \R^n:\forall i\neq j,\ \mu_i+\mu_j\geq 0 \right\rbrace\ \textnormal{and}\ \mathcal{M}_n^{>}:=\left\lbrace \mu\in \R^n:\forall i\neq j,\ \mu_i+\mu_j> 0 \right\rbrace.
\]

\section{Distances on rankings}\label{section:distance}


\subsection{Definition}\label{def}
 Fix a measure $\mu:2^{\C}\to \R$ and a vector $\bb=\left(\beta_i\right)_{i=2}^{n}\in \mathbb{R}^{n-1}$. Given a function~$\dd:\LL_n\tim \LL_n\to \R$, we call it $\left(\bb,\mu \right)$-\textit{top difference distance} if
\begin{equation}\label{def:dbetamu}
\dd\left(\sigma,\pi\right)= \sum_{S\in \mathcal{P}_n}\beta_{| S|}\mu\left(\bigtriangleup_S\left(\sigma,\pi\right)\right).
\end{equation}
In this case $\dd$ is denoted by $\ddd$. Whenever, $\mu_i,\beta_i\geq 0$ for all $i$ it is readily observed that $\ddd$ is always a semimetric. If in addition~$\beta_2>0$ and $\mu\in\mathcal M^{++}_n$, then $\ddd$ is also a metric (See Lemma~\ref{Ehvoleeevi} and Corollary \ref{metric}). 
The subsets~$S\subseteq [n]$ can be interpreted as menus of alternatives. For each menu we measure through~$\mu$ the extent of disagreement between $\sigma,\pi$ by comparing their maximal elements in that menu. This is weighted according to the size $|S|$ of the menu. The choices of $\bb$ and $\mu$ are important to determine different properties of $\ddd$. Of particular interest are the following specifications of $\bb$ and $\mu$. If~$\mu$ is the counting measure, we denote $\ddd$ by $\dd_\bb$. This corresponds to the \textit{neutral} version of $\ddd$ where the candidates are treated symmetrically. When the relative weight of the size of the menus of candidates is constant, that is~$\bb=(1,1,\ldots,1)$ we denote $\ddd$ by $\dd^\mu$. When $\mu\in \mathcal{M}^+_n$, this corresponds to the semimetric proposed by \cite{dissimiOkNishi}. Different specifications of $\bb$ allow to restrict distance evaluations to binary comparisons, this occurs when $\bb=(1,0,\ldots,0)$. In this case, we denote the resulting semimetric by~$\dd_{K}^\mu$ , while if furthermore~$\mu=| \cdot |$, then it reduces to $\dd_K$, the Kendall distance. 

\subsection{Parameter space}\label{sect:semimetricdiscu}
In this section we answer to two basic questions. For which values of $\bb$ and $\mu$ is $\ddd$ a semimetric, and which additional restrictions guarantee that it is also a metric?
To this end, first define the following sets
\begin{gather*}
B^{\geq}_n:=\{\bb\in\R^{n-1}:\dd_\bb\textrm{ is a semimetric}\},\ B^{>}_n:=\{\bb\in\R^{n-1}:\dd_\bb\textrm{ is a metric}\},\\
R^+_n:=\R_+\times \{0\}^{n-2},\  R^{++}_n:=\R_{++}\times \{0\}^{n-2},\ R^-_n:=-R^+_n,\ \textnormal{and}\ R^{--}_n:=-R^{++}_n.
\end{gather*}
Notice that $\R_+^{n-1}\subseteq B^{\geq}_n$ and $\R_{++}\tim\R_{+}^{n-2}\subseteq B^{>}_n$ (See Lemma~\ref{Ehvoleeevi} and Corollary \ref{metric}). If $\mu=\mathbf{0}$ or~$\bb=\mathbf{0}$, then $\ddd$ is the \textit{trivial} semimetric constantly equal to $0$. The following proposition answers the questions above.
\begin{proposition}\label{prop:semimetric} If $n=2$, then
\begin{align*}
\{(\bb,\mu)\in\R^{3}:\ddd \textnormal{ is a semimetric}\}&=R^{++}_2\tim \mathcal{M}^{\geq}_2\sqcup R^{--}_2\tim (-\mathcal{M}^{\geq}_2)\sqcup \{0\}\tim \R^2,\\
\{(\bb,\mu)\in\R^{3}:\ddd \textnormal{ is a metric}\}&=R^{++}_2\tim \mathcal{M}^{>}_2\sqcup R^{--}_2\tim (-\mathcal{M}^{>}_2).
\end{align*}
If $n\geq 3$, then
    \begin{align*}
        \{(\bb,\mu)\in\R^{2n-1}:\ddd \textnormal{ is a semimetric}\}=&B_n^\geq\tim\mathcal M_n^+\cup  R^+_n\tim\mathcal M^\geq_n\\
        &\cup (-B_n^\geq)\tim(-\mathcal M_n^+)\cup R^-_n\tim(-\mathcal M^\geq_n)\\
        &\cup \{\mathbf{0}\}\tim \R^n \cup \R^{n-1}\tim\{\mathbf{0}\}
    \end{align*}
    and
    \begin{align*}\{(\bb,\mu)\in\R^{2n-1}:\ddd \textnormal{ is a metric}\}=&B^{>}_n\tim\mathcal M^{++}_n\cup R^{++}_n\tim\mathcal M^{>}_n\\
    &\cup(-B^{>}_n)\tim(-\mathcal M^{++}_n)\cup R^{--}_n\tim(-\mathcal M^{>}_n).
    \end{align*}
\end{proposition}
\noindent The proof of this result is based on the restrictions imposed by two of the defining properties of semimetrics. In particular, apart for the trivial cases, the non-negativity of the semimetric imposes that $\mu$ belongs to $\mathcal{M}^{\geq}_n$. Furthermore, if $\bb\notin \R_+\tim \{0\}^{n-2}$, the triangular inequality leads to the non-negativity of $\mu$, that is $\mu\in \mathcal{M}^+_n$.
Since $\dd^{-\mu}_{-\bb}=\ddd$, it is possible to cover all the possible semimetrics by restricting to $B_n^\geq\tim\mathcal M_n^\geq$.

\subsection{Properties of the weighted top-difference distance}\label{properties of the metric}
\subsubsection*{Uniqueness and neutrality}\label{sect_uniqueness} The first properties we discuss concern the uniqueness of the weights and measure of a given $\left(\bb,\mu \right)$-top difference distance.
\begin{proposition}\label{hoilmercedes}
For any $n \geq 3$ and $\mu\in \R^n\setminus\{\mathbf0\}$, we have that $\ddd=\dd^\mu_\ga$ if and only if $\bb=\ga$.
\end{proposition}
Notice that, if $n=2$ and $\mu([2])\neq 0$, then the claim of Proposition \ref{hoilmercedes} holds.
Under stronger conditions a specular argument holds for the uniqueness of the measure.
\begin{proposition}\label{bigblackcockcomepiaceadandreino}
If $n \geq 3$, for all $\bb\in B^>_n\sqcup (-B^>_n)$, $\ddd=\dd^{\nu}_\bb$ if and only if $\mu=\nu$.
\end{proposition}
However, it is possible to have $\ddd=\dd^\nu_\ga$ for distinct $\bb,\ga$ and $\mu,\nu$. Indeed, for any $\alpha\neq0$, if~$\nu=\alpha \mu$ and $\ga=\bb/\alpha$, then $\ddd=\dd_\ga^\nu$.

We now introduce an invariance property called \textit{neutrality}. A distance function is neutral if it treats all alternatives equally meaning that it is invariant under relabeling the candidates. Formally, a function $\dd:\LL_n\times \LL_n\to\R$ is \textit{neutral} if for every permutations $\sigma,\pi,\tau\in \LL_n$,
$$\dd(\tau\sigma,\tau\pi)=\dd(\sigma,\pi).$$  
\noindent 
This means that $\dd$ is invariant under the canonical left action of $\LL_n$ on itself. In particular, we have that multiplying a permutation on the left by $\tau$ is equivalent to relabelling the alternatives from $(1,\ldots,n)$ to $(\tau_1,\ldots,\tau_n)$.
Since the measure $\mu$ in $\ddd$ captures the relative importance of the alternatives, it is natural to expect that $\ddd$ is neutral if and only if $\mu$ gives the same importance to all options, that is when $\mu$ is proportional to the counting measure. Formally,
\begin{proposition}\label{Galletta}
  Let $n \geq 3$ and $\bb\in B^>_n\sqcup (-B^>_n)$. Then, $\ddd$ is neutral if and only if $\mu$ is constant.
\end{proposition}

\subsubsection*{Computability}\label{Formula Maggica subsection}

At first glance, the computation of the distance $\ddd$ between any two linear orders~$\sigma$ and $\pi$ seems to be highly inefficient. The main reason being that it requires to compute the maximal elements according to both $\sigma$ and $\pi$ for $O(2^n)$ subsets. Following the methods developed by \cite{dissimiOkNishi}, we now provide a convenient reformulation of $\ddd$ that allows to compute the distance between two permutations in polynomial time. Furthermore, as it turns out in our proofs, the formulas provided in this section will play a crucial role also beyond its computational convenience.

\par\medskip
Fix a measure $\mu$, weights $\bb\in \R^{n-1}$, and $\sigma,\pi\in \LL_n$. Notice that
\begin{align*}
  \ddd(\sigma,\pi)&=\sum_{S\in \mathcal{P}_n}\beta_{|S|}\mu\left(\bigtriangleup_S(\sigma,\pi)\right)\\
  &=\sum_{S\in \mathcal{P}_n}\sum_{x\in S}\beta_{|S|}\mu\left( x \right)\1_{\bigtriangleup_S(\sigma,\pi)}(x)\\
  &=\sum_{x\in [n]}\left(\sum_{S\in \mathcal{P}_n:x\in S}\beta_{|S|}\1_{\bigtriangleup_S(\sigma,\pi)}(x)\right)\mu\left( x \right)\\
  &=\sum_{x\in [n]}\left(\sum_{k=2}^n\beta_{k}\sum_{S:x\in S\subseteq [n], |S|=k}\1_{\bigtriangleup_S(\sigma,\pi)}(x)\right)\mu\left( x \right)\\
  &=\sum_{x\in [n]}\left(\sum_{k=2}^n\beta_{k}\theta_{x,k}(\sigma,\pi)\right)\mu\left( x\right)
\end{align*}
where $\theta_{x,k}(\sigma,\pi)$ is the number of all subsets $S$ of cardinality $k$ such that $x\in \bigtriangleup_S(\sigma,\pi)$. We now \textit{compute} $\theta_{x,k}(\sigma,\pi)$ for each $x$ and $k$. To this end fix $x\in [n]$ and define the following
\begin{align*}
  A_{x}(\sigma,\pi)&:=x^{\downarrow,\sigma}\cap x^{\downarrow,\pi},&&&
  B_{x}(\sigma,\pi)&:=x^{\uparrow,\sigma}\cap x^{\downarrow,\pi},&&& 
  C_{x}(\sigma,\pi)&:=x^{\downarrow,\sigma}\cap x^{\uparrow,\pi}\\
  a_{x}(\sigma,\pi)&:=| A_{x}(\sigma,\pi)|,&&&
  b_{x}(\sigma,\pi)& :=| B_{x}(\sigma,\pi)|,&&&
  c_{x}(\sigma,\pi)&:=| C_{x}(\sigma,\pi)|.
\end{align*}
Let $S$ be such that $|S|=k$, then $x\in M(S,\sigma)\setminus M(S,\pi)$ if and only if $S=\{x\}\sqcup K\sqcup L$ for some non empty $L\subseteq C_x(\sigma,\pi)$ of cardinality $j\in\{1,\ldots,k-1\}$ and some $K\subseteq A_x(\sigma,\pi)$ with cardinality~$k-1-j$. The number of sets $S$ of cardinality $k$ such that $x\in M(S,\sigma)\setminus M(S,\pi)$ is given by
\begin{align*}
\sum_{j=1}^{k-1}\binom{c_x}{j}\binom{a_x}{k-1-j}&=\sum_{j=0}^{k-1}\binom{c_x}{j}\binom{a_x}{k-1-j}-\binom{c_{x}}{0}\binom{a_{x}}{k-1}\\
&=\binom{a_x+c_x}{k-1}-\binom{a_{x}}{k-1}
\end{align*}
where the last equality follows from the Chu-Vandermonde identity.\footnote{We recall the Chu-Vandermonde identity:
\[
\binom{m+\ell}{r}=\sum_{j=0}^r\binom{m}{j}\binom{\ell}{r-j}.
\]} Similarly, the number of sets~$S$ of cardinality $k$ such that $x\in M(S,\pi)\setminus M(S,\sigma)$ is
$$\binom{a_x+b_x}{k-1}-\binom{a_{x}}{k-1}.$$
Thus,
$$\theta_{x,k}(\sigma,\pi)=\binom{a_x+b_x}{k-1}+\binom{a_x+c_x}{k-1}-2\binom{a_{x}}{k-1}.$$
Now we define
\[
f_\bb:t\mapsto \sum_{k=2}^n\beta_k\binom{t}{k-1}.
\]
Then, we have that
\begin{align*}\sum_{k=2}^n\beta_k\theta_{x,k}(\sigma,\pi)=f_\bb(a_x(\sigma,\pi)+b_x(\sigma,\pi))+f_\bb(a_x(\sigma,\pi)+c_x(\sigma,\pi))-2f_\bb(a_x(\sigma,\pi)).
\end{align*}
Next, notice that $A_x(\sigma,\pi)\sqcup B_x(\sigma,\pi)=x^{\downarrow,\pi}$; whence $a_x(\sigma,\pi)+b_x(\sigma,\pi)=|x^{\uparrow,\pi}|$.
Similarly, $a_x(\sigma,\pi)+c_x(\sigma,\pi)=| x^{\downarrow,\sigma}|$. Combining all these computations yield
\begin{equation}\label{DioMerda}
\ddd(\sigma,\pi)=\sum_{x\in [n]}\left(f_\bb(| x^{\downarrow,\pi}|)+f_\bb(| x^{\downarrow,\sigma}|)-2f_\bb(| x^{\downarrow,\sigma}\cap x^{\downarrow,\pi} |)\right)\mu(x).
\end{equation}

In particular, by using the Pascal's identity, the above formula implies that $\ddd$ can be computed in time $O(n^2)$. Furthermore, the simplicity of the presented formulas will come in handy in the social choice application and in the algorithmic part of this work.

\subsection{Axiomatic characterization}\label{Sect_Main_Axiomatiz}
In this section we provide an axiomatic foundation for the class of weighted top-difference distances presented above. As a pivotal step to our analysis, we first \emph{isolate} the implications of some \emph{betweenness} axioms, previously studied in the literature \emph{only in conjunction} with invariance properties.

It is well known that for any two permutations $\sigma,\pi\in\mathbb S_n$, their Kendall distance $\dd_K(\sigma,\pi)$ is equal to the minimum number of adjacent transpositions that transform $\sigma$ into $\pi$. Therefore, the set of all \textit{paths} of adjacent transpositions that transform $\sigma$ into $\pi$ in $\dd_K(\sigma,\pi)$-steps is not empty. Formally, the following is nonempty for any $\sigma,\pi\in\mathbb S_n$,
$$T(\sigma,\pi):=\left\{\left(t_{a_k,a_k+1}\right)_{k=1}^{\dd_K(\sigma,\pi)}:\pi=\sigma \prod_{k=1}^{\dd_K(\sigma,\pi)}t_{a_k,a_k+1}\right\}.$$

Given an element $\mathbf t\in T(\sigma,\pi)$ and a $k\in[\dd_K(\sigma,\pi)]$, we denote by $a_k(\mathbf t)$ the smallest index of the $k^{\textrm{th}}$ transposition. Moreover, we denote by $i_k(\mathbf t)<j_k(\mathbf t)$ the unique elements in $[n]$ that are exchanged by the $k^{\textrm{th}}$ transposition. 
It follows that if $\mathbf t\in T(\sigma,\pi)$, then,
$$\pi=\sigma \prod_{k=1}^{\dd_K(\sigma,\pi)}t_{a_k(\mathbf t),a_k(\mathbf t)+1}=\left(\coprod_{k=1}^{\dd_K(\sigma,\pi)}t_{i_{k}(\mathbf t),j_{k}(\mathbf t)}\right)\sigma.$$

To ease the notation, when $\mathbf t$ is clear from the context we write $a_k,i_k,j_k$ omitting their dependence from $\mathbf t$. Before providing our axioms we need to also define the employed betweenness condition. Given distinct $\sigma,\omega,\pi\in \mathbb{S}_n$, we say that $\omega$ is \textit{in between} $\sigma$ and $\pi$ if, over all menus of size two, when $\sigma$ and $\pi$ agree, $\omega$ agrees with them; formally,
$$\forall i \neq j\in \left[n\right],\,\,\,\, M(\{i,j\},\sigma)=M(\{i,j\},\pi) \implies M(\{i,j\},\omega)=M(\{i,j\},\sigma).$$ 
If this happens we write $\sigma - \omega - \pi$. The rankings $\pi^1,\pi^2,\ldots,\pi^s$ are said to be \textit{on a line}, denoted by~$\pi^1-\pi^2-\cdots-\pi^s$, if for every $i,j$, and $k$ for which $1\leq i< j< k \leq s$, we have that $\pi^i-\pi^j-\pi^k$.

\noindent We now introduce the axioms that we will employ in the subsequent characterizations. We provide an interpretation of all these axioms immediately before each result in which they are employed.
\begin{axioms}\color{white}Dio Merda\color{black}\\\vspace{-0.5cm}
\begin{enumerate}[label=\textnormal{\textbf{A}.\arabic*}]
\item \label{A1} For all $\sigma,\omega,\pi\in \mathbb{S}_n$, if $\sigma-\omega-\pi$, then $\dd(\pi,\sigma) = \dd(\pi,\omega) + \dd(\omega,\sigma)$.
\item \label{A2} For all distinct $i,j,k,\ell\in[n]$, $\dd(\Id,t_{i,j})+\dd(\Id,t_{k,\ell})=\dd(\Id,t_{i,k})+\dd(\Id,t_{j,\ell})$.
\item \label{A3} For all $\sigma,\pi\in \mathbb{S}_n$ disagreeing on more than one pair of elements there exists $\omega\in \mathbb{S}_n$, distinct from $\sigma$ and $\pi$ such that $\sigma-\omega-\pi$ and $\dd(\sigma,\pi) = \dd(\sigma,\omega) + \dd(\omega,\pi)$.
\item \label{A4} For all distinct $\sigma,\pi\in\LL_n$ and $a\in[n-1]$, $$I(\sigma,\sigma t_{a,a+1})=I(\pi,\pi t_{a,a+1})\ \Longrightarrow\ \dd(\sigma,\sigma t_{a,a+1})=\dd(\pi,\pi t_{a,a+1}).$$
\item \label{A2ter} For all $\sigma,\pi,\rho,\tau\in\LL_n$ and distinct $a,b\in[n-1]$ with $$I(\sigma,\sigma t_{a,a+1})=I(\pi,\pi t_{b,b+1})\textrm{ and }I(\rho,\rho t_{a,a+1})=I(\tau,\tau t_{b,b+1})$$
we have that
\[
\dd(\sigma,\sigma t_{a,a+1})\dd(\tau,\tau t_{b,b+1})=\dd(\pi,\pi t_{b,b+1})\dd(\rho,\rho t_{a,a+1}).
\]
\item \label{A2'} For all $\sigma,\pi,\rho,\tau\in \LL_n$ and $a\in [n-1]$ with $$I(\sigma,\sigma t_{a,a+1})\sqcup I(\pi,\pi t_{a,a+1})=I(\rho,\rho t_{a,a+1})\sqcup I(\tau,\tau t_{a,a+1}),$$
we have that
$$\dd(\sigma,\sigma t_{a,a+1})+\dd(\pi,\pi t_{a,a+1})=\dd(\rho,\rho t_{a,a+1})+\dd(\tau,\tau t_{a,a+1}).$$
\end{enumerate}
\end{axioms}
Axiom \ref{A1} axiom requires the triangle inequality to hold as an equality for \emph{all} permutations that lie on a line. We now completely characterize the distances satisfying this axiom. 

\begin{theorem}\label{ronniecolemanleggero}
  A semimetric $\dd$ satisfies axiom \ref{A1} if and only if there exists a symmetric function $g:[n]^2\to\R_{+}$ with $g(i,i)=0$ for all $i\in [n]$ and such that for any $\mathbf{t}\in T(\sigma,\pi)$,
  \[
  \dd(\sigma,\pi)= \sum_{k=1}^{\dd_K(\sigma,\pi)}{}g(i_k(\mathbf{t}),j_k(\mathbf{t}))=\sum_{(i,j)\in I(\sigma,\pi)}g(i,j).
  \]
\end{theorem}
\noindent Thus, Axiom \ref{A1} guarantees a simple expression for the distance function, only depending on the $\binom n2$ numbers $g(i,j)$ for $1\leq i<j\leq n$. 
\begin{remark}\label{remarkmetrico}
Some scholars considered the following stronger version of Axiom \ref{A1}: For all~$\sigma,\omega,\pi\in \mathbb{S}_n$, $\dd(\pi,\sigma) = \dd(\pi,\omega) + \dd(\omega,\sigma)$ if
and only if $\sigma-\omega-\pi$. Using this stronger version in Theorem~\ref{ronniecolemanleggero} would require to work with a metric instead of a semimetric. The statement of the theorem would be the same apart from $\dd$ being required to be a metric and $g$ to take strictly positive values outside of the diagonal.
\end{remark}
\noindent In order to achieve additive separability of the function $g$, we apply Axiom \ref{A2}. The resulting semimetric turns out to be the \textit{binary} version of $\dd^{\mu}_{\bb}$, i.e., $\dd^{\mu}_{K}$, as shown below.
\begin{theorem}\label{roar}
  A semimetric $\dd$ satisfies axioms \ref{A1} and \ref{A2} if and only if there exists $\mu\in \mathcal{M}^{\geq}_n$ such that for all $\sigma,\pi\in\LL_n$, and $\mathbf t\in T(\sigma,\pi)$,
  \[
\dd(\sigma,\pi)=\sum_{k=1}^{\dd_K(\sigma,\pi)}{}\mu(i_k(\mathbf t))+\mu(j_k(\mathbf t))=\sum_{(i,j)\in I(\sigma,\pi)}\mu(i)+\mu(j)=\sum_{S\in \mathcal{P}_n:|S|=2}\mu \left( \triangle_{S}(\sigma,\pi) \right).
  \]
\end{theorem}
Therefore, in order to characterize $\ddd$ without the restriction to binary comparisons, we need to weaken Axiom \ref{A1}. Such weakening is shown to be provided by Axiom \ref{A3} which requires the triangular inequality to hold as equality only for \emph{some} permutation between $\sigma$ and $\pi$. First, it is possible see (Proposition \ref{pazzia} in the Appendix) that \ref{A3} is equivalent to a metric $\dd$ being a graphic distance on the permutahedron, that is the polytope with vertices the permutations in~$\LL_n$ and as edges the couples of permutations with Kendall distance equal to one.\footnote{The reader can consult Figure \ref{fecipermutaedriche} right before the Appendix.} In \cite{Hassandez}, the authors prove that by joining Axiom \ref{A3} with neutrality (which they refer to as \textit{left-invariance}), it is possible to characterize the class of weighted Kendall metrics.\footnote{We notice that Axiom \ref{A3}, \textit{left-invariance} and \textit{right-invariance}, together force the metric to be a multiple of the $\dd_{K}$. Indeed all the metrics satisfying just left and right invariance must satisfy $\dd(\sigma,\pi)=\dd(\tau\sigma\tau^{-1},\tau\pi\tau^{-1})$.}
Axiom \ref{A4} requires the length of each edge of the permutahedron to depend exclusively on the elements that are transposed and by the positions at which they occur. By combining them, we obtain

\begin{theorem}\label{main_g_thm}
  A semimetric $\dd$ satisfies Axioms \ref{A3} and \ref{A4} if and only if there exists $g:[n-1]\tim [n]^2\to\R_+$ symmetric in the last two entries with $g_{\cdot}(i,i)=0$ for all $i\in [n]$ such that for all~$\sigma,\pi\in \LL_n$,
  \begin{equation}\label{LuanaCerini2}
  \dd(\sigma,\pi)=\min_{\mathbf t\in T(\sigma,\pi)}\sum_{k=1}^{\dd_K(\sigma,\pi)}g_{a_k(\mathbf t)}(i_k(\mathbf t),j_k(\mathbf t)).
  \end{equation}
\end{theorem}

\noindent Now assuming Axioms \ref{A2ter} and \ref{A2'}, that impose proportionality and additive separability conditions, we are able to retrieve $\ddd$.
\begin{theorem}\label{axiomatic_characterization_THM}
For a non-trivial semimetric $\dd$, the following are equivalent:
\begin{enumerate}[label=(\roman*)]
\item $\dd$ satisfies Axioms \ref{A3}, \ref{A2ter}, and \ref{A2'}.
\item There are $\pp=\left(\phi_a\right)_{a=1}^{n-1}\in \mathbb{R}^{n-1}_{+}$ and $(\bb,\mu)\in B_n^{\geq}\tim\mathcal{M}^{+}_n\cup R^n_+\tim \mathcal{M}^{\geq}_n$ such that for all $\sigma,\pi\in \LL_n$,
   \begin{equation}\label{LuanaCerini_mainthm}
   \dd(\sigma,\pi)=\min_{\mathbf t\in T(\sigma,\pi)}\sum_{k=1}^{\dd_K(\sigma,\pi)}\phi_{a_k(\mathbf t)}\left(\mu(i_k(\mathbf t))+\mu(j_k(\mathbf t))\right)=\sum_{S \in \mathcal{P}_n} \beta_{|S|}\mu\left(\triangle_{S}(\sigma,\pi)\right).
   \end{equation}
   Moreover, $\mu$ can be chosen so that $\sum_i\mu(i)=1$.
\end{enumerate}
In addition, given any measure $\mu$ there exists a bijection $F$ mapping $\bb$ into $\pp$. 
\end{theorem}

\subsubsection*{Relation between $\bb$ and $\pp$} We now explore in more detail the relationship between the weights~$\bb$ and $\pp$, when $\ddd$ is assumed to be a non-trivial semimetric.  Let~$F:\bb\mapsto\pp$ be defined as
  $$F:\bb\mapsto \left(\sum_{k=2}^n\beta_k\binom{n-a-1}{k-2}\right)_{a=1}^{n-1}.$$
  As we show in the proof of Theorem \ref{scorpions}, $F$ is the bijective linear map associating each $\bb$ to the weights $\pp$ as in \eqref{popo}. Recall that if $\mu\in\mathcal M_n^{++}$, then for all $\bb\in B^>_n$, $\ddd$ is a metric. Moreover, define~$\Phi_n:=F(B^>_n)$. It is easy to verify\footnote{If $\boldsymbol\phi\in\R_{++}^{n-1}$, then for any $a\in[n-1]$ and distinct $i,j\in[n]$, $(\mu(i)+\mu(j))\phi_a>0$ and thus, by equation \eqref{popo}, $\dd_{F^{-1}(\boldsymbol\phi)}^\mu(\sigma,\tau)>0$ for any distinct $\sigma,\tau$. Conversely, if there is some $\phi_a\leq0$, then $\dd_{F^{-1}(\boldsymbol\phi)}^\mu(\sigma,\sigma t_{a,a+1})=(\mu(\sigma_a)+\mu(\sigma_{a+1}))\phi_a\leq0$.} that $\Phi=\R_{++}^{n-1}$.
  An interesting subset of~$B^>_n=F^{-1}(\R_{++}^{n-1})$ is provided by
  $$B_n^*:=\R_{++}\tim\R_+^{n-2}\subseteq B^>_n.$$ 
  Indeed, negative values of $\beta_j$ would be difficult to interpret as the importance of disagreements in menus of size $j$. Moreover, $\Phi_n^*:=F(B_n^*)\subseteq \Phi_n$ corresponds to the set of positive and totally monotone sequences. In particular, if $\bb\in B_n^*$, then $F(\bb)$ is non-increasing, meaning that swaps in the top positions are more significant than swaps in the bottom ones. Formally,
  \begin{proposition}\label{vintoalfantamorto2017GrazieTotoRiiiina}
    $$\Phi_n^*=\left\{\boldsymbol\phi\in\R_{++}^{n-1}:\forall k\in[0,n),\forall j\in[1,n-k),\ (-1)^k(\Delta^k \boldsymbol\phi)_j\geq0\right\},$$
  where
  $$(\Delta^k f)_j:=\sum_{i=0}^k(-1)^{i+k}\binom{k}{i}f(j+i).$$
  \end{proposition}
  
  \begin{example}
    The extremal elements of $\Phi_n^*$ are the exponentials $\phi_k=c\alpha^{-k}$ for $\alpha\geq 1$, $c>0$. In this case we have that
    (if $\bb=F^{-1}(\pp)$) 
    \begin{align*}
    \beta_k&=\sum_{j=0}^{k-2}(-1)^{k+j}\binom{k-2}j\phi_{n-1-j}\\
    &=c(-1)^k\alpha^{1-n}\sum_{j=0}^{k-2}\binom{k-2}j(-\alpha)^{j}\\
    &=c\alpha^{1-n}(\alpha-1)^{k-2}
    \end{align*}
    Thus $\bb$ also takes an exponential form with scale $\alpha-1\geq0$. (If $\alpha=1$ we adopt the convention~$0^0=1$ and set $\beta_2=c$.)
  \end{example}

\subsubsection*{Some special cases}\label{Connection}

As we noticed in §\ref{def}, our distance $\ddd$ relates to well-known dissimilarity functions established in the literature. In particular,
\begin{itemize}
  \item If $\bb=(1,\ldots,1)$, then we recover the distance described in \cite{dissimiOkNishi}.
  \item If $\bb=(1,\ldots,1,0,\ldots,0)$ we get the metrics proposed in \cite{GILBERT202227}.
  \item If $\bb = \left((\alpha-1)^{k-2}\right)_{k=2}^n$ for some $\alpha>1$ and $\mu$ is the counting measure, then we recover the \emph{unavailable candidate model} of \cite{lu2010unavailable}.
  \item If $\bb=(1,0,\ldots,0)$ and $\mu$ is the counting measure, then we recover the Kendall distance proposed in \cite{Kendall48} and axiomatized in \cite{Kemeny59}.
  \item If $\mu$ is the counting measure, then we recover the weighted Kendall distance proposed and axiomatized by \cite{Hassandez}.
  \item In \cite{kumar2010generalized} a particular case of the class of metrics characterized in Theorem \ref{roar} is considered. More specifically, the authors set $g(i,j)=w(i)w(j)$ for some function $w:[n]\to\R_{++}$.
\end{itemize}

\subsubsection*{Isometric embedding}\label{sect:isoembedding}Isometric embedding results can be used to translate least squares methods of approximation from one space to another. In the context of finite metric space, this approach turns out to be important to avoid the difficulties inherent to the lack of convexity in optimization problems. In the working paper version of their work \cite{dissimiOkNishi}, the authors provided a result showing the impossibility to embed the set of partial orders\footnote{Reflexive, transitive, and antisymmetric binary relations.} over $[n]$ in any strictly convex normed linear space, whenever such set is endowed with $\dd^{\mu}$ and $n\geq 3$. Here, we strengthen their result showing that it applies to a larger family of metrics
defined on the smaller collection of linear orders. Before that we recall two definitions. We say that a metric space $(X,d_X)$ can be \textit{isometrically embedded} in a metric space $(Y,d_Y)$ if there exists $f:X\to Y$ such that $d_Y(f(x),f(y))=d_X(x,y)$ for all $x,y\in X$. A (real) normed linear space $(X,\|\cdot\|)$ is \textit{strictly convex} if for all $\lambda\in (0,1)$ and distinct $x,y\in X$ with $\| x\|=\| y \|=1$ we have 
\[
\| \lambda x+(1-\lambda)y\|<1.
\]
\begin{proposition}\label{propo:iso_embedding}
Let $n\geq3$ and $\dd$ be any metric on $\LL_n$ satisfying axioms $\ref{A3}$. Then, $(\LL_n,\dd)$ cannot be isometrically embedded in any strictly convex (real) normed linear space. 
\end{proposition}
\noindent In particular, we have that for all $n\geq 3$, when  $(\LL_n,\ddd)$ is a metric space, it cannot be isometrically embedded in any strictly convex (real) normed linear space.

\section{Social choice theory}\label{section:voting}

We study the properties of the median preference function induced by the class of distances introduced in the previous sections. To this end we start providing a primer on the terminology used in Social Choice theory.

\subsubsection*{Notation and terminology}
Throughout the section we fix a nonempty finite set $C=[n]$ of candidates with $n\geq2$, over which decision makers (or \textit{voters}) express their preferences. Each voter is represented by a ranking $\sigma\in \LL_n$. A profile of voters $V=(v^1,\ldots,v^m)$ (also called, \textit{electorate}) is a finite ordered list of rankings over $[n]$.\footnote{With abuse of notation we denote by $|V|$ the dimension of $V$.} We define by 
\[
\mathbf{V}_n:=\bigcup_{k=1}^{\infty}\LL_n^k
\]
the set of all possible lists of voters over $C$. Often, we may write 
\[
\left(s_1 v^1,s_2 v^2,\ldots,s_m v^m\right)
\]
to identify a profile of voters where for each $j=1,\ldots,m$, the share of voters whose preferences are expressed by $v^j$ is $s_j\in \left[0,1\right]$.\footnote{Sometimes we replace $s_j$ with $n_j$, where the latter represents the number of voters in the electorate with preference $\sigma^j$.} Pivotal to our analysis are \textit{preference correspondences}. In particular, a preference correspondence $P:\V\rightrightarrows \LL_n$ associates to each profile of voters a subset of $\LL_n$. Given a preference correspondence $P$ and a profile of voters $V$ we sometimes refer to $P(V)$ as the set of \textit{consensus preferences}. Each preference correspondence $P$ induces a \textit{social choice correspondence}~$W_P:\V\rightrightarrows C$ defined as 
\[
W_P(V):=\bigcup\limits_{v\in P(V)} M\left(C,v\right),
\]
namely the set of \textit{winners} identified by all the consensus preferences.\\

\noindent An active part of the literature on Social Choice theory focuses on \textit{distance-based voting methods}. Specifically, a semimetric $\dd:\LL_n\tim \LL_n\to \R_+$ induces a \textit{median preference correspondence} as follows:
\[
P_{\dd}:V\mapsto \argmin\limits_{\sigma\in \LL_n}\sum_{j=1}^{|V|}\dd(\sigma,v^j).
\]
In the rest of this section, we focus on the median preference correspondence induced by $\ddd$. To ease notation we set $P^\mu_\bb:=P_{\ddd}$ and $P_\bb:=P_{\dd_\bb}$ and we define $W^\mu_\bb$ and $W_\bb$ analogously. Given~$V\in \V$, $\sigma\in \LL_n$, we also conveniently denote the cumulative distance by

\begin{equation}\label{aggregate distance}
 \ddd(\sigma,V):=\sum_{j=1}^{|V|}\ddd(\sigma,v^j). 
\end{equation}
When we refer to the Kemeny-Young preference correspondence, we mean $P_{\dd_K}$.
\subsection{Social choice properties}

We proceed to study the main properties of the preference and social choice correspondences~$P^\mu_\bb$ and~$W^\mu_\bb$. In doing so we briefly recall the definitions of such properties. Fix a preference correspondence $P$. We say that $P$ satisfies \textit{anonymity} if for all~$(v^1,\ldots,v^m)\in \V$ and permutations $\sigma\in \LL_m$ we have 
\[
P(v^1,\ldots,v^m)=P(v^{\sigma_1},\ldots,v^{\sigma_m}).
\]
Anonymity requires that voters' labels do not matter when determining the set of consensus preferences. A specular requirement is \textit{neutrality}. In particular, $P$ is \textit{neutral} if for all $V\in \V$ and permutations $\sigma\in \mathbb{S}_{n}$ we have\footnote{By $\sigma (v^1,\ldots,v^m)=(\sigma v^1,\ldots,\sigma v^m)$, similarly, $\sigma \{w^1,\ldots,w^k\}=\{\sigma w^1,\ldots,\sigma w^k\}$.}
\[
P(\sigma V)=\sigma P(V).
\]
While anonymity is a property that immediately holds for all median preference correspondences, neutrality is not always satisfied. 
It is known in the literature that $P_{\dd}$ is neutral whenever $\dd$ is neutral (see Lemma 2 and its proof in \cite{conitzer2009preference}). As it concerns our metric, we have that if $\beta_2>0$ and $\mu$ is proportional to the counting measure, then $P^\mu_\bb$ is neutral. On the other hand, if $\mu$ is not proportional to the counting measure, then $P_\dd^{\mu}$ might not be neutral, as shown in the next example.
\begin{example}\label{ex:neutrality}
Let $n=3$, $V=((1,2,3),(3,1,2),(2,3,1))$, $\bb=(1,0)$ and, $\mu\in\mathcal M^{++}_3$ be such that~$\mu_3>\max\{\mu_1,\mu_2\}$. Now, if $\sigma=(3,2,1)$, then $P^\mu_\bb(V)=\{(2,3,1)\}$ and $P^\mu_\bb(\sigma V)=\{(1,3,2)\}\neq \{(2,1,3)\}=\sigma P^\mu_\bb(V)$ showing that $P^\mu_\bb$ is not neutral.\footnote{The relative computations are reported in the Appendix.}
\end{example}
\noindent
However, we found that a metric $\dd$ satisfying Axioms \ref{A3} and \ref{A4} is neutral if and only if the corresponding $P_\dd^f$ is neutral for any increasing and strictly convex function $f:\R_+\to\R_+$.\footnote{We are defining $$P^f_\dd:V\mapsto\argmin_{\sigma\in\LL_n}\sum_{i=1}^{|V|}f(\dd(\sigma,v^i)).$$}

A powerful property that stands at that basis of the characterizations of the Kemeny-Young method (\cite{YoungLev78}) and positional voting rules (\cite{YOUNG_Borda}) is \textit{reinforcing}.\footnote{Some authors also refer to such property as \textit{consistency}.} In words, reinforcing requires that whenever two separate committees share the same consensus preferences, then deliberating together their consensus preferences should not change. Formally, a preference correspondence satisfies \textit{reinforcing} if
\begin{equation}\label{eq:reinforcingdef}
P(V_1)\cap P(V_2)\neq \emptyset \Longrightarrow P(V_1\oplus V_2)=P(V_1)\cap P(V_2),
\end{equation}
where $\oplus$ is the concatenation operator.\footnote{More explicitly,
\[
(v^1,\ldots,v^k)\oplus (w^1,\ldots,w^k):=(v^1,\ldots,v^k,w^1,\ldots,w^k).
\]} As it is relevant for our analysis we distinguish between reinforcing for $P$ and $W_P$. The latter is defined by swapping $P$ with $W_P$ in \eqref{eq:reinforcingdef}, and we refer to it as \textit{reinforcing for winners}.
\begin{proposition}\label{reinf_rank}
If $\ddd$ is a semimetric, then, $P^\mu_\bb$ satisfies reinforcing.
\end{proposition}
\begin{proposition}\label{prop:reinforcingwinners}
Let $n \geq 3$ and $\ddd$ be a semimetric. If ${\beta}_k = 0$ for all $2 \leq k \leq n-1$, then $W^\mu_\bb$ satisfies reinforcing for winners. 
Moreover, if $\bb\in \R^{n-1}_+$ and $\mu$ is proportional to the counting measure, then $W^\mu_\bb$ satisfies reinforcing winners if and only if ${\beta}_k = 0$ for all $2 \leq k \leq n-1$.  
\end{proposition}

An other important requirement is \textit{monotonicity}. Monotonicity of a preference correspondence requires that upranking a winning candidate cannot harm him. Formally, for a profile of voters~$V\in \V$ and $c\in C$, we denote by $V^c$ an \textit{upranking}. More specifically, $V^c$ is obtained by upranking $c$ for some voters, coeteris paribus.\footnote{An upranking of candidate $k$ for $\sigma\in \LL_n$ is a ranking $\pi$ such that for all $i,j\in[n]\setminus\{k\}$ we have that $i>_\sigma j$ iff $i>_\pi j$ and there is some $i\neq k$ such that $i>_{\sigma}k$ and $i<_{\pi}k$.
} We say that $P$ is \textit{monotone}, whenever
\[
c\in W_P(V)\Longrightarrow \forall V^c, \;c\in W_P(V^c).
\]
\begin{proposition}\label{prop:monotonicity}
Let $\bb\in\R_{+}^{n-1}$ and $\mu \in \mathcal{M}^+_n$. Then, $P^\mu_\bb$ is monotone.
\end{proposition}

Arguably the most iconic aspects of voting rules are related to \textit{majority}. For a given profile of voters $V\in \V$ and $c\in C$, we denote by $n(V,c)$ the number of voters that consider $c$ as their most favourite alternative \textit{minus} the number of voters who do not. A preference correspondence~$P$ satisfies the \textit{majority principle} if for all $V\in \V$ and $c\in C$,
\[
n(V,c)\geq 0\Longrightarrow c\in W_P(V).
\]
\begin{proposition}\label{prop:majority}
If $\ddd$ is a semimetric, then $P^\mu_\bb$ satisfies the majority principle.
\end{proposition}
The majority principle readily implies that $P^\mu_\bb$ also satisfies unanimity. We prove that actually more holds true. In particular, we provide a \textit{blockwise} version of Pareto dominance. To this end we need some further notation. For all $\sigma\in \LL_n$ and $i\leq k\leq n$, we define
\[
\sigma(\left[i,k\right]):=\left\lbrace \sigma_i,\ldots,\sigma_{k}\right\rbrace,
\]
\noindent
the \emph{set} of candidates in positions from $i$ to $k$ according to $\sigma$.
A preference correspondence $P$ satisfies
\begin{itemize}
\item \textit{blockwise Pareto winners} if for all $V=\left(v^1,\ldots,v^{m}\right)\in \V$,
\small
\[ \forall k \in [n], \; \left(
\forall j,j'\in \left[m\right], v^j(\left[k\right])=v^{j'}(\left[k\right])\Rightarrow \forall v \in P(V),\ v(\left[k\right])=v^1(\left[k\right]) \right)
\]
\normalsize
\item \textit{blockwise Pareto losers} if for all $V=\left(v^1,\ldots,v^{m}\right)\in \V$,
\small
\[ \forall k \in [n], \; \left(
\forall j,j'\in \left[m\right], v^j(\left[k\right]^{\mathsf{c}})=v^{j'}(\left[k\right]^{\mathsf{c}})\Rightarrow \forall v \in P(V),\ v(\left[k\right]^{\mathsf{c}})=v^1(\left[k\right]^{\mathsf{c}}) \right)
\]

\normalsize
\end{itemize}
\begin{proposition}\label{prop:blockpareto}
For all $\bb\in\R^{n-1}_+$ with $\beta_2 >0$, $P_\bb$ satisfies blockwise Pareto winners and losers.
\end{proposition}
\noindent At this point one may ask whether $P_{\bb}$ satisfies an even stronger version of the blockwise Pareto property. That is, whether for all $(v^j)_{j=1}^m\in \V$
\[
\forall k,i \in [n], \; \left(
\forall j,j'\in \left[m\right], v^j(\left[k,i\right])=v^{j'}(\left[k,i\right])\Rightarrow \forall v \in P(V),\ v(\left[k,i\right])=v^1(\left[k,i\right]) \right)
\]
This does not hold in general as shown in the following example.
\begin{example}
$C=\left\lbrace 1,2,3,4 \right\rbrace$ and $V=\left((1,2,3,4),(3,2,1,4),(4,2,3,1)\right)$. Clearly, $2$ is in the second position for all voters, though it can be shown that for $\bb=(1,1,\ldots,1)$,
$P_\bb(V)=\{(2,3,1,4)\}.$
\end{example}

A further condition is a partitionwise version of the Pareto property. In particular, we say that a preference correspondence $P$ has the \textit{partitionwise Pareto property} if for any increasing sequence~$(k_i)_{i=0}^\ell$ with $k_0=0$ and $k_\ell=n$, and any profile of voters $V=(v^j)_{j=1}^m$, it holds
\begin{gather*}
  \forall i\in[\ell],\forall j,j'\in[m],\ v^j((k_{i-1},k_{i}])=v^{j'}((k_{i-1},k_{i}])\\
  \Downarrow\\
  \forall i\in[\ell],\forall v\in P(V),\  v((k_{i-1},k_{i}])=v^1((k_{i-1},k_{i}]).
\end{gather*}

\begin{corollary}\label{coro:partitionwisepareto}
  For all $\bb\in\R^{n-1}_+$ with $\beta_2 > 0$, $P_\bb$ satisfies the partitionwise Pareto property.
\end{corollary}

A stronger majority condition is the \textit{Condorcet principle.} For all $i,j\in [n]$ and $V\in \V$, we denote by~$n_{i,j}(V)$ the number of voters who prefer $i$ to $j$ \textit{minus} the number of voters who do not. A preference correspondence $P$ satisfies the \textit{Condorcet principle} if 
$$n_{i,j}(V)>0\Longrightarrow \not\exists \sigma\in P(V): (j,i)\in \mathrm{adj}(\sigma)$$\vspace{-0.45cm}
\begin{equation}\label{eq:condorcet_def_P}
  \textrm{ and }
\end{equation}\vspace{-0.35cm}
$$n_{i,j}(V)=0\Longrightarrow\left(\left(\exists\sigma\in P(V): (i,j)\in\mathrm{adj}(\sigma) \right)\Longleftrightarrow\left(\exists\pi\in P(V):(j,i)\in\mathrm{adj}(\pi)\right)\right),
$$
where by adj$(\sigma)$ we denote the set of all couples $(i,j)$ adjacent in $\sigma$ and with $i>_\sigma j$.
Notice that a much stronger requirement in the style of the Condorcet principle would be the following,
\begin{equation}\label{eq:example_Condorcet}
n_{i,k}(V)>0\Longrightarrow \forall \sigma\in P(V),\ i>_{\sigma} k.
\end{equation}
The distinction between \eqref{eq:condorcet_def_P} and \eqref{eq:example_Condorcet} lies on the adjacency between $i$ and $k$ which must be satisfied for the Condorcet principle to hold, but it is not required by \eqref{eq:example_Condorcet}. Even though also this latter and stronger requirement appears to be meaningful, it is not satisfied even by the Kemeny-Young's rule as shown in the upcoming example. 
\begin{example}
  Let $C=\left\lbrace 1,2,3\right\rbrace$ and $V=((1,2,3),(2,3,1),(3,1,2))$. 
  Suppose that $\mu$ is the counting measure and that $\bb=(1,0)$. Then, $n_{1,2}(V)=1,\ n_{3,1}(V)=1,\ n_{2,3}(V)=1$.
Notwithstanding, we have that 
$P_\bb(V)=\left\lbrace (1,2,3),(2,3,1),(3,1,2)\right\rbrace $
and thus $P_\bb$ does not satisfy condition \eqref{eq:example_Condorcet}.
\end{example}
Given a profile of voters $V\in \V$ we say that a candidate $i$ is \textit{Condorcet for} $V$, if $i$ is preferred to all other candidates in binary comparisons by a majority of voters, i.e., $n_{i,j}(V) \geq 0$, for all~$j \neq i$. The set of Condorcet candidates for a profile $V$, will be denoted by $\mathcal{C}(V)$. A social choice correspondence~$W$ satisfies the \textit{Condorcet principle} if $\mathcal{C}(V)\subseteq W(V)$ for all $V\in \V$. The following result provides a full characterization of the Condorcet principles for both $P^\mu_\bb$ and $W^\mu_\bb$.
\begin{proposition}\label{Condorcet_characterization_PW}
Suppose $n \geq 3$, $\bb\in \R^{n-1}_+$ with $\beta_2>0$, and $\mu\in \mathcal{M}^{++}_n$. The following are equivalent:
  \begin{enumerate}[label=(\roman*)]
  \item \label{it1:condch} $W_\bb^\mu$ satisfies the Condorcet principle.
  \item \label{it2:condch} $P_\bb^\mu$ satisfies the Condorcet principle.
  \item \label{it3:condch} $\beta_k=0$ for all $k\geq 3$.
  \end{enumerate}
\end{proposition}
\noindent Such characterization strenghtens the results of \cite{YoungLev78} highlighting the restrictions to binary comparisons imposed by the Condorcet principle.
\section{Computational Aspects}\label{section:computational}
In this section we investigate the computational aspects of the preference aggregation problem. In §\ref{ILP subsection}, we advance an integer program formulation to compute a consensus preference exactly. Given the impossibility of this approach to scale efficiently with the number of alternatives ($n$) and voters~($m$), we then turn our attention to approximation algorithms. Specifically, in §\ref{Diaconis subsection}, we provide a set of generalized Diaconis-Graham inequalities and a unified framework to analyze them. These allow to obtain approximation algorithms for the corresponding preference aggregation problem. Finally, in §\ref{PTAS subsection}, we propose some truncation ideas building upon in \cite{lu2010unavailable}. By leveraging them, we retrieve a Polynomial-Time Approximation Scheme (PTAS) for the rank aggregation problem.
\subsection{Integer program formulation}\label{ILP subsection}
We recall that for all $\sigma,\pi \in \LL_n$, $\ddd(\sigma,\pi)$ is computable in polynomial time via the following identity
 \begin{equation}\label{Formula Maggica Again}
\ddd(\sigma,\pi) = \sum_{x \in [n]} \left(f_\bb(|x^{\downarrow,\sigma}|) + f_\bb(|x^{\downarrow,\pi}|) - 2f_\bb(|x^{\downarrow,\sigma} \cap x^{\downarrow,\pi}|) \right) \mu(x).
 \end{equation}
\noindent
Observe that any permutation $\sigma \in \LL_n$ can be associated to a binary square matrix $P^{\sigma} \in \{0,1\}^{n \tim n}$ such that, for all $i \neq j$, $P^{\sigma}_{i,j} = 1$ if and only if $i >_{\sigma} j$. Conversely any $P\in \{0,1\}^{n \tim n}$ which satisfies the following properties
\begin{equation}\tag{C.1}\label{Set of constraints 1}
  \begin{aligned}
  &P_{i,i} = 0, \;\textnormal{for all $i$} \in [n] \; \\
  &P_{i,j} + P_{j,i} = 1, \;\textnormal{ for all $i$} \neq j \in [n] \;\textnormal{(completeness)} \\
  &P_{i,j} + P_{j,k} -1 \leq P_{i,k},\;\textnormal{ for all $i,j,k \in [n]$ distinct}  \;\textnormal{ (transitivity)},
  \end{aligned}
\end{equation}
determines a unique linear order. Therefore, $\sigma\mapsto P^{\sigma}$ is a well-defined bijection. Given this formulation, we can see that $|{i}^{\downarrow,\sigma}| = \sum_{j=1}^{n} P^{\sigma}_{i,j}$ and $|{i}^{\downarrow,\sigma} \cap {i}^{\downarrow,\pi}| = \sum_{j=1}^{n} P^{\sigma}_{i,j} P^{\pi}_{i,j}$.

Consider a profile of $m$ voters $V\in \mathbf{V}_n$ and their corresponding matrices $(P^{{v}})_{v=1}^{m}$. By using (\ref{Formula Maggica Again}) and the matrix representation, we can obtain the following function that is equivalent (for optimization purposes) to the objective function of the preference aggregation problem.\footnote{Observe that the middle term depends only of $P^v$ and not on $P$; therefore, we can neglect it when computing an aggregate ranking.}
\begin{equation}\label{matrici ovunque}
\begin{aligned}
&\sum_{v = 1}^{m} \sum_{i = 1}^{n} \left(f_\bb\left(\sum_{j=1}^{n} P_{i,j}\right) - 2f_\bb\left(\sum_{j=1}^{n} P_{i,j} P^{{v}}_{i,j} \right) \right) \mu({i}) = \\
& \sum_{v = 1}^{m} \sum_{i = 1}^{n} \sum_{r = 0}^{n-1} \sum_{s = 0} ^{n-1} \left( f_\bb(r+s) - 2 f_\bb(s) \right) \mu(i) \mathbf{1} \begin{bmatrix}
      \sum_{j=1}^{n} (1-P_{i,j}^{v})P_{i,j} = r \\
      \sum_{j=1}^{n} P_{i,j}^{v}P_{i,j} = s 
     \end{bmatrix}.
\end{aligned}
\end{equation}
We linearize the this function by introducing binary variables $Q^{v,i}_{r,s}$ satisfying the following constraints
\begin{equation}\tag{C.2}\label{Set of constraints 2}
  \begin{aligned}
  &Q^{v,i}_{r,s} \leq 1+\frac{1}{n}\left( \sum_{j=1}^{n} P_{i,j}(1-P_{i,j}^{{v}}) - r \right) \; \textnormal{and} \; Q^{v,i}_{r,s} \leq 1 + 
  \frac{1}{n}\left( r- \sum_{j=1}^{n} P_{i,j}(1-P_{i,j}^{v}) \right)\\
&Q^{v,i}_{r,s} \leq 1 + \frac{1}{n}\left( \sum_{j=1}^{n} P_{i,j}P_{i,j}^{v} - s \right), \; Q^{v,i}_{r,s} \leq 1 + \frac{1}{n} \left( s- \sum_{j=1}^{n} P_{i,j}P_{i,j}^{{v}} \right), \; \textnormal{and} \sum_{r=0}^{n} \sum_{s=1}^{n} Q^{v,i}_{r,s} = 1. \\
  \end{aligned}
\end{equation}
The resulting optimization problem turns out to be 
\begin{equation}
\begin{aligned}
\min_{P,Q}& && \sum_{v = 1}^{m} \sum_{i = 1}^{n} \sum_{r = 0}^{n-1} \sum_{s = 0} ^{n-1} \left( f_\bb(r+s) - 2 f_\bb(s) \right) \mu(i) Q^{v,i}_{r,s} \\
\textrm{s.t.}& &&P \in \{0,1 \}^{n \tim n},\; \forall v \in [m], \forall i \in [n], Q^{v,i} \in \{0,1\}^{n\tim n}, \; \textnormal{\ref{Set of constraints 1}}, \; \textnormal{and} \;\textnormal{\ref{Set of constraints 2}}.
\end{aligned}
\end{equation}

\noindent Notably, this linear integer program is completely specified by $O(m n^3)$ variables and constraints.

\subsection{Generalized Diaconis-Graham inequalities}\label{Diaconis subsection}
Here, we provide a set of generalized Diaconis-Graham inequalities and a unified framework to study them. This allows us to retrieve approximation algorithms for the corresponding rank aggregation problems (see also \cite{dwork2001rank}). The celebrated result by \cite{diaconis1977spearman} states that the Kendall and Spearman's footrule distances are within a factor of two.\footnote{That is $\dd^{\textnormal{apx}} \leq \dd_{K} \leq 2 \dd^{\textnormal{apx}}$, where $\dd^{\textnormal{apx}}$ denotes the Spearman's footrule distance.} We generalize this result to the family of $(\bb,\mu)$-top difference distances. We first focus on the \emph{neutral} case and then we argue how to extend our results when neutrality fails.

\noindent For all $\bb \in \mathbb{R}_{+}^{n-1}$, the $\bb$-\textit{Spearman's footrule distance} is defined as 
\[ \dd^{\textnormal{apx}}_\bb:(\sigma,\pi)\mapsto \sum_{x \in [n]} \big | f_\bb(|x^{\downarrow,\sigma}|) - f_\bb(|x^{\downarrow,\pi}|) \big |. \]
Notably, differently from \eqref{DioMerda}, to compute $\dd^{\textnormal{apx}}_\bb(\sigma,\pi)$ we do not need to evaluate all the intersections~$x^{\downarrow,\sigma}\cap x^{\downarrow,\pi}$. This makes the computation of the $\bb$-Spearman's footrule distance relatively easier than computing our distance.

The main result of this subsection is the following theorem which provides a unified analysis of generalized Diaconis-Graham inequalities. 
\begin{theorem}\label{general approximation thm}
Let $\bb \in \mathbb{R}_{+}^{n-1}$ with $\beta_2> 0$. Then, for all $\sigma, \pi \in \mathbb{S}_{n}$,
\[ \dd_\bb^{\textnormal{apx}}(\sigma, \pi) \leq \dd_\bb(\sigma, \pi) \leq \gamma_\bb \dd_\bb^{\textnormal{apx}}(\sigma, \pi),\]
\noindent
where $ \gamma_\bb := \max_{h \in [n]} \left ( 1+\frac{f_\bb(n-h)}{f_\bb(n-(h-1))} \right ).$
\end{theorem}

Theorem \ref{general approximation thm} implies that for each choice of the weights $\bb$, the $\bb$-top difference metric and~$\bb$-Spearman's footrule distance remain within a factor of two. Moreover, Theorem \ref{general approximation thm} can be applied to specific weights in order to obtain tighter bounds. In particular,

\begin{corollary}
For all permutations $\sigma, \pi \in \mathbb{S}_{n}$, it holds
\begin{itemize}
\item For $\bb = (1,0,\ldots,0)$,
$\dd_\bb^{\textnormal{apx}}(\sigma, \pi) \leq \dd_\bb(\sigma, \pi) \leq 2 \dd_\bb^{\textnormal{apx}}(\sigma, \pi)$, 
\item For $\bb = (1,1,\ldots,1)$, $\dd_\bb^{\textnormal{apx}}(\sigma, \pi) \leq \dd_\bb(\sigma, \pi) \leq \frac{3}{2} \dd_\bb^{\textnormal{apx}}(\sigma, \pi)$,
\item For $\beta_k = p^{n-k}(1-p)^k$ with $p \in (0,1)$, $\dd_\bb^{\textnormal{apx}}(\sigma, \pi) \leq \dd_\bb(\sigma, \pi) \leq (1+p) \dd_\bb^{\textnormal{apx}}(\sigma, \pi)$
\item For $\beta_k = k$, $\dd_\bb^{\textnormal{apx}}(\sigma, \pi) \leq \dd_\bb(\sigma, \pi) \leq \frac{3}{2} \dd_\bb^{\textnormal{apx}}(\sigma, \pi)$.


\end{itemize}
\end{corollary}

In their seminal work, \cite{dwork2001rank} proposed to obtain polynomial-time approximation algorithms for the Kemeny rank aggregation problem by minimizing the sum of Spearman's footrule distances via a minimum cost perfect matching algorithm. 

Following a similar spirit, we leverage on Theorem \ref{general approximation thm} to provide approximation algorithms for our rank aggregation problems. First, we show that, by minimizing the sum of $\bb$-Spearman's footrule distances, we obtain $\gamma_\bb$-approximations for the corresponding rank aggregation problem. Using the notation introduced in §\ref{section:voting}, we let
\[
P^{\textnormal{apx}}_\bb(V):=\underset{\sigma \in \LL_n}{\textnormal{argmin}} \; \dd^{\textnormal{apx}}_\bb(\sigma,V),
\]
for all $V\in \V$ and $\bb\in \R^{n-1}$.
\begin{corollary}\label{corollary approximation}
Let $\bb \in \mathbb{R}^{n-1}_{+}$ with $\beta_2 > 0$, $V\in \V$, and $\sigma^{\textnormal{apx}} \in P^{\textnormal{apx}}_\bb(V)$.
\noindent
Then,
\[ \min_{\sigma \in \LL_n} \dd_\bb(\sigma,V) \leq \dd_\bb(\sigma^{\textnormal{apx}},V) \leq \gamma_\bb \min_{\sigma \in \LL_n} \dd_\bb(\sigma,V).\]
\end{corollary}
\noindent
The Hungarian algorithm can be implemented to retrieve a permutation minimizing $\dd^{\textnormal{apx}}_\bb(\cdot,V)$.

\begin{lemma}\label{hungarian lemma}
 Let $\bb \in \mathbb{R}_{+}^{n-1}$, $V\in \V$, and $\sigma^{\textnormal{apx}} \in P^{\textnormal{apx}}_\bb(V)$. Then, $\sigma^{\textnormal{apx}}$ is computable in $O(m n^3)$.
\end{lemma}
\noindent
By combining Theorem \ref{general approximation thm} with Corollary \ref{corollary approximation} and Lemma \ref{hungarian lemma}, we obtain the following. 
\begin{theorem}\label{theoremungulato}
Let $\bb \in \mathbb{R}_{+}^{n-1}$ with $\beta_2> 0$ and $V \in \mathbf{V}_n$. Then, there exists a $\gamma_\bb$-approximation algorithm for the $\bb$-rank aggregation problem running in time $O(m n ^3)$.
\end{theorem}

Finally, we observe that the machinery developed above can accommodate some amount of non-neutrality, as captured by the \textit{range} of $\mu$. 

\begin{corollary}\label{corollariocinghiale}
Let $\bb \in \mathbb{R}^{n-1}_{+}$ with $\beta_2 > 0$, and $\mu \in \mathcal{M}_n^{++}$ with $u \leq \mu_i \leq U$, for some positive $u$ and~$U$. Then, for all $\sigma, \pi \in \LL_n$, it holds
\[ \dd^{\textnormal{apx}}_{\mu,\bb}(\sigma,\pi) \leq \ddd(\sigma,\pi) \leq \gamma_\bb \frac{U}{u} \dd^{\textnormal{apx}}_{\mu,\bb}(\sigma,\pi),\]
\noindent
where $\dd^{\textnormal{apx}}_{\mu,\bb}(\sigma,\pi) := \sum_{x \in [n]} \big | f_\bb(|x^{\downarrow,\sigma}|) - f_\bb(|x^{\downarrow,\pi}|) \big |\mu(x)$.
\end{corollary}
We address the question of whether it is possible to obtain competitive ratios that do not depend on the support of $\mu$ in the following section.

\subsection{Polynomial-Time-Approximation schemes}\label{PTAS subsection}
In this final subsection, by building upon \cite{lu2010unavailable}, we provide a PTAS for the $(\bb,\mu)$-rank aggregation problem (Algorithm \ref{MyopicTop}). We defer the reader to the discussion paragraph below Theorem \ref{PTAS} for a comparison with \cite{lu2010unavailable}.
 
Recall that our class of metrics penalize more mismatches happening over the most favorite alternatives compared to disagreements over least favorite ones. The intuition behind the Algorithm~\ref{MyopicTop} is simple yet effective: as disagreements over least favorite alternatives are less important, we can ``truncate'' our metrics to account only for mismatches over the top positions and minimize the sum of these truncated ``distances''. We informally refer to the corresponding minimization problem as truncated rank aggregation problem. Alternatives in the remaining positions are arbitrarily ordered. By appropriately choosing the truncation depth, we can ensure simultaneously a polynomial running time for the algorithm and a small error due to the arbitrary ordering for the least preferred alternatives.
We now introduce the truncation that we are going to exploit subsequently. Let $\mu \in \mathcal{M}^+_n$, $\bb \in \mathbb{R}_{+}^{n-1}$, and $b,t \in \mathbb{N}$ such that $1\leq b \leq t \leq n$. The $(\mu,\bb,b,t)$-\textit{top-difference truncation} is defined as
\begin{equation}\label{Truncation}
\ddd(\cdot,\cdot;b,t):(\sigma,\pi)\mapsto \sum_{i=b}^{t}\left[ f_\bb(n-i) \left(\mu({\sigma_{i}}) + \mu({\pi_{i}})\right) - 2 f_\bb(| {\sigma_{i}}^{\downarrow,\sigma} \cap {\sigma_{i}}^{\downarrow,\pi}|)\mu({\sigma_{i}})\right].
\end{equation}
Notice that, for all permutations $\sigma,\pi,\omega,\tau \in \LL_n$ agreeing on positions $b \leq i \leq t$, we have that~$\ddd(\sigma,\pi;b,t) = \ddd(\omega,\tau;b,t)$. For simplicity we denote $\ddd(\sigma,V;b,t) := \sum_{j=1}^{|V|} \ddd(\sigma,v^j;b,t)$, for all $V=(v^j)\in \V$. It is observed that $\dd^{\mu}_\bb(\cdot,\cdot;1,n)=\dd^\mu_\bb(\cdot,\cdot)$. Given the above notation, Algorithm \ref{MyopicTop} works as follows

\vspace{0.25cm}

\begin{algorithm}[H]\label{MyopicTop}
{\small
\SetKwInput{Init}{Input}
\SetKwInOut{Out}{Output}
\SetKwInput{Par}{Parameters}
 \Init{$K \geq 1$, $V = (v^1,\ldots,v^m)$, $\bb \in \mathbb{R}_{+}^{n-1}$, and $\mu\in \mathcal{M}_n$.}
 Set $Q \leftarrow [n]$\\
 Set $i \leftarrow 1$\\
 \While{$ \exists c \in Q: \left|\left\{ j \in [m] \mid \forall a\in Q \setminus \{c\},\ c >_{v^j} a  \right\}\right| > \frac{m}{2} $}{
 \vspace{0.1cm}
 ${\sigma_{i}} \leftarrow c$\\
 $Q \leftarrow Q \setminus \{c\}$\\
 $i \leftarrow i+1$
 }
 $K \leftarrow \min\{K,|Q|\}$ \\
 $\sigma^{\textnormal{apx}} \in \argmin_{\omega \in \LL_n } \ddd(\omega,V;i,i+K-1)$ s.t. $\omega_{k} = \sigma_{k}$, for $k < i$. \\
\Out{$\sigma^{\textnormal{apx}}$}
 \caption{MyopicTop \cite{lu2010unavailable}}
 }
\end{algorithm}
\vspace{0.25cm}
\noindent Algorithm \ref{MyopicTop} first searches for majority candidates to be set in the top positions of the consensus ranking. Then, it solves the truncated rank aggregation problem. Before passing to the main theorem of this section we need some notation. We define
$$\boldsymbol{\mathcal{M}}:=\left\{\left(\mu^{(k)}\right)_{k=2}^\infty\in \bigtimes_{k=2}^\infty\mathcal M^+_k: 0<\inf_{k,i}\mu^{(k)}_i\leq\sup_{k,i}\mu^{(k)}_i<\infty\right\}.$$
The next result provides a polynomial time approximation scheme for the rank aggregation problem. 

\begin{theorem}\label{PTAS}
Let $(\beta_{k})_{k=2}^{\infty}\in \mathbb{R}^{\mathbb{N}}_+$ with $\beta_2 > 0$ and satisfying
\begin{equation}\label{growth rate}
  \lim_{K\to\infty} \sup_{t \geq K \vee 2} \frac{\sum_{j=2}^{t-K} \beta_{j} \binom{t-K}{j}}{\sum_{j=0}^{t-2} \beta_{j+2} \binom{t-2}{j}} = 0.
\end{equation}
Let $(\mu^{(k)})_{k=2}^{\infty} \in \boldsymbol{\mathcal{M}}$, and suppose that $0 < u \leq \mu^{(k)}_i \leq U < \infty$, for all $i$ and $k$. Consider the function

\[g: \frac{1}{\epsilon} \longmapsto \min \left\lbrace K \in \mathbb{N} \bigg| \sup_{t \geq K \vee 2} \frac{\sum_{j=2}^{t-K} \beta_{j} \binom{t-K}{j}}{\sum_{j=0}^{t-2} \beta_{j+2} \binom{t-2}{j}} \leq \epsilon \right\rbrace. \]

Fix $\epsilon > 0$. Then, for all $n \geq 2$ and $V \in \mathbf{V}_n$, 

\[ \dd^{\mu^{(n)}}_{\bb_{2:n}}(\sigma^{\textnormal{apx}},V)\leq (1 + \epsilon)\dd^{\mu^{(n)}}_{\bb_{2:n}}(\sigma^{*},V) ,\]
\noindent
where $\sigma^{*} \in P^{\mu^{(n)}}_{\bb_{2:n}}(V)$, and $\sigma^{\textnormal{apx}}$ is the ranking obtained by running Algorithm \ref{MyopicTop} with input $K = g(\frac{12U}{u \epsilon})$, $V$, $\bb_{2:n}$ and $\mu^{(n)}$. Furthermore, Algorithm \ref{MyopicTop} runs in time $O\left(n^{g\left(\frac{12 U }{ u \epsilon}\right)+1}mg\left(\frac{12 U }{ u \epsilon}\right) \right)$.
\end{theorem}

\noindent We remark that the non-polynomial dependence on $\epsilon$ suffices for a PTAS, as $\epsilon$ is a fixed but arbitrary constant. The following are simple applications of Theorem \ref{PTAS}:\footnote{The computations are reported in Appendix \ref{comps}.}

\begin{itemize}
  \item If $\beta_j = (\alpha - 1)^{j}$ for $\alpha > 1$, then $g(1/\epsilon) = \left\lceil \log_{\alpha}\left(\frac{\alpha^2}{(\alpha-1)^2 \epsilon}\right) \right\rceil$.
  \item If $\beta_j = j+1$, then $g(1/\epsilon) = \lceil \log_2(4/\epsilon) \rceil$.
  \item If $\beta_j = 1 + (-1)^j$, then $g(1/\epsilon) = \lceil \log_2(4/\epsilon) \rceil $.
\end{itemize}

\noindent\paragraph{\textbf{Discussion on PTAS}}{The \emph{unavailable candidate model} proposed by \cite{lu2010unavailable} is a particular case of the distances considered in this work; it corresponds to the choice of exponential weights and counting measure. Our analysis extends their result along two dimensions: on the one hand, Theorem \ref{PTAS} applies to all metrics satisfying Condition \eqref{growth rate} and not just exponential weights. This allows us to identify sufficient conditions on the growth rate of the weights $\bb$ for this approach to work. See the list above for several examples. Further, by using the analogy established in Theorem~\ref{axiomatic_characterization_THM}, Algorithm~\ref{MyopicTop} turns out to be also a PTAS for the family of weighted Kendall distances. 

On the other hand, our analysis applies to the case where neutrality fails. Specifically, we show that it is possible to obtain competitive ratios that do not degrade with the support of $\mu$ at the cost of increasing the running time of the algorithm. Finally, we stress out that the formulation of the weighted top-difference truncation and its applicability heavily rely on the identity \eqref{DioMerda}.



\section{Future research}
In this paper we studied a class of distances focusing on its axiomatic foundations, social choice implications, and computational features. We deem that there is much room for future research.

On the social choice aspects, it would be interesting to understand how to characterize the median voting rule induced by our distance. That is, to provide conditions for a preference correspondence~$P$ which hold if and only if $P=P_{\bb}^{\mu}$ for some $\bb$ and $\mu$. In addition, there are many properties of preference correspondences that we did not investigate. Strategy-proofness would be a very interesting property to study, as well as the distance rationalization aspects linked to the weighted top-difference distance.

From a more computational standpoint, it remains an open problem to show the existence of practical approximation algorithms whose competitive ratio does not degrade with the support of~$\mu$. Even though we did not pursue this avenue, it is an interesting direction to study the complexity of strategic manipulation (e.g., \cite{bartholdi1989computational}, \cite{conitzer2003universal}, \cite{conitzer2003many}), as well as  elicitation of voters' preferences (e.g., \cite{conitzer2002vote}, \cite{goel2019knapsack}, \cite{benade2021preference}) under the distances proposed in this work.

\section*{Acknowledgements}
We are grateful to Andrea Agazzi, Shreya Arya, Arjada Bardhi, Fabio Maccheroni, Erik Madsen, Sayan Mukherjee, Debraj Ray, Daniela Saban, Giacomo Mantegazza, Ilan Morgenstern, Katerina Papagiannouli, Rohit Kumar, and Lucrezia Villa for useful discussions. We are particularly indebted to Efe Ok for suggesting to study the social choice aspects of the metric developed in \cite{dissimiOkNishi}. Andrea Aveni gratefully acknowledges the financial support of Duke University, Max Plank Institute and the Humboldt fellowship. 
Ludovico Crippa gratefully acknowledges the financial support of the Stanford Graduate School of Business. 
Giulio Principi gratefully acknowledges the financial support of the MacCracken fellowship. 
\bibliographystyle{plain}
\bibliography{bibliography}

\begin{figure}[ht!]
\centering
\scalebox{0.65}{
  \begin{tikzpicture}[scale=5,tdplot_main_coords]
  \tikzstyle{vertex}=[minimum size=1pt,fill=white]
  \tikzstyle{1edge} = [draw,line width=2pt,-,blue!90,opacity=1]
  \tikzstyle{1edge+} = [draw,line width=2pt,-,blue!90,dashed]
  \tikzstyle{2edge} = [draw,line width=2pt,-,red!90]
  \tikzstyle{2edge+} = [draw,line width=2pt,-,red!90,dashed]
  \tikzstyle{3edge} = [draw,line width=2pt,-,green!90]
  \tikzstyle{3edge+} = [draw,line width=2pt,-,green!90,dashed]
    \coordinate (1234) at (1, 2, 3, 4);
    \coordinate (2134) at (2, 1, 3, 4);
    \coordinate (2314) at (3, 1, 2, 4);
    \coordinate (2341) at (4, 1, 2, 3);
    \coordinate (3241) at (4, 2, 1, 3);
    \coordinate (3214) at (3, 2, 1, 4);
    \coordinate (3124) at (2, 3, 1, 4);
    \coordinate (1324) at (1, 3, 2, 4);
    \coordinate (1342) at (1, 4, 2, 3);
    \coordinate (3142) at (2, 4, 1, 3);
    \coordinate (3412) at (3, 4, 1, 2);
    \coordinate (3421) at (4, 3, 1, 2);
    \coordinate (4321) at (4, 3, 2, 1);
    \coordinate (4312) at (3, 4, 2, 1);
    \coordinate (4132) at (2, 4, 3, 1);
    \coordinate (1432) at (1, 4, 3, 2);
    \coordinate (1423) at (1, 3, 4, 2);
    \coordinate (4123) at (2, 3, 4, 1);
    \coordinate (4213) at (3, 2, 4, 1);
    \coordinate (4231) at (4, 2, 3, 1);
    \coordinate (2431) at (4, 1, 3, 2);
    \coordinate (2413) at (3, 1, 4, 2);
    \coordinate (2143) at (2, 1, 4, 3);
    \coordinate (1243) at (1, 2, 4, 3);

    \draw[1edge+] (1234) -- (2134)node [rectangle,minimum size=5pt,midway, fill=white] {$g_1(1,2)$};
    \draw[2edge+] (2134) -- (2314)node [rectangle,minimum size=5pt,midway, fill=white] {$g_2(1,3)$};
    \draw[1edge+] (2314) -- (3214)node [rectangle,minimum size=5pt,midway, fill=white] {$g_1(2,3)$};
    \draw[2edge+] (3214) -- (3124)node [rectangle,minimum size=5pt,midway, fill=white] {$g_2(1,2)$};
    \draw[1edge+] (3124) -- (1324)node [rectangle,minimum size=5pt,midway, fill=white] {$g_1(1,3)$};
    \draw[2edge+] (1324) -- (1234)node [rectangle,minimum size=5pt,midway, fill=white] {$g_2(2,3)$};
    \draw[3edge+] (2143) -- (2134)node [rectangle,minimum size=5pt,midway, fill=white] {$g_3(3,4)$};
    \draw[3edge+] (2314) -- (2341)node [rectangle,minimum size=5pt,midway, fill=white] {$g_3(1,4)$};
    \draw[3edge+] (3241) -- (3214)node [rectangle,minimum size=5pt,midway, fill=white] {$g_3(1,4)$};
    \draw[3edge+] (3124) -- (3142)node [rectangle,minimum size=5pt,midway, fill=white] {$g_3(2,4)$};
    \draw[3edge+] (1324) -- (1342)node [rectangle,minimum size=5pt,midway, fill=white] {$g_2(2,4)$};
    \draw[3edge+] (1234) -- (1243)node [rectangle,minimum size=5pt,midway, fill=white] {$g_2(3,4)$};
    \draw[3edge] (4321) -- (4312)node [rectangle,minimum size=5pt,midway, fill=white] {$g_3(1,2)$};
    \draw[2edge] (4312) -- (4132)node [rectangle,minimum size=5pt,midway, fill=white] {$g_2(1,3)$};
    \draw[3edge] (4132) -- (4123)node [rectangle,minimum size=5pt,midway, fill=white] {$g_3(2,3)$};
    \draw[2edge] (4123) -- (4213)node [rectangle,minimum size=5pt,midway, fill=white] {$g_2(1,2)$};
    \draw[3edge] (4213) -- (4231)node [rectangle,minimum size=5pt,midway, fill=white] {$g_3(1,3)$};
    \draw[2edge] (4231) -- (4321)node [rectangle,minimum size=5pt,midway, fill=white] {$g_2(2,3)$};
    \draw[1edge] (4123) -- (1423)node [rectangle,minimum size=5pt,midway, fill=white] {$g_1(1,4)$};
    \draw[2edge] (1423) -- (1243)node [rectangle,minimum size=5pt,midway, fill=white] {$g_2(2,4)$};
    \draw[1edge] (1243) -- (2143)node [rectangle,minimum size=5pt,midway, fill=white] {$g_1(1,2)$};
    \draw[2edge] (2143) -- (2413)node [rectangle,minimum size=5pt,midway, fill=white] {$g_2(1,4)$};
    \draw[1edge] (2413) -- (4213)node [rectangle,minimum size=5pt,midway, fill=white] {$g_1(2,4)$};
    \draw[1edge] (4312) -- (3412)node [rectangle,minimum size=5pt,midway, fill=white] {$g_1(3,4)$};
    \draw[2edge] (3412) -- (3142)node [rectangle,minimum size=5pt,midway, fill=white] {$g_2(1,4)$};
    \draw[1edge] (3142) -- (1342)node [rectangle,minimum size=5pt,midway, fill=white] {$g_1(1,3)$};
    \draw[2edge] (1342) -- (1432)node [rectangle,minimum size=5pt,midway, fill=white] {$g_2(3,4)$};
    \draw[1edge] (1432) -- (4132)node [rectangle,minimum size=5pt,midway, fill=white] {$g_1(1,4)$};
    \draw[1edge] (4321) -- (3421)node [rectangle,minimum size=5pt,midway, fill=white] {$g_1(3,4)$};
    \draw[2edge] (3421) -- (3241)node [rectangle,minimum size=5pt,midway, fill=white] {$g_2(2,4)$};
    \draw[1edge] (3241) -- (2341)node [rectangle,minimum size=5pt,midway, fill=white] {$g_1(2,3)$};
    \draw[2edge] (2341) -- (2431)node [rectangle,minimum size=5pt,midway, fill=white] {$g_2(3,4)$};
    \draw[1edge] (2431) -- (4231)node [rectangle,minimum size=5pt,midway, fill=white] {$g_1(2,4)$};
    \draw[3edge] (3421) -- (3412)node [rectangle,minimum size=5pt,midway, fill=white] {$g_3(1,2)$};
    \draw[3edge] (1423) -- (1432)node [rectangle,minimum size=5pt,midway, fill=white] {$g_3(2,3)$};
    \draw[3edge] (2413) -- (2431)node [rectangle,minimum size=5pt,midway, fill=white] {$g_3(1,3)$};
    \node[vertex] at (1234) {$1234$};
    \node[vertex] at (2134) {$2134$};
    \node[vertex] at (2314) {$2314$};
    \node[vertex] at (2341) {$2341$};
    \node[vertex] at (3241) {$3241$};
    \node[vertex] at (3214) {$3214$};
    \node[vertex] at (3124) {$3124$};
    \node[vertex] at (1324) {$1324$};
    \node[vertex] at (1342) {$1342$};
    \node[vertex] at (3142) {$3142$};
    \node[vertex] at (3412) {$3412$};
    \node[vertex] at (3421) {$3421$};
    \node[vertex] at (4321) {$4321$};
    \node[vertex] at (4312) {$4312$};
    \node[vertex] at (4132) {$4132$};
    \node[vertex] at (1432) {$1432$};
    \node[vertex] at (1423) {$1423$};
    \node[vertex] at (4123) {$4123$};
    \node[vertex] at (4213) {$4213$};
    \node[vertex] at (4231) {$4231$};
    \node[vertex] at (2431) {$2431$};
    \node[vertex] at (2413) {$2413$};
    \node[vertex] at (2143) {$2143$};
    \node[vertex] at (1243) {$1243$};
  \end{tikzpicture}}
  \label{fecipermutaedriche}
  \caption{The permutahedron on $[4]$.}
\end{figure}

\clearpage
\appendix
\counterwithin{theorem}{subsection}

\section*{Appendix}
\addcontentsline{toc}{section}{Appendix}
\renewcommand{\thesubsection}{\Alph{subsection}}

\subsection{Toolkit lemmas}

\begin{lemma}\label{Ehvoleeevi}
  For all $\bb\in\R_+^{n-1}$ and $\mu\in \mathcal{M}_n^+$, we have that $\ddd$ is a semimetric.
\end{lemma}
\begin{proof}
  Non-negativity and symmetry are straightforward. Moreover, if $\sigma=\pi$, then for all $S\subseteq [n]$, we have that $M(S,\sigma)=M(S,\pi)$ thus $\bigtriangleup_S(\sigma,\pi)=\emptyset$. Hence, $\ddd(\sigma,\pi)=0$. In conclusion, we need to show the triangle inequality. For all $S\subseteq [n]$ and $\sigma,\omega,\pi\in \mathbb{S}_n$,
  $$M(S,\sigma)\bigtriangleup M(S,\pi)\subseteq (M(S,\sigma)\bigtriangleup M(S,\omega))\sqcup(M(S,\omega)\bigtriangleup M(S,\pi)).$$
  Thus, we conclude that
  $\ddd(\sigma,\pi)\leq \ddd(\sigma,\omega)+\ddd(\omega,\pi)$.\end{proof}

\begin{corollary}\label{metric}
  For all $\bb\in\R_+^{n-1}$ with $\beta_2>0$ and $\mu \in \mathcal{M}_n^{++}$, we have that $\ddd$ is a metric.
\end{corollary}
\begin{proof}
  By Lemma \ref{Ehvoleeevi}, it is sufficient to prove that $\ddd(\sigma,\pi)=0$ implies $\sigma=\pi$. Indeed, if $\sigma\neq\pi$, then we can find $i,j\in [n]$ distinct so that $M(\{i,j\},\sigma)\neq M(\{i,j\},\pi)$, then $$\ddd(\sigma,\pi)\geq \beta_2\mu(\bigtriangleup_{\{i,j\}}(\sigma,\pi))\geq\beta_2(\mu_i+\mu_j)>0.$$\end{proof}

\begin{lemma}\label{DioMajale}
  If $\bb\in \R_+^{n-1}$ and $\beta_2>0$, then $f_\bb:\mathbb N\to\R$ is increasing and non-negative. If in addition $\beta_2>0$, then $f_\bb$ is strictly increasing and positive.
\end{lemma}
\begin{proof}
  We have that
  $f_\bb(t)\geq \beta_2\binom{t}{1}=t\beta_2\geq 0$ for all $t\geq 1$. As for the monotonicity, we have that
  $$f_\bb(t+1)-f_\bb(t)=\sum_{k=2}^n\beta_k\left(\binom{t+1}{k-1}-\binom{t}{k-1}\right)=\sum_{k=2}^n\beta_k\binom{t}{k-2}\geq \beta_2\binom{t}{0}=\beta_2\geq 0$$
  for all $t\geq 1$. Thus, $f_\bb$ is increasing. If $\beta_2>0$, then the previous inequalities yields the second part of the claim.
\end{proof}
\begin{lemma}\label{DioCinghiale} If $\bb\in \mathbb{R}^{n-1}_+$, then $f_\bb$ has increasing increments. If in addition $\beta_2>0$, then $f_\bb$ has strictly increasing increments.
\end{lemma}
\begin{proof} Let $l\geq m$ and $h\geq0$, then
\begin{align*}
  (f_{\bb}(l+h)-f_{\bb}(l))-(f_{\bb}(m+h)-f_{\bb}(m))&=\sum_{j=1}^{n-1}\beta_{j+1}\left(\binom{l+h}{j}-\binom{l}{j}-\binom{m+h}{j}+\binom{m}{j}\right)\\
  &=\sum_{j=1}^{n-1}\beta_{j+1}\left(\sum_{k=l}^{l+h-1}\binom{k}{j-1}-\sum_{k=m}^{m+h-1}\binom{k}{j-1}\right)\\
  &=\sum_{j=1}^{n-1}\beta_{j+1}\left(\sum_{k=m}^{m+h-1}\binom{k+l-m}{j-1}-\binom{k}{j-1}\right)\geq 0.
\end{align*}
Therefore, $f_{\bb}$ has non-decreasing increments. If $\beta_2>0$, then the second part of the claim follows from analogous steps.
\end{proof}


The following three lemmas are used in the proof of Theorem \ref{PTAS}. First, by exploiting the fact that $\mu$ is a bounded measure, we directly get the following upper bound.

\begin{lemma}\label{lemmino per upper boundino}
Let $\mu \in \mathcal{M}^+_n$, with $\mu \leq U$, $\bb \in \mathbb{R}_{+}^{n-1}$ and $b,t \in \mathbb{N}$ such that $1\leq b \leq t \leq n$. Then, for all $\sigma \in \mathbb{S}_{n}$ and $V \in \V$,
\begin{equation}\label{scampanatameritata} \ddd(\sigma,V;b,t) \leq 2 |V| U \sum_{i = b}^{t} f_\bb(n-i). \end{equation}
\end{lemma}
\begin{proof}
It follows immediately from the definition.
\end{proof}
Notice that when $\mu\equiv U$, $V=(\pi,\pi,\ldots,\pi)$ for some $\pi\in \LL_n$, and $\sigma$ is \textit{antipodal} to $\pi$, then \eqref{scampanatameritata} holds as an equality. Secondly, we provide a lower bound. 

\begin{lemma}\label{lemmino per lower boundino}
Let $\mu \in \mathcal{M}_n^{++}$, with $\mu_i \geq u > 0$, $\bb \in \mathbb{R}_{+}^{n-1}$, $\sigma \in \mathbb{S}_{n}$, and $V=(v^j)_{j=1}^{|V|} \in \V$. Suppose there exists~$k \in \{0,\ldots,n-2\}$ such that 
\[ | \left\{ j \in [|V|] \mid {\sigma_{k+1}} >_{v^j} a, \; \textnormal{for all } a \in C \setminus \{{\sigma_{1}},\ldots,{\sigma_{k+1}}\} \right\}| \leq \frac{|V|}{2}.\]
Then, $\ddd(\sigma,V) \geq \frac{u |V| }{2} \phi_{k+1},$ where we recall that $\phi_j :=f_\bb(n-j)-f_\bb(n-j-1)$.
\end{lemma}

\begin{proof}
First of all, observe that $\ddd(\sigma,V) \geq u \dd_\bb(\sigma,V)$. Then, consider the following profile ~$\tilde{V}$ for $j \in \{1,\ldots, \lfloor |V|/2 \rfloor\}$, $\omega^{j} := \sigma$, while, for $j \in \{ \lfloor |V|/2 \rfloor +1,\ldots,|V|\}$, 
$\omega^{j}=\sigma t_{k+1,k+2}$. By construction, $\dd_\bb(\sigma,V) \geq \dd_\bb(\sigma,\tilde{V})$ and $\dd_\bb(\sigma,\tilde{V}) \geq \frac{u |V|}{2} \phi_{k+1}$.
\end{proof}
Suppose we have a consensus ranking $\sigma$ and remove the first $k$ alternatives. 
The following lemma ensures that for any consensus ranking $\pi$ on $C\setminus \sigma([k])$, the concatenation 
\[
(\sigma_1,\ldots,\sigma_k,\pi_1,\ldots,\pi_{n-k})
\]
is a consensus ranking on $C$. For each $A \subset [n]$, we denote by $\sigma^{-A}$ the linear order we obtain by removing all candidates in $A$; for example, suppose $\sigma = (1,\ldots,{n})$ and $A = \{n\}$, then $\sigma^{-A} = (1,\ldots,{n-1})$.

\begin{lemma}\label{split optimum}
Let $\bb \in \mathbb{R}_{+}^{n-1}$, $\mu \in \mathcal{M}_n^{++}$, $V=(v^j)_{j=1}^m \in \V$, $\sigma \in P^{\mu}_{\bb}(V)$, and $1 \leq k \leq n-1$. Define $C_{1:k} := \{{\sigma_{1}},\ldots,{\sigma_{k}}\}$. Let $\pi: \{k+1,\ldots,n\} \to C_{k+1:n}$ be a bijective map minimizing 
{\small
\begin{equation}\label{porco dio gli indici}
\omega\mapsto \sum_{j=1}^{m} \sum_{i=k+1}^{n}\left(f_{\bb_{2:n-k}}\left(\left| {\sigma_{i}}^{\downarrow,\omega}  \right|\right)+f_{\bb_{2:n-k}}(| {\sigma_{i}}^{\downarrow,\tilde{v}^j}  |)-2f_{\bb_{2:n-k}} (| {\sigma_{i}}^{\downarrow,\omega} \cap {\sigma_{i}}^{\downarrow,\tilde{v}^j}  |)\right)\mu({\sigma_{i}}),\footnote{For some bijection $\omega:\left\lbrace k+1,\ldots,n\right\rbrace\to C_{k+1:n}$, and $r,s \in C_{k+1:n}$, we define $r>_{\omega}s \iff \omega^{-1}_r < \omega^{-1}_{s}$. Accordingly, 
\[
\sigma_i^{\downarrow,\omega} := \{j\in C_{k+1:n}:\sigma_i>_{\omega}j\}.
\]}
\end{equation}
}
where, for each $j \in [m]$, $\tilde{v}^j = (v^j)^{-C_{1:k}}$. Then, $\sigma^{*} := (\sigma_{1},\ldots,\sigma_{k}, \pi_{k+1},\ldots,\pi_{n}) \in P^{\mu}_{\bb}(V)$ .
\end{lemma}

\begin{proof}
The proof follows by observing that
\begin{equation*}
  \begin{aligned}
  \dd^\mu_\bb(\sigma,V) - \sum_{j=1}^{m} \sum_{i=1}^{n}f_\bb\left(\left| {i}^{\downarrow,v^j}\right|\right)\mu(i) &= \sum_{j=1}^{m} \sum_{i=1}^{n}\left(f_\bb\left(\left| {i}^{\downarrow,\sigma}\right|\right) -2f_\bb \left( \left| {i}^{\downarrow,\sigma} \cap {i}^{\downarrow,v^j}\right| \right)\right)\mu({i}) \\
  &= \sum_{j=1}^{m} \sum_{i=1}^{n}\left(f_\bb\left(\left| {\sigma_{i}}^{\downarrow,\sigma}\right|\right)  -2f_\bb \left( \left| {\sigma_{i}}^{\downarrow,\sigma} \cap {\sigma_{i}}^{\downarrow,v^j}\right| \right)\right)\mu({\sigma_{i}}) \\
  &= \sum_{j=1}^{m} \sum_{i=1}^{k}\left(f_\bb\left(\left| {\sigma_{i}}^{\downarrow,\sigma}\right|\right) -2f_\bb \left( \left| {\sigma_{i}}^{\downarrow,\sigma} \cap {\sigma_{i}}^{\downarrow,v^j}\right| \right)\right)\mu({\sigma_{i}}) \\
  &+ \sum_{j=1}^{m} \sum_{i=k+1}^{n} \left(f_\bb\left(\left| {\sigma_{i}}^{\downarrow,\sigma}\right|\right) -2f_\bb \left( \left| {\sigma_{i}}^{\downarrow,\sigma} \cap {\sigma_{i}}^{\downarrow,v^j}\right| \right)\right)\mu({\sigma_{i}}) \\
  &\overset{(a)}{=} \sum_{j=1}^{m} \sum_{i=1}^{k}\left(f_\bb\left(\left| {\sigma_{i}}^{\downarrow,\sigma}\right|\right) -2f_\bb \left( \left| {\sigma_{i}}^{\downarrow,\sigma} \cap {\sigma_{i}}^{\downarrow,v^j}\right| \right)\right)\mu({\sigma_{i}}) \\
  &+ \sum_{j=1}^{m} \sum_{i=k+1}^{n} \left(f_{\bb_{2:n-k}}\left(\left| {\sigma_{i}}^{\downarrow,\sigma}\right|\right) -2f_{\bb_{2:n-k}} \left( \left| {\sigma_{i}}^{\downarrow,\sigma} \cap {\sigma_{i}}^{\downarrow,\tilde{v}^j}\right| \right)\right)\mu({\sigma_{i}}) \\
  &\overset{(b)}{\geq} \sum_{j=1}^{m} \sum_{i=1}^{k}\left(f_\bb\left(\left| {\sigma_{i}}^{\downarrow,\sigma}\right|\right) -2f_\bb \left( \left| {\sigma_{i}}^{\downarrow,\sigma} \cap {\sigma_{i}}^{\downarrow,v^j}\right| \right)\right)\mu({\sigma_{i}}) \\
  &+ \sum_{j=1}^{m} \sum_{i=k+1}^{n} \left(f_{\bb_{2:n-k}}\left(\left| {\sigma_{i}}^{\downarrow,\pi}\right|\right) -2_{\bb_{2:n-k}} \left( \left| {\sigma_{i}}^{\downarrow,\pi} \cap {\sigma_{i}}^{\downarrow,\tilde{v}^j}\right| \right)\right)\mu({\sigma_{i}}) \\
  &\overset{(c)}{=} \ddd(\sigma^{*},V)  - \sum_{j=1}^{m} \sum_{i=1}^{n}f_\bb\left(\left| {i}^{\downarrow,v^j}\right|\right)\mu(i),
  \end{aligned}
\end{equation*}

where ($a$) holds because, for $i \in [k+1,n]$, we have that $\sigma_{j} \notin {\sigma_{i}}^{\downarrow,\sigma}$, for all $j \leq k$. Therefore, for all~$i \in [k+1,n]$, we can replace $v^j$ with $\tilde{v}^j$, as ${\sigma_{i}}^{\downarrow,\sigma} \cap {\sigma_{i}}^{\downarrow,v^j} = {\sigma_{i}}^{\downarrow,\sigma} \cap {\sigma_{i}}^{\downarrow,\tilde{v}^j}$. Furthermore, we have $f_{\bb_{2:n-k}}\left(\left| {\sigma_{i}}^{\downarrow,\sigma}\right|\right) = f_{\bb}\left(\left| {\sigma_{i}}^{\downarrow,\sigma}\right|\right)$, as $\left| {\sigma_{i}}^{\downarrow,\sigma}\right| \leq n-k-1$; $(b)$ holds because $\pi$ is a minimizer of (\ref{porco dio gli indici}) and that the central term of ($15$) does not depend on $\pi$; ($c$) holds because~$\sigma^{*}$ agrees with $\sigma$ on the first $k$ positions and with $\pi$ in the remaining $n-k$, and with the same arguments underlying ($a$).
\end{proof}

\begin{remark}\label{remark for majority}
Observe that minimizing (\ref{porco dio gli indici}) is equivalent (up to appropriate relabeling) to minimize the cumulative sum of some (appropriately defined) top-difference distances. Specifically, let $A \subset [n]$, and define $V^{-A} := ((v^{j})^{-A})_{j=1}^{|V|}$. Then, let  
\begin{itemize}
\item $\hat{\bb} := \bb_{2:n-k}$
\item $z: C_{k+1:n} \to [n-k]$ any bijective map to relabel the indices
\item $\hat{\mu}_{z(i)} := \mu_{i}$, for all $i \in C_{k+1:n}$
\item $\hat{V} := \{v \in \mathbb{S}_{n-(k+1)} \mid \exists \tilde{v} \in V^{-C_{1:k}} \; \textnormal{s.t.} \; v_{z(i)} = \tilde{v}_i, \; \textnormal{for all } i \in C_{k+1:n}\}$
\end{itemize}
Then, minimizing (\ref{porco dio gli indici}) is clearly equivalent to 
\[ \min_{\sigma \in \mathbb{S}_{n-k}} \dd^{\hat{\mu}}_{\hat{\bb}}(\sigma,\hat{V}).\]
\end{remark}

\subsection{Proofs in §\ref{sect:semimetricdiscu}}
\begin{lemma}\label{lemmadelcazzo}
If $n>2$, then $\mathcal{M}^{\geq}_n\cap-\mathcal{M}^{\geq}_n= \{\mathbf{0}\}$.
\end{lemma}
\begin{proof}
Let $\mu\in\mathcal{M}^{\geq}_n\cap-\mathcal{M}^{\geq}_n$. For all distinct $i,j,k\in [n]$,
\[
\mu_i+\mu_j=0,\ \mu_i+\mu_k=0, \ \textnormal{and}\, \mu_j+\mu_k=0
\]
which yields $\mu_i=-\mu_j=\mu_k=-\mu_i$ which implies $\mu_i=0$, and hence $\mu\equiv 0$.
\end{proof}
\begin{lemma}\label{lemmadelculoslavato} Suppose $n>2$.
Given $\bb\in \R^{n-1}$,
\begin{itemize}
\item If $\bb\in (B^\geq_n)^{\mathsf c}\cap(-B^\geq_n)^{\mathsf c}$, then $\ddd$ is a semimetric iff $\mu\equiv0$.
\item If $\bb\in B^\geq_n\cap(-B^\geq_n)=R_n^+\cap R_n^-$, then $\ddd$ is a semimetric iff $\mu\in \R^n$.
\item If $\bb\in R_n^+\setminus R_n^-$, then $\ddd$ is a semimetric iff $\mu\in \mathcal M_n^\geq$.
\item If $\bb\in R_n^-\setminus R_n^+$, then $\ddd$ is a semimetric iff $\mu\in -\mathcal M_n^\geq$.
\item If $\bb\in B^\geq_n\setminus R_n^+$, then $\ddd$ is a semimetric iff $\mu\in \mathcal M_n^+$.
\item If $\bb\in(-B^\geq_n)\setminus R_n^-$, then $\ddd$ is a semimetric iff $\mu\in -\mathcal M_n^+$.
\end{itemize}
\end{lemma}
\begin{proof}
Let $\bb\in (B^\geq_n)^{\mathsf c}\cap(-B^\geq_n)^{\mathsf c}$. By Theorem \ref{axiomatic_characterization_THM}, and by the fact that $F$ is a bijective linear map, we have that $(B^\geq_n)^{\mathsf c}\cap(-B^\geq_n)^{\mathsf c}=F^{-1}((\R_+^{n-1})^{\mathsf c}\cap (\R_-^{n-1})^{\mathsf c})$. Thus, there are $a,b\in[n-1]$ such that $\phi_a>0$ and $\phi_b<0$. If $\mu$ is not identically zero, then by Lemma \ref{lemmadelcazzo} there are distinct $i,j\in[n]$ such that $\mu_i+\mu_j\neq0$. We can now find $\sigma,\pi\in\LL_n$ such that $\sigma_a=i=\pi_b,\sigma_{a+1}=j=\pi_{b+1}$ and thus
$$\min(\ddd(\sigma,\sigma t_{a,a+1}),\ddd(\pi,\pi t_{b,b+1}))=\min(\phi_a(\mu_i+\mu_j),\phi_b(\mu_i+\mu_j))<0$$
which is impossible if $\ddd$ is a semimetric. Thus, if $\bb\in (B^\geq_n)^{\mathsf c}\cap(-B^\geq_n)^{\mathsf c}$, then $\mu\equiv0$.\\
If $\bb\in B_n^\geq\cap (-B_n^\geq)=F^{-1}(\R_+^{n-1}\cap \R_-^{n-1})=F^{-1}(\{\mathbf0\})=\mathbf\{0\}$, then $\ddd\equiv0$ for all $\mu\in \R^n$, and hence it is always a semimetric.\\
In all the remaining case we always have that $\bb\neq\mathbf0$. On the one hand if $\bb\in R_+\setminus R_-\sqcup B^\geq_n\setminus R_+$, the non-negativity of $\ddd$ requires $\mu_i+\mu_j\geq0$ for all distinct $i,j\in [n]$, that is $\mu\in \mathcal{M}^{\geq}_n$. Indeed, since $\bb\neq \mathbf{0}$, there exists $a\in [n-1]$ such that $\phi_a>0$, by Theorem \ref{axiomatic_characterization_THM}. For all distinct $i,j\in [n]$, there exists $\sigma\in \LL_n$ such that $\sigma_a=i,\sigma_{a+1}=j$, and hence
\[
\ddd(\sigma,\sigma t_{a,a+1})=(\mu(\sigma_a)+\mu(\sigma_{a+1}))\phi_a.
\]
On the other hand, by an analogous argument, if $\bb\in R_-\setminus R_+\sqcup (-B^\geq_n)\setminus R_-$, the non-negativity of $\ddd$ requires $\mu_i+\mu_j\leq0$ for all distinct $i,j\in [n]$, that is $\mu\in -\mathcal{M}^+_n$. 
Therefore in the last four cases of the Lemma, $\ddd$ is non-negative if and only if $\mu\in \mathcal{M}^{\geq}_n$. From now on we assume that $\mu\in \mathcal{M}^{\geq}_n$. Next we consider triangle inequality. Since $n>2$, we have that for all $\sigma,\omega,\pi\in \LL_n$,
\begin{align*}
&\ddd(\sigma,\omega)+\ddd(\omega,\pi)-\ddd(\sigma,\pi)\\
    &=\sum_{A\in\mathcal P_n}\beta_{|A|}\left(\mu(\bigtriangleup_A(\sigma,\omega))+\mu(\bigtriangleup_A(\omega,\pi))-\mu(\bigtriangleup_A(\sigma,\pi))\right)\\
    &=\sum_{A\in\mathcal P_n}\beta_{|A|}\begin{cases}
        2\mu(M(A,\omega))&\textrm{ if }|\{M(A,\sigma),M(A,\omega),M(A,\pi)\}|=3\\
        2\left(\mu(M(A,\omega))+\mu(M(A,\sigma))\right)&\textrm{ if }M(A,\sigma)=M(A,\pi)\neq M(A,\omega)\\
        0&\textrm{ otherwise. }
    \end{cases}
\end{align*}
If $\bb\in R_n^+\setminus R^-_n\sqcup R_n^-\setminus R^+_n$, then any non-negative $\ddd$ satisfies triangle inequality because condition $|\{M(A,\sigma),M(A,\omega),M(A,\pi)\}|=3$ cannot be satisfied by sets $A$ of cardinality $2$. Therefore, the third and fourth bullet points of the lemma are fullfilled. 
\\
Assume now that $\bb\in B_n^\geq\setminus  R^+_n\sqcup (-B_n^\geq)\setminus R^-_n$. Let $\sigma\in \LL_n$ and $\pi$ be its antipodal permutation. Then, for any $A\in\mathcal P_n$, we have $M(A,\sigma)\neq M(A,\pi)$, thus $$\left(\mu(\bigtriangleup_A(\sigma,\omega))+\mu(\bigtriangleup_A(\omega,\pi))-\mu(\bigtriangleup_A(\sigma,\pi))\right)\neq 0$$ iff $M(A,\sigma), M(A,\pi)$, and $M(A,\omega)$ are all distinct. Let $a\in [2,n]$. If $\omega=\sigma t_{a-1,a}$, then $M(A,\sigma)=M(A,\omega)$ for all $A\in \mathcal{P}_n$, except those $A$ containing both $\sigma_{a-1}$ and $\sigma_{a}$ but no $\sigma_{j}$ for $j<a-1$. 
Overall we have $\binom{n-a}{k-2}$ sets of size $k$ with this property. Moreover, for all these sets we have that $M(A,\omega)\neq M(A,\pi)$ except for the case in which $A=\{\sigma_{a-1},\sigma_a\}$. Thus, we have that
\begin{equation}\label{eq:dioassassinostronzoinculatodaunarinoceronte}
\ddd(\sigma,\omega)+\ddd(\omega,\pi)-\ddd(\sigma,\pi)=2\mu(\sigma_{a})\sum_{k=3}^{n+2-a}\binom{n-a}{k-2}\beta_k.
\end{equation}
If $\bb\in B^{\geq}_n\setminus R^+_n$, $\dd_\bb$ is a semimetric and thus, by triangle inequality,
$$
\dd_\bb(\sigma,\omega)+\dd_\bb(\omega,\pi)-\dd_\bb(\sigma,\pi)=2\sum_{k=3}^{n+2-a}\binom{n-a}{k-2}\beta_k\geq 0
$$
and hence $\sum_{k=3}^{n+2-a}\binom{n-a}{k-2}\beta_k\geq 0$. Notice that the map $(\beta_k)_{k=3}^n\mapsto \left(\sum_{k=3}^{n+2-a}\binom{n-a}{k-2}\beta_k\right)_{a=2}^{n-1}$ is bijective and thus the quantities $\sum_{k=3}^{n+2-a}\binom{n-a}{k-2}\beta_k$ are all zero iff $(\beta_k)_{k=3}^n\equiv 0$. However, this is impossible since we are assuming $\bb\in B_n^\geq\setminus R^+_n=B_n^\geq\setminus (R^+_n\cup R^-_n)$. Thus, there is some $a\in[2,n-1]$ such that $\sum_{k=3}^{n+2-a}\binom{n-a}{k-2}\beta_k>0$.
\\
Specularly, since $\bb\in (-B^{\geq}_n)\setminus R^-_n$, then $-\bb\in B^{\geq}_n\setminus R^+_n$ and $\dd_{-\bb}=-\dd_\bb$ is a semimetric, we conclude that there is some $a\in[2,n-1]$ such that
$\sum_{k=3}^{n+2-a}\binom{n-a}{k-2}\beta_k<0$. \\
In conclusion, if $\bb\in B^{\geq}_n\setminus R^+_n $ (or $(-B^{\geq}_n)\setminus R^-_n$), by \eqref{eq:dioassassinostronzoinculatodaunarinoceronte}, and the arbitrariety of $\sigma$, we have that $\ddd$ satisfies the triangle inequality if and only if $\mu\in \mathcal{M}^+_n$ (or $\mathcal{M}^+_n$).
\end{proof}
\color{black}
\begin{lemma}\label{dioassassinodicazzullo}
   Suppose $n>2$. Given $\mu\in\R^n$,
    \begin{itemize}
        \item If $\mu\in (\mathcal{M}^{\geq}_n)^{\mathsf c}\cap(-\mathcal{M}^{\geq}_n)^{\mathsf c}$, then $\ddd$ is a semimetric iff $\bb=\mathbf0$.
        \item If $\mu\in \mathcal{M}^{\geq}_n\cap(-\mathcal{M}^{\geq}_n)$, then $\ddd$ is a semimetric iff $\bb\in\R^{n-1}$.
        \item If $\mu\in \mathcal{M}^{+}_n\setminus(-\mathcal{M}^{+}_n)$, then $\ddd$ is a semimetric iff $\bb\in B^\geq_n$.
        \item If $\mu\in (-\mathcal{M}^{+}_n)\setminus\mathcal{M}^{+}_n$, then $\ddd$ is a semimetric iff $\bb\in -B^\geq_n$.
        \item If $\mu\in \mathcal{M}^{\geq}_n\setminus\mathcal{M}^{+}_n$, then $\ddd$ is a semimetric iff $\bb\in R^+_n$.
        \item If $\mu\in (-\mathcal{M}^{\geq}_n)\setminus(-\mathcal{M}^{+}_n)$, then $\ddd$ is a semimetric iff $\bb\in R^-_n$.
    \end{itemize}
\end{lemma}
\begin{proof}
Suppose that $\ddd$ is a semimetric. If $\mu\notin (\mathcal{M}^{\geq}_n\cup-\mathcal{M}^{\geq}_n)$, then there exist distinct $i,j\in [n]$ such that $\mu_i+\mu_j<0$ and distinct $k,\ell\in [n]$ such that $\mu_k+\mu_\ell>0$. Let now $a\in\argmax_{c\in[n-1]}\phi_c$ and $b\in\argmin_{c\in[n-1]}\phi_c$.
Let now $\sigma\in \LL_n$ be such that $\sigma_a=i$ and $\sigma_{a+1}=j$. Then,
\[
\ddd(\sigma,\sigma t_{a,a+1})=\phi_a(\mu_i+\mu_j)
\]
is non-negative if and only if $\phi_a\leq 0$. Similarly, if $\sigma\in \LL_n$ is such that $\sigma_b=k$ and $\sigma_{b+1}=\ell$. Then,
$$
\ddd(\sigma,\sigma t_{b,b+1})=\phi_b(\mu_k+\mu_\ell)
$$
is non-negative iff $\phi_b\geq 0$. Thus, if $\mu\notin (\mathcal{M}^{\geq}_n\cup-\mathcal{M}^{\geq}_n)$, then $\phi\equiv0$ and, by Theorem \ref{axiomatic_characterization_THM}, $\bb=\mathbf{0}$.\\ 
Suppose now that $\mu\in\mathcal{M}^{\geq}_n\cap-\mathcal{M}^{\geq}_n$, then, by Lemma \ref{lemmadelcazzo}, $\mu\equiv 0$ and thus any $\bb\in\R^{n-1}$ will make $\ddd$ the trivial semimetric.\\
If $\mu\in\mathcal M^{\geq}_n\setminus\{\mathbf 0\}=\mathcal M^+_n\setminus(-\mathcal M^+_n)\sqcup\mathcal M^\geq_n\setminus\mathcal M^+_n$, then for any distinct $i,j\in[n]$, $\mu_i+\mu_j\geq0$ and there are some distinct $i,j\in [n]$ such that $\mu_i+\mu_j>0$. Now, for any $a\in[n-1]$, there is some $\sigma\in\LL_n$ such that $\sigma_a=i$ and $\sigma_{a+1}=j$. Therefore,
$$\ddd(\sigma,\sigma t_{a,a+1})=\phi_a(\mu_i+\mu_j)\geq0,$$
implies that $\phi_a\geq0$ for any $a\in[n-1]$ and thus, $\bb$ must be in $B^\geq_n$.
Thus, we have shown that in cases 3 and 5, if $\ddd$ is a semimetric, then $\bb$ must be in $B^\geq_n$. Applying Lemma \ref{lemmadelculoslavato}, we conclude that this is sufficient in case 3, and that under case 5, $\ddd$ is a semimetric iff $\bb\in R^+_n$. \\
Similarly, if $\mu\in(-\mathcal M^{\geq}_n)\setminus\{\mathbf 0\}=(-\mathcal M^{+}_n)\setminus\mathcal M^{+}_n\sqcup (-\mathcal M^{\geq}_n)\setminus(-\mathcal M^{+}_n)$, we conclude that under cases 4 and 6, if $\ddd$ is a semimetric, then $\bb$ must be in $-B^\geq_n$. By Lemma \ref{lemmadelculoslavato}, we complete these two cases.
\end{proof}

\begin{proof}[Proof of Proposition \ref{prop:semimetric}]
Suppose first that $n=2$. If $\ddd$ is a semimetric, then we have three cases. If $\beta_2>0$, then we must have that $\mu_1+\mu_2\geq 0$, and hence $(\bb,\mu)\in R_2^{++}\tim \mathcal{M}^{\geq}_2$. If $\beta_2<0$, then $\mu_1+\mu_2\leq 0$, and hence $(\bb,\mu)\in R^{--}_2\tim (-\mathcal{M}^{\geq}_2)$. If $\beta_2=0$, then $\ddd\equiv 0$ for all $\mu\in \R^2$. The converse inclusion readily follows.
\par\medskip
Suppose $n>2$. The claim follows by Lemmas \ref{lemmadelculoslavato} and \ref{dioassassinodicazzullo}.
\end{proof}

\subsection{Proofs in §\ref{properties of the metric}}
\begin{proof}[Proof of Proposition \ref{hoilmercedes}]
One implication is straightforward. Suppose that $\ddd=\dd_\ga^\mu$ for some $\bb,\ga\in \mathbb{R}^{n-1}$. Then, for all~$a\in[n-1]$ and $\sigma\in\LL_n$,
$$
\phi^\bb_a(\mu(\sigma_a)+\mu(\sigma_{a+1}))=\ddd(\sigma,\sigma t_{a,a+1})=\dd_\ga^\mu(\sigma,\sigma t_{a,a+1})=\phi^\ga_a(\mu(\sigma_{a})+\mu(\sigma_{a+1})).
$$
Since, for every $a\in[n-1]$ there is some $\sigma$ such that $\mu(\sigma_a)+\mu(\sigma_{a+1})\neq0$, we conclude that $\phi_a^\bb=\phi_a^\ga$ for all $a\in[n-1]$. By Theorem \ref{scorpions} this implies that $\bb=\ga$.
\end{proof}

\begin{proof}[Proof of Proposition \ref{bigblackcockcomepiaceadandreino}]
  One implication is straightforward. For any $a\in[n-1]$,
  $$\phi_a(\mu(\sigma_a)+\mu(\sigma_{a+1}))=\ddd(\sigma,\sigma t_{a,a+1})=\dd^\nu_\bb(\sigma,\sigma t_{a,a+1})=\phi_a(\nu(\sigma_a)+\nu(\sigma_{a+1})).$$
  Since $\dd_{\bb}$ or $\dd_{-\bb}$ is a metric, $|\phi_a|=|\dd_{\bb}(\sigma,\sigma t_{a,a+1})|/2>0$ for all $a\in [n-1]$. Therefore $\mu(i)+\mu(j)=\nu(i)+\nu(j)$ for any distinct $i,j\in[n]$. Thus, in particular, for any distinct~$i,j,k\in[n]$,
  \begin{align*}
    \mu(i)+\mu(j)=\nu(i)+\nu(j),\,\,\,
    \mu(i)+\mu(k)=\nu(i)+\nu(k),\,\,\,
    \mu(j)+\mu(k)=\nu(j)+\nu(k)
  \end{align*}
  Summing the first two equations and subtracting the last, we get that $\mu(i)=\nu(i)$. By the arbitrariness of $i,j,k$, the proof is concluded.
\end{proof}


\begin{proof}[Proof of Proposition \ref{Galletta}]
  Let $\ddd$ be neutral, then for any $a\in[n-2]$, we have that
  \begin{equation}\label{Crisostomo}
  \ddd(\Id,t_{a+1,a+2})=\ddd(t_{a,a+1},t_{a,a+1}t_{a+1,a+2})=\ddd(t_{a,a+2},t_{a,a+2}t_{a+1,a+2})
  \end{equation}
  But
  \begin{align*}
    \ddd(\Id,t_{a+1,a+2})&=\phi_{a+1}(\mu(a+1)+\mu(a+2))\\
    \ddd(t_{a,a+1},t_{a,a+1}t_{a+1,a+2})&=\phi_{a+1}(\mu(a)+\mu(a+2))\\
    \ddd(t_{a,a+2},t_{a,a+2}t_{a+1,a+2})&=\phi_{a+1}(\mu(a)+\mu(a+1))
  \end{align*}
  Since $\dd_\bb$ or $\dd_{-\bb}$ is a metric, we have that $|\phi_{a+1}|=|\dd_\bb(\Id,t_{a+1,a+2})|/2>0$. Therefore, by \eqref{Crisostomo} it follows that  $\mu(a)=\mu(a+1)=\mu(a+2)$. By the arbitrariness of $a$, we obtain that $\mu$ is constant. To show the converse, fix $S\subseteq [n]$, $\sigma,\pi,\omega\in \LL_n$. We have that, $M(\sigma(S), \sigma\pi)=\sigma(M(S,\pi))$. Indeed, 
  \begin{align*}
    M(\sigma(S),\sigma\pi)&=\{x\in \sigma(S): 
 \not\exists y\in\sigma(S) \textrm{ s.t. } y>_{\sigma\pi} x\}\\
 &=\{\sigma_z\in \sigma(S):\not\exists \sigma_w\in \sigma(S) \textrm{ s.t. } \sigma_w>_{\sigma\pi} \sigma_z \}\\
 &=\{\sigma_z\in \sigma(S):\not\exists w\in S \textrm{ s.t. } w>_{\pi} z \}\\
 &=\sigma\left(\{z\in S:\not\exists w\in S \textrm{ s.t. } w>_{\pi} z \}\right)\\
 &=\sigma\left(M(S,\pi)\right).
  \end{align*}
  Thus, 
  \begin{align*}
  \bigtriangleup_{\sigma(S)}(\sigma\pi,\sigma\omega)&= M(\sigma(S),\sigma\pi)\bigtriangleup M(\sigma(S),\sigma\omega)
  =\sigma(M(S,\pi))\bigtriangleup\sigma(M(S,\omega))\\
  &=\sigma(M(S,\pi)\bigtriangleup M(S,\omega))
  =\sigma(\bigtriangleup_{S}(\pi,\omega)).
  \end{align*}
  This implies that
  \begin{align*}
  \dd_\bb(\sigma\pi,\sigma\omega)&=\sum_{S\in \mathcal{P}_n}\beta_{|S|}|\bigtriangleup_S(\sigma\pi,\sigma\omega)|=\sum_{S\in \mathcal{P}_n}\beta_{|S|}|\bigtriangleup_{\sigma(\sigma^{-1}(S))}(\sigma\pi,\sigma\omega)|\\
  &=\sum_{S\in \mathcal{P}_n}\beta_{|\sigma^{-1}(S)|}|\sigma(\bigtriangleup_{\sigma^{-1}(S)}(\sigma,\pi))|=\sum_{S'\in \mathcal{P}_n}\beta_{|S'|}|\bigtriangleup_{S'}(\pi,\omega)|=\dd_\bb(\pi,\omega).
  \end{align*}
\end{proof}

\subsection{Proofs in §\ref{Sect_Main_Axiomatiz}}
The following two auxiliary lemmas are immediate observations.
\begin{lemma}\label{DioPernice}
For all $\sigma,\pi\in \mathbb{S}_n$ and $\mathbf{t}\in T(\sigma,\pi)$,
  $$\{\{i_k(\mathbf t),j_k(\mathbf t)\}:k\in[\dd_K(\sigma,\pi)]\}=\{\{i,j\}:i<_\sigma j\textrm{ and }i>_\pi j\}.$$
\end{lemma}
\begin{lemma}\label{DioPavone}
  For all distinct $i,j\in[n]$, the number of couples $(\sigma,\pi)\in \mathbb{S}_n^2$ with $I(\sigma,\pi)=\{(i,j)\}$ is $(n-1)!$. Moreover, $|\sigma^{-1}_j-\sigma^{-1}_i|=1$ and $\pi=\sigma t_{i,j}= t_{\sigma^{-1}_i,\sigma^{-1}_j}\sigma$.
\end{lemma}

The following lemma is essentially due to \cite{Hassandez} (Lemma 2) and it can be proved analogously.
\begin{lemma}\label{lemmaesplosivo}
For all $\mathbf t\in T(\sigma,\pi)$ we have that 
$$\left(\sigma\prod_{j=1}^kt_{a_j(\mathbf t),a_j(\mathbf t)+1}\right)_{k=0}^{\dd_K(\sigma,\pi)}$$
are on a line.
\end{lemma}

\begin{lemma}\label{andrealosveglio}
  Let $\dd$ be a semimetric satisfying axiom \ref{A1}.
  For any $\sigma,\sigma',\pi,\pi'\in\mathbb S_n$ with, $I(\sigma,\sigma')=I(\pi,\pi')=\{(i,j)\}$ we have that $\dd(\sigma,\sigma')=\dd(\pi,\pi')$.
\end{lemma}
\begin{proof}
If $n\leq 2$, the claim is trivial. Without loss of generality, by Lemma \ref{DioPavone}, there are $a,b\in[n-1]$ such that 
\[\sigma_a=i,\ \sigma_{a+1}=j,\ \textnormal{and}\ \pi_b=i, \pi_{b+1}=j.\]
\noindent First assume that $n>3$ and that $\pi=t_{k,l}\sigma$, 
for some distinct $k,l$ both different from $i,j$. Then, we have the following betweenness relations
    $$\begin{matrix}
      &&\sigma&&\\
      &\color{red}{\diagup} &&\color{blue}{\diagdown} &\\
      t_{i,j}\sigma&&&&t_{k,l}\sigma\\
      &\color{blue}{\diagdown} &&\color{red}{\diagup} &\\
      &&t_{i,j}t_{k,l}\sigma&&
    \end{matrix}$$
    Thus, by Axiom \ref{A1} we have that,
    \begin{align*}
    \dd(\sigma,t_{i,j}t_{k,l}\sigma)&=\dd(\sigma,t_{i,j}\sigma)+\dd(t_{i,j}\sigma,t_{i,j}t_{k,l}\sigma)=\dd(\sigma,t_{k,l}\sigma)+\dd(t_{k,l}\sigma,t_{i,j}t_{k,l}\sigma),\\
    \dd(t_{i,j}\sigma,t_{k,l}\sigma)&=\dd(t_{i,j}\sigma,\sigma)+\dd(\sigma,t_{k,l}\sigma)=\dd(t_{i,j}\sigma,t_{i,j}t_{k,l}\sigma)+\dd(t_{i,j}t_{k,l}\sigma,t_{k,l}\sigma).
    \end{align*}
    This system of equations implies that $\dd(\sigma,t_{i,j}\sigma)=\dd(\pi,t_{i,j}\pi)$.\\ Iteratively, we have that if $\pi=\left(\coprod_{k=1}^{m}t_{i_k,j_k}\right)\sigma$ where $i_k,j_k\notin\{i,j\}$ for all $k\in[m]$, then
    $$\dd(\sigma,t_{i,j}\sigma)=\dd(t_{i_1,j_1}\sigma,t_{i,j}t_{i_1,j_1}\sigma)=\ldots=\dd(\pi,t_{i,j}\pi).$$
    This shows that if $\pi$ is a permutation with $\sigma_a=\pi_{a}=i$ and $\sigma_{a+1}=\pi_{a+1}=j$, then $\dd(\sigma,t_{i,j}\sigma)=\dd(\pi,t_{i,j}\pi)$. That is, the claim holds with $a=b$. Indeed, in such case $\pi$ can be obtained transforming $\sigma$ by finitely many transpositions $(t_{i_k,j_k})_{k=1}^m$ with $i_k,j_k\notin \left\lbrace i,j\right\rbrace$.\\

Now we consider $a\neq b$ and allow for $n\geq 3$. We start assuming that $b=a+1$. Let $i=\sigma_a$, $j=\sigma_{a+1}$, and $k=\sigma_{a+2}$. From the definition of betweenness, we have that 
$$\begin{matrix}
  &&\underset{ijk}{\sigma}&&\\
  &\color{red}{\diagup}&&\color{blue}{\diagdown}\\
  \underset{jik}{t_{ij}\sigma}&&&&\underset{ikj}{t_{jk}\sigma}\\
  \color{teal}|&&&&\color{teal}|\\
  \underset{jki}{t_{ik}t_{ij}\sigma}&&&&\underset{kij}{t_{ik}t_{jk}\sigma}\\
  &\color{blue}{\diagdown}&&\color{red}{\diagup}\\
  &&\underset{kji}{t_{ij}t_{ik}t_{jk}\sigma}&&
\end{matrix}$$
Thus, by Axiom \ref{A1} we have that,
\begin{align*}
  \dd(\sigma,t_{ij}t_{ik}t_{jk}\sigma)&=\dd(\sigma,t_{jk}\sigma)+\dd(t_{jk}\sigma,t_{ik}t_{jk}\sigma)+\dd(t_{ik}t_{jk}\sigma,t_{ij}t_{ik}t_{jk}\sigma)\\
  &=\dd(t_{ik}t_{ij}\sigma,t_{ij}t_{ik}t_{jk}\sigma)+\dd(t_{ik}t_{ij}\sigma,t_{ij}\sigma)+\dd(\sigma,t_{ij}\sigma),\\
  \dd(t_{ij}\sigma,t_{ik}t_{jk}\sigma)&=\dd(t_{ij}\sigma,\sigma)+\dd(\sigma,t_{jk}\sigma)+\dd(t_{jk}\sigma,t_{ik}t_{jk}\sigma)\\
  &=\dd(t_{ij}\sigma,t_{ik}t_{ij}\sigma)+\dd(t_{ik}t_{ij}\sigma,t_{ij}t_{ik}t_{jk}\sigma)+\dd(t_{ij}t_{ik}t_{jk}\sigma,t_{ik}t_{jk}\sigma),\\
  \dd(t_{jk}\sigma,t_{ik}t_{ij}\sigma)&=\dd(t_{jk}\sigma,\sigma)+\dd(\sigma,t_{ij}\sigma)+\dd(t_{ij}\sigma,t_{ik}t_{ij}\sigma)\\
  &=\dd(t_{jk}\sigma,t_{ik}t_{jk}\sigma)+\dd(t_{ik}t_{jk}\sigma,t_{ij}t_{ik}t_{jk}\sigma)+\dd(t_{ij}t_{ik}t_{jk}\sigma,t_{ik}t_{ij}\sigma).
\end{align*}
From this linear system, we conclude that
$$\dd(\sigma,t_{ij}\sigma)=\dd(t_{ik}t_{jk}\sigma,t_{ij}t_{ik}t_{jk}\sigma).$$
If $n=3$, then this equality yields the claim. Instead, whenever $n>3$, if $\sigma$ is any permutation with $\sigma_a=i$, $\sigma_{a+1}=j$, and if $\pi$ is any permutation with $\pi_b=i$, $\pi_{b+1}=j$, then by "hexagonal moves" we can obtain recursively $\tau$ such that $\tau_b=i$ and $\tau_{b+1}=j$ and such that~$\dd(\sigma,t_{ij}\sigma)=\dd(\tau,t_{ij}\tau)$.\footnote{More specifically assuming wlog that $b>a$, we can take
$$\tau=\left(\coprod_{s=1}^{b-a}t_{i,\sigma_{a+1+s}}t_{j,\sigma_{a+1+s}}\right)\sigma$$
which has
$\tau_b=i$ and $\tau_{b+1}=j$.} Now, by the previous argument (case $a=b$ above), we have that $\dd(\tau,t_{ij}\tau)=\dd(\pi,t_{ij}\pi)$.
\end{proof}

\begin{proof}[Proof of Theorem \ref{ronniecolemanleggero}]
For all distinct $i,j\in \left[n\right]$, let $g(i,j):=\dd(\sigma,\pi)$ for any $\sigma,\pi$ with $I(\sigma,\pi)=\{(i,j)\}$. We let $g(i,i):=0$ otherwise. By Lemma \ref{andrealosveglio}, this $g:[n]^2\to \mathbb{R}_{+}$ is a well-defined symmetric function. Finally, by Axiom \ref{A1}, Lemmas \ref{DioPernice}, and \ref{lemmaesplosivo}, we obtain that for all $\sigma,\pi\in \mathbb S_n$, and $\mathbf t\in T(\sigma,\pi)$, 
$$\dd(\sigma,\pi)=\sum_{k=1}^{\dd_K(\sigma,\pi)}g(i_k(\mathbf t),j_k(\mathbf t))=\sum_{(i,j)\in I(\sigma,\pi)}g(i,j).$$
Conversely, suppose that $\sigma-\omega-\pi$ for some $\sigma,\omega,\pi\in \LL_n$. Then, there is a $\mathbf t\in T(\sigma,\pi)$ and a $\ell\leq \dd_K(\sigma,\pi)$ such that $\omega=\sigma\prod_{j=1}^\ell t_{a_j(\mathbf{t}),a_j(\mathbf{t})+1}$ and $\pi=\omega\prod_{j=\ell+1}^{\dd_K(\sigma,\pi)}t_{a_j(\mathbf{t}),a_j(\mathbf{t})+1}$. Thus, 
\small$$\dd(\sigma,\pi)=\sum_{k=1}^{\dd_K(\sigma,\pi)}g(i_k(\mathbf t),j_k(\mathbf t))=\sum_{k=1}^{\ell}g(i_k(\mathbf t),j_k(\mathbf t))+\sum_{k=\ell+1}^{\dd_K(\sigma,\pi)}g(i_k(\mathbf t),j_k(\mathbf t))=\dd(\sigma,\omega)+\dd(\omega,\pi).$$\normalsize
\end{proof}
\begin{proof}[Expansion on Remark \ref{remarkmetrico}]
By Theorem \ref{ronniecolemanleggero} it is sufficient to show that if 
$\dd(\sigma,\omega)+\dd(\omega,\pi)=\dd(\sigma,\pi)$ for some $\sigma,\omega,\pi\in \LL_n$, then $\sigma-\omega-\pi$. Given that,
$$\dd(\sigma,\omega)=\sum_{(i,j)\in I(\sigma,\omega)} g(i,j)\,\,\,\,\,\,\,\,\,\,\,\,\,\,\dd(\omega,\pi)=\sum_{(i,j)\in I(\omega,\pi)} g(i,j)$$
were $\omega$ not in between $\sigma$ and $\pi$, we would have $I(\sigma,\pi)\subsetneq I(\sigma,\omega) \cup I(\omega,\pi)$.
But this is impossible, since $g$ only takes positive values.
\end{proof}

\begin{proof}[Proof of Theorem \ref{roar}]
By Theorem \ref{ronniecolemanleggero} it is sufficient to show that $\dd$ satisfies Axiom \ref{A2} if and only if there exists a $\mu\in \mathcal{M}^{\geq}_n$ such that $g(i,j)=\mu(i)+\mu(j)$ for all distinct $i,j\in [n]$. If $n<4$, then Axiom \ref{A2} holds vacuously and indeed $g(i,j)$ can be written as a sum. In particular, if $n=3$, taking \small$$\mu(1)=\frac{g(1,2)+g(1,3)-g(2,3)}2\,\,\,\,\,\,\mu(2)=\frac{g(1,2)-g(1,3)+g(2,3)}2\,\,\,\,\,\,\mu(3)=\frac{-g(1,2)+g(1,3)+g(2,3)}2.$$\normalsize
we have that $g(i,j)=\mu(i)+\mu(j)$ for all $i,j\in [n]$. If $n\geq 4$, we define $\mu(i):=(g(i,j)-g(k,j)+g(i,k))/2$ for some distinct $j,k$ both different from $i$. Under Axiom \ref{A2}, for all $i,j,k,\ell$ with $i\neq j$, $k\neq\ell$, $i\neq k$, and $j\neq \ell$,
$$g(i,j)+g(k,\ell)=g(i,k)+g(j,\ell).$$
This implies that $\mu$ is well defined. Indeed, if $k',j'$ are distinct and different from $i$, then
  \begin{align*}
  &\left(g(i,j)-g(k,j)+g(i,k)\right)-\left(g(i,j')-g(k',j')+g(i,k')\right)\\
&=\left(g(i,j)+g(k',j')\right)-\left(g(i,j')+g(k,j)\right)+g(i,k)-g(i,k')\\  &=\left(g(i,k')+g(j,j')\right)-\left(g(i,k)+g(j,j')\right)+g(i,k)-g(i,k')=0
  \end{align*}
  Moreover, we have that for all $i\neq j$ \begin{align*}
g(i,j)&=\frac12\left(g(i,j)-g(k,j)+g(i,k)\right)+\frac12\left(g(j,i)-g(k,i)+g(j,k)\right)=\mu(i)+\mu(j).
  \end{align*} 
  The converse is straightforward.
\end{proof}

\begin{proposition}\label{pazzia}
  A semimetric $\dd$ satisfies Axiom \ref{A3} if and only if, for all $\sigma,\pi\in\LL_n$,
  \begin{equation}\label{graph}
  \dd(\sigma,\pi)
  =\min_{\mathbf t\in T(\sigma,\pi)}\sum_{j=1}^{\dd_K(\sigma,\pi)}\dd\left(\sigma\prod_{k=1}^{j-1}t_{a_k(\mathbf t),a_k(\mathbf t)+1},\sigma\prod_{k=1}^{j}t_{a_k(\mathbf t),a_k(\mathbf t)+1}\right).
  \end{equation}
\end{proposition}
\begin{proof}
  If $\sigma,\pi$ have Kendall distance lower or equal than one, the claim is trivial. Suppose that the claim holds for all couples $\sigma,\pi$ with $\dd_K(\sigma,\pi)<m$ and let $\sigma,\pi$ have distance $m$. By the axiom \ref{A3} there is some $\omega$ lying between $\sigma,\pi$ such that $\dd(\sigma,\omega)+\dd(\omega,\pi)=\dd(\sigma,\pi)$. But clearly,~$\dd_K(\sigma,\omega),\dd_K(\omega,\pi)<m$, thus by inductive hypothesis there are some $\mathbf t'\in T(\sigma,\omega)$ and~$\mathbf t'' \in T(\omega,\pi)$, such that
  \begin{align*}
  \ddd(\sigma,\omega)&=\sum_{k=1}^{\dd_K(\sigma,\omega)}\dd\left(\sigma\prod_{k=1}^{j-1}t_{a_k(\mathbf t'),a_k(\mathbf t')+1},\sigma\prod_{k=1}^{j}t_{a_k(\mathbf t'),a_k(\mathbf t')+1}\right)\\
  \ddd(\omega,\pi)&=\sum_{k=1}^{\dd_K(\omega,\pi)}\dd\left(\sigma\prod_{k=1}^{j-1}t_{a_k(\mathbf t''),a_k(\mathbf t'')+1},\sigma\prod_{k=1}^{j}t_{a_k(\mathbf t''),a_k(\mathbf t'')+1}\right).
  \end{align*}
  Denote by $\mathbf{\tilde t}$ the concatenation of $\mathbf t'$ and $\mathbf t''$. Since $\dd_K$ is a metric and $\sigma-\omega-\pi$, we have that $\tilde{t}\in T(\sigma,\pi)$. Therefore,
  \begin{align*}
  \ddd(\sigma,\pi)&=\sum_{k=1}^{\dd_K(\sigma,\pi)}\dd\left(\sigma\prod_{k=1}^{j-1}t_{a_k(\mathbf{\tilde t}),a_k(\mathbf{\tilde t})+1},\sigma\prod_{k=1}^{j}t_{a_k(\mathbf{\tilde t}),a_k(\mathbf{\tilde t})+1}\right)\\
  &\geq \min_{\mathbf t\in T(\sigma,\pi)}\sum_{k=1}^{\dd_K(\sigma,\pi)}\dd\left(\sigma\prod_{k=1}^{j-1}t_{a_k(\mathbf{ t}),a_k(\mathbf{ t})+1},\sigma\prod_{k=1}^{j}t_{a_k(\mathbf{t}),a_k(\mathbf{ t})+1}\right).
  \end{align*}
  The reverse inequality follows from triangle inequality. Conversely, if \eqref{graph} holds and $$\mathbf t^*\in\argmin_{\mathbf t\in T(\sigma,\pi)}\sum_{k=1}^{\dd_K(\sigma,\pi)}\dd\left(\sigma\prod_{k=1}^{j-1}t_{a_k(\mathbf{ t}),a_k(\mathbf{ t})+1},\sigma\prod_{k=1}^{j}t_{a_k(\mathbf{t}),a_k(\mathbf{ t})+1}\right),$$
  then for any $j\in\{1,\ldots,\dd_K(\sigma,\pi)\}$ we have that by Lemma \ref{lemmaesplosivo}
  $$\omega^j:=\sigma\prod_{k=1}^{j}t_{a_k(\mathbf{t^*}),a_k(\mathbf{t}^*)+1}$$
  lies between $\sigma$ and $\pi$. We need to prove that $\dd(\sigma,\omega^j)+\dd(\omega^j,\pi)=\dd(\sigma,\pi)$. To prove this, we need to show that $\mathbf t^{'j}:= \mathbf{t}^*|_{\{1,\ldots,j\}}$ and $\mathbf t^{''j}:=\mathbf t^*|_{\{j+1,\ldots,\dd_K(\sigma,\pi)\}}$ are such that
  \begin{align*}
  \mathbf t^{'j}&\in\argmin_{t\in T(\sigma,\omega^j)}\sum_{k=1}^{j}\dd\left(\sigma\prod_{k=1}^{j-1}t_{a_k(\mathbf{ t}),a_k(\mathbf{ t})+1},\sigma\prod_{k=1}^{j}t_{a_k(\mathbf{t}),a_k(\mathbf{ t})+1}\right),\\
  \mathbf{t^{''}}^j&\in\argmin_{t\in T(\omega^j,\pi)}\sum_{k=1}^{\dd_K(\sigma,\pi)-j}\dd\left(\sigma\prod_{k=1}^{j-1}t_{a_k(\mathbf{ t}),a_k(\mathbf{ t})+1},\sigma\prod_{k=1}^{j}t_{a_k(\mathbf{t}),a_k(\mathbf{ t})+1}\right).
  \end{align*}
  Indeed, assume they are not and $\mathbf u^{'j}$ is such that
  \small$$\sum_{k=1}^{j}\dd\left(\sigma\prod_{k=1}^{j-1}t_{a_k(\mathbf{ u'}^j),a_k(\mathbf{ u'}^j)+1},\sigma\prod_{k=1}^{j}t_{a_k(\mathbf{u'}^j),a_k(\mathbf{ u'}^j)+1}\right)< \sum_{k=1}^{j}\dd\left(\sigma\prod_{k=1}^{j-1}t_{a_k(\mathbf{ t'}^j),a_k(\mathbf{ t'}^j)+1},\sigma\prod_{k=1}^{j}t_{a_k(\mathbf{t'}^j),a_k(\mathbf{ t'}^j)+1}\right),$$\normalsize
  then by concatenating $\mathbf u'^j$ with $\mathbf t''^j$, we would get a new element $\mathbf u^*$ of $T(\sigma,\pi)$ with
  \small$$\sum_{k=1}^{\dd_K(\sigma,\pi)}\dd\left(\sigma\prod_{k=1}^{j-1}t_{a_k(\mathbf{ u^*}),a_k(\mathbf{ u^*})+1},\sigma\prod_{k=1}^{j}t_{a_k(\mathbf{u^*}),a_k(\mathbf{ u^*})+1}\right)<\sum_{k=1}^{\dd_K(\sigma,\pi)}\dd\left(\sigma\prod_{k=1}^{j-1}t_{a_k(\mathbf{ t^*}),a_k(\mathbf{ t^*})+1},\sigma\prod_{k=1}^{j}t_{a_k(\mathbf{t^*}),a_k(\mathbf{ t^*})+1}\right)$$\normalsize
  contradicting the minimality of $\mathbf{t}^*$.
\end{proof}

\begin{proof}[Proof of Theorem \ref{main_g_thm}]
  By Proposition \ref{pazzia}, we have that Axiom \ref{A3} yields
  $$\dd:(\sigma,\pi)\mapsto\min_{\mathbf t\in T(\sigma,\pi)}\sum_{k=1}^{\dd_K(\sigma,\pi)}\dd\left(\sigma\prod_{k=1}^{j-1}t_{a_k(\mathbf t),a_k(\mathbf t)+1},\sigma\prod_{k=1}^{j}t_{a_k(\mathbf t),a_k(\mathbf t)+1}\right).$$
  By Axiom \ref{A4}, for all $\omega,\omega'\in \mathbb{S}_n$, $a\in \left[n-1\right]$ and $i,j\in \left[n\right]$, if $\omega'=\omega t_{a,a+1}=t_{i,j}\omega$, we can define~$g_a(i,j):=\dd(\omega,\omega')$. Therefore, we can rewrite the previous expression for $\dd(\sigma,\pi)$ as in \eqref{LuanaCerini2}. Now we prove the converse. Notice that Axiom \ref{A3} holds by Proposition \ref{pazzia}. Passing to Axiom \ref{A4} let $\sigma,\pi\in \mathbb{S}_n$ be such that
  \[
  I(\sigma,\sigma t_{a,a+1})=I(\pi,\pi_{a,a+1})=\left\lbrace (i,j) \right\rbrace.
  \]
  Then,
\[
 \dd(\sigma,\sigma t_{a,a+1})=g_{a}(i,j)= \dd(\pi,\pi t_{a,a+1}).
\]
\end{proof}

\begin{theorem}\label{coro_separazioneperianale}
 A semimetric $\dd$ satisfies Axioms \ref{A3}, \ref{A2ter}, and \ref{A2'} if and only if there are $\pp\in \R^{n-1}_+$ and $\mu\in \mathcal{M}^{\geq}_n$ such that for all $\sigma,\pi\in \LL_n$,
    \begin{equation}\label{LuanaCerini1}
    \dd(\sigma,\pi)=\min_{\mathbf t\in T(\sigma,\pi)}\sum_{k=1}^{\dd_K(\sigma,\pi)}\phi_{a_k(\mathbf t)}\left(\mu(i_k(\mathbf t))+\mu(j_k(\mathbf t))\right).
    \end{equation}
    Moreover, $\mu$ can be chosen so that $\sum_i\mu(i)=1$.    
 \end{theorem}
\begin{proof} If $\dd\equiv 0$ there is nothing to show. Therefore we assume that $\dd$ is not identically zero.\\
If $n\geq 4$, notice that \ref{A2'} implies \ref{A4}. Indeed, let $\sigma,\pi\in \LL_n$ and $a\in [n-1]$ such that
\[
I(\sigma,\sigma t_{a,a+1})=I(\pi,\pi t_{a,a+1}).
\]
Let $\rho\in \LL_n$ be such that $I(\sigma,\sigma t_{a,a+1})\cap I(\rho,\rho t_{a,a+1})=\emptyset$ and $I(\pi,\pi t_{a,a+1})\cap I(\rho,\rho t_{a,a+1})=\emptyset$. Then, by Axiom \ref{A2'}, we have that
\[
\dd(\sigma,\sigma t_{a,a+1})+\dd(\rho,\rho t_{a,a+1})=\dd(\pi,\pi t_{a,a+1})+\dd(\rho,\rho t_{a,a+1})
\]
and hence Axiom \ref{A4} holds. 
\noindent Therefore, by Theorem \ref{main_g_thm} it is sufficient to prove that $\dd$ satisfies axioms \ref{A2ter} and \ref{A2'} if and only if there are $\pp\in \R^{n-1}_+$ and $\mu\in \mathcal{M}^{\geq}_n$ such that for all $a\in [n-1]$ and distinct $i,j\in [n]$,
 $$g_{a}(i,j)=(\mu(i)+\mu(j))\phi_{a}.$$
\noindent From Axiom \ref{A2'}, we have that for all $a\in[n-1]$, and all distinct $i,j,k,\ell\in[n]$,
  $$g_a(i,j)+g_a(k,\ell)=g_a(i,k)+g_a(j,\ell).$$
  We now define
  $\mu_a(i)=(g_a(i,j)-g_a(k,j)+g_a(i,k))/2$ for some distinct $j,k$ both different from $i$. This is well defined because if $k',j'$ are also distinct and different from $i$, then
  \begin{align*}
  &\left(g_a(i,j)-g_a(k,j)+g_a(i,k)\right)-\left(g_a(i,j')-g_a(k',j')+g_a(i,k')\right)\\
  &=\left(g_a(i,j)+g_a(k',j')\right)-\left(g_a(i,j')+g_a(k,j)\right)+g_a(i,k)-g_a(i,k')\\
  &=\left(g_a(i,k')+g_a(j,j')\right)-\left(g_a(i,k)+g_a(j,j')\right)+g_a(i,k)-g_a(i,k')=0
  \end{align*}
  But then, we have that for distinct $i,j\in [n]$,
  \begin{align*}
g_a(i,j)&=\frac12\left(g_a(i,j)-g_a(k,j)+g_a(i,k)\right)+\frac12\left(g_a(j,i)-g_a(k,i)+g_a(j,k)\right)\\
  &=\mu_a(i)+\mu_a(j).
  \end{align*}
  Now we define $\mu:i\mapsto\mu_a(i)/\sum_{j=1}^n\mu_a(j)$. This is well defined because of Axiom \ref{A2ter}. Indeed, Axiom \ref{A2ter} yields that
  $$g_a(i,j)g_b(k,\ell)=(\mu_a(i)+\mu_a(j))(\mu_b(k)+\mu_b(\ell))=(\mu_b(i)+\mu_b(j))(\mu_a(k)+\mu_a(\ell))=g_a(k,\ell)g_b(i,j).$$
  Thus, since $\dd$ is not identically zero, there exist $x\in[n-1]$ and distinct $p,q\in[n]$ such that $g_x(p,q)>0$. Moreover, in this case we must also have that $\mu_x([n])=\frac1{n-1}\sum_{(i,j)}g_x(i,j)\geq \frac{g_x(p,q)}{n-1}>0$. 
  Thus, we define
  $$c_a:=\frac{g_a(p,q)}{g_x(p,q)},$$
  then
  $g_a(i,j)=\mu_a(i)+\mu_a(j)=c_a(\mu_x(i)+\mu_x(j))$. Now if $i,j,k$ are all distinct, we get the system
  \begin{align*}\mu_a(i)+\mu_a(j)&=c_a(\mu_x(i)+\mu_x(j))\\
  \mu_a(i)+\mu_a(k)&=c_a(\mu_x(i)+\mu_x(k))\\
  \mu_a(j)+\mu_a(k)&=c_a(\mu_x(j)+\mu_x(k))
  \end{align*}
  Summing the first two and subtracting the last equation, we get $\mu_a(i)=c_a\mu_x(i)$. This shows that $\mu$ is well defined independently of $a$ and that it sums up to 1. Finally, we can define $\phi_a:=\mu_a([n])/\mu([n])$. So that we conclude that $g_a(i,j)=\mu_a(i)+\mu_a(j)=\phi_a(\mu(i)+\mu(j))$.
  
    If $n<3$, then all the Axioms hold vacuously and such does the claim. Instead if $n=3$, then Axioms \ref{A4} and \ref{A2'} still hold vacuously, and by Axiom \ref{A3} there is some $g$ such that $\dd(\sigma,\sigma t_{a,a+1})=g_a(\sigma_a,\sigma_{a+1})$ for all $\sigma\in \LL_n$ and $a\in [n-1]$. We define $\mu_x$ as
\begin{align*}\mu_x(1)&=\frac{g_x(1,2)+g_x(1,3)-g_x(2,3)}2,\\\mu_x(2)&=\frac{g_x(1,2)-g_x(1,3)+g_x(2,3)}2,\\\mu_x(3)&=\frac{-g_x(1,2)+g_x(1,3)+g_x(2,3)}2.\end{align*}
Since $\dd$ is not trivial, there exist $x\in [n-1]$ and distinct $p,q\in [n]$, such that $g_x(p,q)>0$. This implies that $\mu_x([3])>0$, thus we can define $\mu:=\mu_x/\mu_x([3])$. Finally we define
    $\phi_x:=\mu_x([3])$ and $\phi_y:=\phi_xg_y(p,q)/g_x(p,q)$ where $y\in[1,2]\setminus\{x\}$. Then we have that
    $g_x(i,j)=\phi_x(\mu(i)+\mu(j))$ for any distinct $i,j\in[3]$. Moreover, by Axiom \ref{A2ter},
    $$g_y(i,j)=\frac{g_y(p,q)}{g_x(p,q)}g_x(i,j)=\frac{g_y(p,q)}{g_x(p,q)}\phi_x(\mu(i)+\mu(j))=\phi_y(\mu(i)+\mu(j)).$$
    Therefore, the proof is concluded.
\end{proof}

\begin{lemma}\label{dsatisfiesA3}
  If $\ddd$ is a semimetric, then it satisfies Axiom \ref{A3}.
\end{lemma}
\begin{proof}
Fix $\sigma,\pi\in \LL_n$. If $\dd_K(\sigma,\pi)\leq 1$, there is nothing to prove. Therefore, suppose that $\dd_K(\sigma,\pi)\geq 2$ and let $X:=\{i\in [n-1]: \sigma_{i+1}>_{\pi}\sigma_{i}\}$. Let ~$\sigma^i=\sigma t_{i,i+1}$ and define $\chi:i\mapsto \ddd(\sigma,\sigma^i)+\ddd(\sigma^i,\pi)-\ddd(\sigma,\pi)$. We want to find $i\in [n-1]$ such that $\chi(i)=0$.
Here we display the three orders that we are considering
$$\begin{matrix}
\sigma&=&\sigma_1&\ldots&\sigma_{i-1}&\sigma_i&\sigma_{i+1}&\sigma_{i+2}&\ldots&\sigma_{n}\\
  \sigma^i&=&\sigma_1&\ldots&\sigma_{i-1}&\sigma_{i+1}&\sigma_i&\sigma_{i+2}&\ldots&\sigma_{n}\\
  \pi&=&\pi_1&\ldots&\pi_{i-1}&\pi_i&\pi_{i+1}&\pi_{i+2}&\ldots&\pi_{n}\\
\end{matrix}$$
From this, it is easy to verify that for all $i\in [n-1]$,
\begin{align*}
  \ddd(\sigma,\sigma^i)=&(\mu(\sigma_i)+\mu(\sigma_{i+1}))\left(f_\bb(n-i)-f_\bb(n-i-1)\right)\\
  \ddd(\sigma^i,\pi)=&\sum_{j\neq i,i+1}\mu(\sigma_j)\left(f_\bb(n-j)+f_\bb\left(\left|\sigma_j^{\downarrow,\pi}\right|\right)-2f_\bb\left(\left|\sigma_{(j,n]}\cap \sigma_j^{\downarrow,\pi}\right|\right)\right)\\
  &+\mu(\sigma_i)\left(f_\bb(n-i-1)+f_\bb\left(\left|\sigma_i^{\downarrow \pi}\right|\right)-2f_\bb\left(\left|\sigma_{(i+1,n]}\cap \sigma_i^{\downarrow \pi}\right|\right)\right)\\
  &+\mu(\sigma_{i+1})\left(f_\bb(n-i)+f_\bb\left(\left|\sigma_{i+1}^{\downarrow \pi}\right|\right)-2f_\bb\left(\left|\sigma_{\{i\}\cup(i+1,n]}\cap\sigma_{i+1}^{\downarrow \pi}\right|\right)\right)\\
  \ddd(\sigma,\pi)=&\sum_{j=1}^n\mu(\sigma_j)\left(f_\bb(n-j)+f_\bb \left(\left|\sigma_j^{\downarrow,\pi}\right|\right)-2f_\bb\left(\left|\sigma_{(j,n]}\cap \sigma_j^{\downarrow,\pi}\right|\right)\right).
\end{align*}
We have that $\chi$ is given by
\begin{align*}
\chi(i)=&\mu(\sigma_i)\left(f_\bb(n-i)-f_\bb(n-i-1)\right)+\\
  &\mu(\sigma_i)\left(f_\bb(n-i-1)+f_\bb\left(\left|\sigma_i^{\downarrow \pi}\right|\right)-2f_\bb\left(\left|\sigma_{(i+1,n]}\cap \sigma_i^{\downarrow \pi}\right|\right)\right)+\\
  &\mu(\sigma_i)\left(2f_\bb\left(\left|\sigma_{(i,n]}\cap\sigma_i^{\downarrow,\pi}\right|\right)-f_\bb(n-i)-f_\bb\left(\left|\sigma_i^{\downarrow,\pi}\right|\right)\right)+\\
  &\mu(\sigma_{i+1})\left(f_\bb(n-i)-f_\bb(n-i-1)\right)+\\
  &\mu(\sigma_{i+1})\left(f_\bb(n-i)+f_\bb\left(\left|\sigma_{i+1}^{\downarrow \pi}\right|\right)-2f_\bb\left(\left|\sigma_{\{i\}\cup(i+1,n]}\cap\sigma_{i+1}^{\downarrow \pi}\right|\right)\right)+\\
  &\mu(\sigma_{i+1})\left(2f_\bb\left(\left|\sigma_{(i+1,n]}\cap\sigma_{i+1}^{\downarrow,\pi}\right|\right)-f_\bb(n-i-1)-f_\bb\left(\left|\sigma_{i+1}^{\downarrow,\pi}\right|\right)\right)\\
  =&2\mu(\sigma_i)\left(f_\bb\left(\left|\sigma_{(i,n]}\cap\sigma_i^{\downarrow,\pi}\right|\right)-f_\bb\left(\left|\sigma_{(i+1,n]}\cap \sigma_i^{\downarrow \pi}\right|\right)\right)+\\
  &2\mu(\sigma_{i+1})\left(f_\bb(n-i)-f_\bb(n-i-1)+f_\bb\left(\left|\sigma_{(i+1,n]}\cap\sigma_{i+1}^{\downarrow,\pi}\right|\right)-f_\bb\left(\left|\sigma_{\{i\}\cup(i+1,n]}\cap\sigma_{i+1}^{\downarrow \pi}\right|\right)\right).
\end{align*}
\normalsize
We need to find $i^*\in X$ such that $\sigma_{i^*+1}>_\pi\sigma_{j}$ for all $j>i^*+1$ and for $j=i^*$. Indeed, if this is the case, we have that \begin{align*}
  \sigma_{(i^*+1,n]}&\subseteq \sigma_{\{i^*\}\cup(i^*+1,n]}\subseteq \sigma_{i^*+1}^{\downarrow \pi},
\end{align*}
which implies $$\left|\sigma_{(i+1,n]}\cap\sigma_{i+1}^{\downarrow \pi}\right|=n-i-1,\,\,\,\,\,\,\textrm{ and }\,\,\,\,\,\,\,\left|\sigma_{\{i\}\cup(i+1,n]}\cap\sigma_{i+1}^{\downarrow \pi}\right|=n-i.$$
Moreover, if $\sigma_{i^*+1}>_\pi\sigma_{i^*}$ we also have that 
$$\left|\sigma_{(i^*,n]}\cap\sigma_{i^*}^{\downarrow,\pi}\right|=\left|\sigma_{(i^*+1,n]}\cap \sigma_{i^*}^{\downarrow, \pi}\right|.$$
Thus, finding such $i^*$ would conclude the proof. In particular, for $i^*=\max X$, we have indeed that~$\sigma_{i^*+1}>_\pi\sigma_{i^*}$ because $i^*\in X$. Finally, because of maximality, for all $j\geq i^*+1$, we must have that $\sigma_{j}>_\pi\sigma_{j+1}$ and thus,
$$\sigma_{i^*+1}>_\pi\sigma_{i^*+2}>_\pi\ldots>_\pi\sigma_{n-1}>_\pi\sigma_{n}.$$
Therefore, $\chi(i^*)=0$ and hence we have that there exists $\omega=\sigma^{i^*}$ such that $\sigma-\omega-\pi$ and
\[
\ddd(\sigma,\pi)=\ddd(\sigma,\omega)+\ddd(\omega,\pi)
\]
proving that $\ddd$ satisfies Axiom \ref{A3}.
\end{proof}

Given a vector $\pp\in \R^{n-1}$ and $\mu\in \R^n$, we define the map $D_{\pp}^\mu:\LL_n\tim \LL_n\to \R$ as
\[
D_{\pp}^\mu:(\sigma,\pi)\mapsto\min_{\mathbf t\in T(\sigma,\pi)}\sum^{\dd_K(\sigma,\pi)}_{k=1}(\mu(i_k(\mathbf t))+\mu(j_k(\mathbf t)))\phi_{a_k(\mathbf t)}.
\]
\begin{theorem}\label{scorpions}
If $\ddd$ is a non-trivial semimetric, then there is a unique $\pp\in \R^{n-1}$ such that
\begin{equation}\label{popo}   \forall\sigma,\pi\in\LL_n ,\,\,\,\,\,\,\ddd(\sigma,\pi)=D_{\pp}^\mu(\sigma,\pi).
   \end{equation}
   Conversely, if $D^{\mu}_\pp$ is a non-trivial semimetric, there is a unique $\bb\in \R^{n-1}$ such that holds \eqref{popo}.
 \end{theorem}
\begin{proof}
We define the function $a\mapsto \phi_a$, by
\begin{equation}\label{PnuozzoArmato}
  \phi_a:=f_\bb(n-a)-f_\bb(n-a-1)=\sum_{k=2}^n\beta_k\binom{n-a-1}{k-2}.
  \end{equation}
  Assume that $\dd_K(\sigma,\pi)=1$, then $\pi=\sigma t_{a,a+1}=t_{i,j}\sigma$. In this case,
  $$\ddd(\sigma,\pi)=(\mu(i)+\mu(j))(f_\bb(n-a)-f_\bb(n-a-1))=(\mu(i)+\mu(j))\phi_a.$$
As $\ddd$ is not identically zero, there exist $x\in [n-1]$ and $\sigma,\pi\in \LL_n$ with $\pi=\sigma t_{x,x+1}$ such that 
\[
(\mu(\sigma_{x})+\mu(\sigma_{x+1}))\phi_x=\ddd(\sigma,\pi)>0.
\]
For all $a\in [n-1]$, there exist $\rho,\tau\in \LL_n$ such that $\tau=\rho t_{a,a+1}$ and $\rho_a=\sigma_x$ and $\rho_{a+1}=\sigma_{x+1}$
and then, as $\mu(\sigma_x)+\mu(\sigma_{x+1})\neq 0$. Therefore, for all $a\in [n-1]$, $\phi_a$ is pinned down uniquely. Indeed, suppose that there exist $\tilde{\pp}$ such that \eqref{popo} holds. Then, for all $a\in [n-1]$,
\[
(\mu(\sigma_x)+\mu(\sigma_{x+1}))\phi_a=(\mu(\sigma_x)+\mu(\sigma_{x+1}))\tilde{\phi}_a
\]
and hence, $\phi_a=\tilde{\phi}_a$. This shows that \eqref{popo} holds for $\dd_K(\sigma,\pi)=1$. Now suppose that $\dd_K(\sigma,\pi)>1$ and let $\mathbf t\in T(\sigma,\pi)$. Then
  \begin{align*}
  \ddd(\pi,\sigma)&\leq \sum_{k=1}^{\dd_K(\sigma,\pi)}\ddd\left(\sigma\prod_{h=1}^{k-1}t_{a_k(\mathbf t),a_k(\mathbf t)+1},\sigma\prod_{h=1}^{k}t_{a_k(\mathbf t),a_k(\mathbf t)+1}\right)\\ &=\sum_{k=1}^{\dd_K(\sigma,\pi)}(\mu(i_k(\mathbf t))+\mu(j_k(\mathbf t)))\phi_{a_k(\mathbf t)}
  \end{align*}
  which yields that
  $$\ddd(\sigma,\pi)\leq\min_{\mathbf t\in T(\sigma,\pi)}\sum_{k=1}^{\dd_K(\sigma,\pi)}(\mu(i_k(\mathbf t))+\mu(j_k(\mathbf t)))\phi_{a_k(\mathbf t)}.$$
   Now assume that \eqref{popo} holds for any couple of permutations with Kendall distance less than $m$, and suppose that $\dd_K(\sigma,\pi)=m$. By Lemma \ref{dsatisfiesA3}, there is some $\omega$ in between $\sigma$ and $\pi$ such that $\ddd(\sigma,\omega)+\ddd(\omega,\pi)=\ddd(\sigma,\pi)$. Clearly, $\dd_K(\sigma,\omega),\dd_K(\omega,\pi)<m$, thus by inductive hypothesis, there exist $\mathbf t'\in T(\sigma,\omega)$ and $\mathbf t'' \in T(\omega,\pi)$, such that
  \begin{align*}
  \ddd(\sigma,\omega)&=\sum_{k=1}^{\dd_K(\sigma,\omega)}(\mu(i_k(\mathbf t'))+\mu(j_k(\mathbf t')))\phi_{a_k(\mathbf t')}\\
  \ddd(\omega,\pi)&=\sum_{k=1}^{\dd_K(\omega,\pi)}(\mu(i_k(\mathbf t''))+\mu(j_k(\mathbf t'')))\phi_{a_k(\mathbf t'')}
  \end{align*}
  which yields, by denoting $\mathbf{\tilde t}$ the concatenation of $\mathbf t'$ and $\mathbf t''$,
  \[
\ddd(\sigma,\pi)=\sum_{k=1}^{\dd_K(\sigma,\pi)}(\mu(i_k(\mathbf{\tilde t}))+\mu(j_k(\mathbf{\tilde t})))\phi_{a_k(\mathbf{\tilde t})}\geq \min_{\mathbf{t}\in T(\sigma,\pi)}\sum_{k=1}(\mu(i_k(\mathbf{t}))+\mu(j_k(\mathbf{t})))\phi_{a_k(\mathbf{t})}
  \]
  as $\mathbf{\tilde{t}}\in T(\sigma,\pi)$. Thus, from $\bb$ we obtained a unique $\boldsymbol\phi$ satisfying \eqref{popo}. Conversely, notice that the relation  \eqref{PnuozzoArmato} is bijective with inverse given by
$$\forall a\geq 2,\ \beta_a=\sum_{k=0}^{a-2}(-1)^{a+k}\binom{a-2}{k}\phi_{n-1-k}$$
  Indeed,
  \begin{align*}
\beta_a&=\sum_{k=0}^{a-2}(-1)^{a+k}\binom{a-2}{k}\phi_{n-1-k}=\sum_{k=0}^{a-2}(-1)^{a+k}\binom{a-2}{k}\sum_{j=2}^n\beta_j\binom{k}{j-2}\\&=\sum_{j=2}^n(-1)^a\beta_j\sum_{k=j-2}^{a-2}(-1)^{k}\binom{a-2}k\binom k{j-2}=\sum_{j=2}^n(-1)^a\binom{a-2}{j-2}\beta_j\sum_{k=j-2}^{a-2}(-1)^{k}\binom {a-j}{k-j+2}\\
&=\sum_{j=2}^n(-1)^{a+j}\binom{a-2}{j-2}\beta_j\sum_{k=0}^{a-j}(-1)^{k}\binom {a-j}{k}=\sum_{j=2}^n(-1)^{a+j}\binom{a-2}{j-2}\beta_j\delta_{a=j}\\
    &=(-1)^{2a}\binom{a-2}{a-2}\beta_a=\beta_a.
  \end{align*}
  And,
  \begin{align*}
\phi_a&=\sum_{k=2}^n\binom{n-a-1}{k-2}\beta_k=\sum_{k=2}^n\binom{n-a-1}{k-2}\sum_{j=0}^{k-2}(-1)^{k+j}\binom{k-2}{j}\phi_{n-1-j}\\
&=\sum_{j=0}^n(-1)^j\phi_{n-1-j}\sum_{k=2}^{n}(-1)^k\binom{n-a-1}{k-2}\binom{k-2}{j}\\
&=\sum_{j=0}^n(-1)^j\binom{n-a-1}{j}\phi_{n-1-j}\sum_{k=2+j}^{n-a+1}(-1)^k\binom{n-a-1-j}{k-2-j}\\
&=\sum_{j=0}^n\binom{n-a-1}{j}\phi_{n-1-j}\sum_{k=0}^{n-a-1-j}(-1)^k\binom{n-a-1-j}{k}\\
    &=\sum_{j=0}^n\binom{n-a-1}{j}\phi_{n-1-j}\delta_{j=n-a-1}=\binom{n-a-1}{n-a-1}\phi_{a}=\phi_a.
  \end{align*}
  Notice this readily implies that through such $\bb$ given $\pp$, we obtain that \eqref{popo} holds.
  \end{proof}

\begin{proof}[Proof of Theorem \ref{axiomatic_characterization_THM}]
It follows from Theorems \ref{coro_separazioneperianale}, \ref{scorpions}, and Proposition \ref{prop:semimetric}, together with the observation that since $\dd^{-\mu}_{-\bb}=\ddd$, it is possible to cover all the possible semimetrics by restricting to $B_n^\geq\tim\mathcal M_n^\geq$.
\end{proof}

 \begin{proof}[Proof of Proposition \ref{vintoalfantamorto2017GrazieTotoRiiiina}]
For all $j\in[2,n]$ we have that $\beta_j=(-1)^j(\Delta^{j-2} \boldsymbol\phi)_{n+1-j}\geq0$. Indeed,
\begin{align*}
(-1)^a(\Delta^{a-2}\pp)_{n+1-a}&=(-1)^a\sum_{i=0}^{a-2}(-1)^{a-2+i}\binom{a-2}{i}\phi_{n+1-a+i}=\sum_{i=0}^{a-2}(-1)^{i}\binom{a-2}{i}\phi_{n+1-a+i}\\
&=\sum_{i=0}^{a-2}(-1)^{i}\binom{a-2}{a-2-i}\phi_{n+1-a+i}=\sum_{j=0}^{a-2}(-1)^{a-2-j}\binom{a-2}{j}\phi_{n+1-a+a-2-j}\\
&=(-1)^{a}\sum_{j=0}^{a-2}(-1)^{j}\binom{a-2}{j}\phi_{n-1-j}=\beta_a.
\end{align*}
Now we prove that $(-1)^j(\Delta^{j-2}\boldsymbol\phi)_{n+1-j-k}\geq0$ for any $k\in[0,n-j]$. We have already verified the basis case with $k=0$, thus now we assume this to hold for every $k\leq m$ and we prove it for $m+1$. By a known fact\footnote{\tiny\begin{align*}
  (\Delta^k f)_{j+1}-(\Delta^k f)_j&=\sum_{i=0}^k(-1)^{i+k}\binom kif(j+i+1)-\sum_{i=0}^k(-1)^{i+k}\binom kif(j+i)=-\sum_{i=1}^{k+1}(-1)^{i+k}\binom k{i-1}f(j+i)-\sum_{i=0}^k(-1)^{i+k}\binom kif(j+i)\\
  &=-\sum_{i=0}^{k+1}(-1)^{i+k}\left(\binom ki+\binom k{i-1}\right)f(j+i)=\sum_{i=0}^{k+1}(-1)^{i+k+1}\binom{k+1}{i}f(j+i)=(\Delta^{k+1} f)_j
  \end{align*}\normalsize}, we have that
    \begin{align*}
    (\Delta^{j-2} \pp)_{n+1-j-m}-(\Delta^{j-2} \pp)_{n+1-j-(m+1)}&=(\Delta^{j+1-2} \pp)_{n+1-j-(m+1)}\\
    &=(\Delta^{j+1-2} \pp)_{n+1-(j+1)-m}.
    \end{align*}
    Thus,
    \begin{gather*}
    (\Delta^{j-2} \pp)_{n+1-j-(m+1)}=(\Delta^{j-2} \pp)_{n+1-j-m}-(\Delta^{j+1-2} \pp)_{n+1-(j+1)-m}\\
    (-1)^j(\Delta^{j-2} \pp)_{n+1-j-(m+1)}=(-1)^j(\Delta^{j-2} \pp)_{n+1-j-m}+(-1)^{j+1}(\Delta^{j+1-2} \pp)_{n+1-(j+1)-m}.
    \end{gather*}
    The latter terms are both non-negative by inductive hypothesis.
    Conversely, if $\pp\in\Phi_n$ is totally monotone, then (if $\bb:=F^{-1}(\pp)$) $\beta_j=(-1)^j(\Delta^{j-2} \boldsymbol\phi)_{n+1-j}\geq0$. Moreover, since $\beta_2=\phi_{n-1}>0$, we have that $\bb\in B_n^*$.
  \end{proof}
  \subsubsection{A characterization of Axiom A.5 and A.6}
  \begin{itemize}
\item[\textbf{A.7}.]    There is a function $c: [n-1]^2\to\mathbb{R}_{++}$ such that for all $\sigma, \pi \in \mathbb{S}_n$ and $a,b \in [n-1]$ 
    \begin{equation}\label{A5A6_c}
    I(\sigma,\sigma t_{a,a+1})=I(\pi,\pi t_{b,b+1}) \implies \frac{\mathrm{d}(\sigma,\sigma t_{a,a+1})}{ \mathrm{d}(\pi,\pi t_{b,b+1})} = c(a,b).
    \end{equation}
\end{itemize}
\begin{proposition}
A metric $\mathrm{d}$ satisfies Axioms A.5 and A.6 if and only if it satisfies Axiom A.7.
\end{proposition}
Notice that whenever such $c$ exists by \eqref{A5A6_c}, we have $c(a,b)c(b,a)=1$ for all $a,b\in [n-1]$, and in particular $c(a,a)=1$.
\begin{proof}

Suppose $\dd$ satisfies Axiom A.7. Let $ \sigma, \pi, \rho, \tau$ and $a$ be such that $I(\sigma,\sigma t_{a,a+1}) \sqcup I(\pi,\pi t_{a,a+1}) = I(\rho,\rho t_{a,a+1}) \sqcup I(\tau,\tau t_{a,a+1})$. Without loss of generality, we can assume that $I(\sigma,\sigma t_{a,a+1}) = I(\rho,\rho t_{a,a+1})$, and $I(\pi,\pi t_{a,a+1}) = I(\tau,\tau t_{a,a+1})$. Since $c(a,a)=1$, we have $\mathrm{d}(\sigma,\sigma t_{a,a+1}) = \mathrm{d}(\rho,\rho t_{a,a+1})$, and $\mathrm{d}(\pi,\pi t_{a,a+1}) = \mathrm{d}(\tau,\tau t_{a,a+1})$. Therefore, it holds $\mathrm{d}(\sigma,\sigma t_{a,a+1}) + \mathrm{d}(\pi,\pi t_{a,a+1}) = \mathrm{d}(\rho,\rho t_{a,a+1}) + \mathrm{d}(\tau,\tau t_{a,a+1})$. Thus, $\dd$ satisfies Axiom A.6.\\

We now show that $\dd$ satisfies A.5. Let $ \sigma, \pi, \rho, \tau$ and distinct $a, b \in [n-1]$ be s.t. $I(\sigma,\sigma t_{a,a+1}) = I(\pi,\pi t_{b,b+1})$ and $I(\rho,\rho t_{a,a+1}) = I(\tau,\tau t_{b,b+1})$. By \eqref{A5A6_c}, 
\[ \frac{\mathrm{d}(\sigma,\sigma t_{a,a+1})}{\mathrm{d}(\pi,\pi t_{b,b+1})} = c(a,b) = \frac{\mathrm{d}(\rho,\rho t_{a,a+1})}{\mathrm{d}(\tau,\tau t_{b,b+1})}.\]

Assume now that $\mathrm{d}$ satisfies both A.5 and A.6. We show that is satisfies Axiom A.7. Let $\rho, \tau$ and distinct $a, b \in [n-1]$ be s.t. $I(\rho,\rho t_{a,a+1}) = I(\tau,\tau t_{b,b+1})$. Define 
$c(a,b) := \mathrm{d}(\rho,\rho t_{a,a+1})/\mathrm{d}(\tau,\tau t_{b,b+1})$. Now we show that $c$ is well-defined. By Axiom A.5, for all $\pi,\sigma$ such that $I(\pi,\pi t_{a,a+1}) = I(\sigma,\sigma t_{b,b+1})$, we have 
\[ \frac{\mathrm{d}(\sigma,\sigma t_{a,a+1})}{\mathrm{d}(\pi,\pi t_{b,b+1})} = \frac{\mathrm{d}(\rho,\rho t_{a,a+1})}{\mathrm{d}(\tau,\tau t_{b,b+1})},\]
and hence $\mathrm{d}(\sigma,\sigma t_{a,a+1})/\mathrm{d}(\pi,\pi t_{b,b+1}) = c(a,b)$.

We claim that A.6 implies that if $\sigma, \rho \in \mathbb{S}_n$, and $a \in [n-1]$ are such that $I(\sigma,\sigma t_{a,a+1}) = I(\rho,\rho t_{a,a+1})$, then $\mathrm{d}(\sigma,\sigma t_{a,a+1}) = \mathrm{d}(\rho,\rho t_{a,a+1})$. Suppose $n \leq 3$. Then, $I(\sigma,\sigma t_{a,a+1}) = I(\rho,\rho t_{a,a+1})$ implies that $\sigma = \rho$; therefore, the result holds. So, assume $n \geq 4$. By contradiction, suppose, that $I(\sigma,\sigma t_{a,a+1}) = I(\rho,\rho t_{a,a+1})$ but $\mathrm{d}(\sigma,\sigma t_{a,a+1})\neq \mathrm{d}(\rho,\rho t_{a,a+1})$. As $\mathrm{d}$ is a metric, there exists $b \in [n-1]$ such that $b \neq a$ and $b \neq a+1$, and that $\sigma_{b} \neq \rho_{b}$. Clearly, $(\sigma_{b},\rho_{b}) \notin I(\sigma,\sigma t_{a,a+1})$. Let $\pi\in \LL_n$ be such that $\pi_a = \sigma_b$, $\pi_{a+1} = \rho_b$, and $\tau = \pi$. 
By construction, we have 
\[ I(\sigma,\sigma t_{a,a+1}) \sqcup I(\pi,\pi t_{a,a+1}) = I(\rho,\rho t_{a,a+1}) \sqcup I(\tau,\tau t_{a,a+1}).\]

\noindent Therefore, by A.6, $\mathrm{d}(\sigma,\sigma t_{a,a+1}) + \mathrm{d}(\pi,\pi t_{a,a+1}) = \mathrm{d}(\rho,\rho t_{a,a+1}) + \mathrm{d}(\tau,\tau t_{a,a+1}).$ However, notice that, by construction of $\pi$, $\mathrm{d}(\pi,\pi t_{a,a+1}) = \mathrm{d}(\tau,\tau t_{a,a+1})$, which, in turn, implies $\mathrm{d}(\sigma,\sigma t_{a,a+1}) = \mathrm{d}(\rho,\rho t_{a,a+1})$, a contradiction. Therefore, \eqref{A5A6_c} holds.
\end{proof}

  \subsubsection{Isometric embeddability}

\begin{proof}[Proof of Proposition \ref{propo:iso_embedding}]
We consider the following orders on $C=[n]$, 
\begin{align*}
    o_{123}&=(1,2,3,4,\ldots,n),\,\,\,\,\ \ \ \  o_{132}=(1,3,2,4,\ldots,n)\\
    o_{213}&=(2,1,3,4,\ldots,n),\,\,\,\,\ \ \ \  o_{231}=(2,3,1,4,\ldots,n)\\
    o_{312}&=(3,1,2,4,\ldots,n),\,\,\,\,\ \ \ \  o_{321}=(3,2,1,4,\ldots,n)
\end{align*}
Assume there exists an isometric embedding $f:\LL_n\to E$, we denote the images of these rankings under the isometry $f$ as $f_{ijk}=f(o_{ijk})$ and we assume without loss of generality that $f_{123}=\mathbf0_E$.
We have the following betweenness relations
$$o_{123}-o_{213}-o_{231}-o_{321}-o_{312}-o_{132}-o_{123}-o_{213}.$$
Also notice that $o_{213}$ is the only element between~$o_{123}$ and $o_{231}$; $o_{231}$ is the only element between $o_{213}$ and $o_{321}$; $o_{321}$ is the only element between $o_{231}$ and $o_{312}$; $o_{312}$ is the only element between $o_{321}$ and $o_{132}$; $o_{132}$ is the only element between $o_{312}$ and $o_{123}$; $o_{123}$ is the only element between $o_{132}$ and $o_{213}$. Thus, by Axiom $\ref{A3}$, we must have that
\begin{subequations}
\begin{align}
    \dd(o_{123},o_{213})+\dd(o_{213},o_{231})&=\dd(o_{123},o_{231}),\label{a1}\\
    \dd(o_{213},o_{231})+\dd(o_{231},o_{321})&=\dd(o_{213},o_{321}),\label{a2}\\
    \dd(o_{231},o_{321})+\dd(o_{321},o_{312})&=\dd(o_{231},o_{312}),\label{a3}\\
    \dd(o_{321},o_{312})+\dd(o_{312},o_{132})&=\dd(o_{321},o_{132}),\label{a4}\\
    \dd(o_{312},o_{132})+\dd(o_{132},o_{123})&=\dd(o_{312},o_{123}),\label{a5}\\
    \dd(o_{132},o_{123})+\dd(o_{123},o_{213})&=\dd(o_{132},o_{213})\label{a6}.
\end{align}
\end{subequations}
Since $E$ is strictly convex and because $o_{213}$ satisfies \eqref{a1}, we can find a $\lambda_1\in(0,1)$ such that
$$f_{213}=\lambda_1f_{231}.$$
In particular, we must have that $\lambda_1=\dd(o_{123},o_{213})/\dd(o_{123},o_{231})$.
Similarly, since $o_{231}$ satisfies \eqref{a2}, we must have some $\lambda_2\in(0,1)$ such that
$$f_{231}=f_{213}+\lambda_2(f_{321}-f_{213}).$$
In particular, we must have
$\lambda_2=\dd(o_{231},o_{213})/\dd(o_{213},o_{321})$.
Together with the previous result,
$$f_{231}=\lambda_1f_{231}+\lambda_2(f_{321}-\lambda_1f_{231})$$
and thus
$$f_{231}=\frac{\lambda_2}{1-\lambda_1+\lambda_1\lambda_2}f_{321}=:\nu_1f_{321},\,\,\,\,\,\,\,f_{213}=\frac{\lambda_1\lambda_2}{1-\lambda_1+\lambda_1\lambda_2}f_{321}=:\nu_2f_{321},$$
where $\nu_1,\nu_2\in(0,1)$. 
Now, since $o_{321}$ satisfies \eqref{a3}, there is some $\lambda_3\in(0,1)$ such that
$$f_{321}=f_{231}+\lambda_3(f_{312}-f_{231}).$$
Combining this with the previous relations,
$$f_{321}=(1-\lambda_3)\nu_1f_{321}+\lambda_3f_{312}$$
\begin{align*}f_{321}&=\frac{\lambda_3}{1-(1-\lambda_3)\nu_1}f_{312}&&&f_{231}&=\frac{\lambda_3\nu_1}{1-(1-\lambda_3)\nu_1}f_{312}&&&f_{213}&=\frac{\lambda_3\nu_2}{1-(1-\lambda_3)\nu_1}f_{312}\\
&=:\upsilon_{1}f_{312} &&&&=:\upsilon_{2}f_{312}&&&&=:\upsilon_{3}f_{312}
\end{align*}
Notice that $\upsilon_1,\upsilon_2,\upsilon_3\in(0,1)$. Again, since $o_{312}$ satisfies \eqref{a4}, there is some $\lambda_4\in(0,1)$ such that
$$f_{312}=f_{321}+\lambda_4(f_{132}-f_{321})$$
Combining this with the previous relations,
$$f_{312}=(1-\lambda_4)\upsilon_1f_{312}+\lambda_4f_{132}$$
\small\begin{align*}
f_{312}&=\frac{\lambda_4}{1-(1-\lambda_4)\upsilon_1}f_{132}=:\theta_{1}f_{132}&&&f_{321}&=\frac{\lambda_4\upsilon_1}{1-(1-\lambda_4)\upsilon_1}f_{132}=:\theta_{2}f_{132}\\
f_{231}&=\frac{\lambda_4\upsilon_2}{1-(1-\lambda_4)\upsilon_1}f_{132}=:\theta_{3}f_{132}&&&f_{213}&=\frac{\lambda_4\upsilon_3}{1-(1-\lambda_4)\upsilon_1}f_{132}=:\theta_{4}f_{132}
\end{align*}\normalsize
Notice that $\theta_1,\theta_2,\theta_3,\theta_4\in(0,1)$.
Finally, since $o_{132}$ satisfies $\eqref{a5}$, there is some $\lambda_5\in(0,1)$ such that
$$f_{132}=f_{312}+\lambda_5(f_{123}-f_{312})$$
Combining this with the previous relations,
$$f_{132}=(1-\lambda_5)\theta_1f_{132}+\lambda_5f_{123},\ \ \ \ f_{132}=\frac{\lambda_5}{1-(1-\lambda_5)\theta_1}f_{123}$$
But, $f_{123}=\mathbf0_E$ yields $f_{132}=f_{312}=f_{321}=f_{231}=f_{213}=\mathbf0_E$ contradicting the injectivity of $f$.
\end{proof}

\subsection{Proofs in §\ref{section:voting}}


\begin{proof}[Computations of Example \ref{ex:neutrality}]
  Let $n=3$, $V=((1,2,3),(3,1,2),(2,3,1))$, $\bb=(1,0)$ and $\mu_3>\mu_1,\mu_2>0$. Let define $m:=2(\mu_1+\mu_2+\mu_3)$, then
\begin{align*}
  \dd_\beta^\mu((1,2,3),V)&=m+\mu_1+\mu_3 &&& \dd_\beta^\mu((3,2,1),V)&=m+\mu_1+2\mu_2+\mu_3\\
  \dd_\beta^\mu((3,1,2),V)&=m+\mu_2+\mu_3 &&& \dd_\beta^\mu((2,1,3),V)&=m+2\mu_1+\mu_2+\mu_3\\
  \dd_\beta^\mu((2,3,1),V)&=m+\mu_1+\mu_2 &&& \dd_\beta^\mu((1,3,2),V)&=m+\mu_1+\mu_2+2\mu_3
\end{align*}\normalsize
Thus, the median is $P^\mu_\bb(V)=\{(2,3,1)\}$.\\
Now let $\sigma=(3,2,1)$, then $\sigma V=((3,2,1),(1,3,2),(2,1,3))$ and 
\begin{align*}
  \dd_\beta^\mu((1,2,3),\sigma V)&=m+\mu_1+2\mu_2+\mu_3 &&& \dd_\beta^\mu((3,2,1),\sigma V)&=m+\mu_1+\mu_3\\
  \dd_\beta^\mu((3,1,2),\sigma V)&=m+2\mu_1+\mu_2+\mu_3 &&& \dd_\beta^\mu((2,1,3),\sigma V)&=m+\mu_2+\mu_3\\
  \dd_\beta^\mu((2,3,1),\sigma V)&=m+\mu_1+\mu_2+2\mu_3 &&& \dd_\beta^\mu((1,3,2),\sigma V)&=m+\mu_1+\mu_2
\end{align*}\normalsize
In this case the median is $P^\mu_\bb(\sigma V)=\{(1,3,2)\}\neq \{(2,1,3)\}=\sigma\{(2,3,1)\}=\sigma P^\mu_\bb(V)$. Also notice that in this case $W^\mu_\bb(\sigma V)=\{1\}\neq \{2\}=\sigma\{2\}=\sigma W^\mu_\bb(V)$.
\end{proof}

\begin{proof}[Proof of Proposition \ref{reinf_rank}]
Let $V_1,V_2\in \V$ be such that $P^\mu_\bb(V_1)\cap P^\mu_\bb(V_2)\neq \emptyset$. We denote by $V$ the concatenation $V:=V_1\oplus V_2$. First of all notice that, 
\begin{align*}
\min\limits_{\sigma\in \LL_n}\ddd(\sigma,V)=\min\limits_{\sigma\in \LL_n}\left[\ddd(\sigma,V_1)+\ddd(\sigma,V_2)\right]\geq \min\limits_{\sigma\in \LL_n}\ddd(\sigma,V_1)+\min\limits_{\pi\in \LL_n}\ddd(\pi,V_2).
\end{align*}
Therefore, we have that $P^\mu_\bb(V_1)\cap P^\mu_\bb(V_2)\subseteq P^\mu_\bb(V).$ Conversely, suppose that $\sigma \notin P^\mu_\bb(V_1)\cap P^\mu_\bb(V_2)$. Then, it holds 
\begin{equation}\label{scopazzo}
\ddd(\sigma,V)=\ddd(\sigma,V_1)+\ddd(\sigma,V_2)> \min\limits_{\omega \in \LL_n}\ddd(\omega,V_1)+\min\limits_{\pi\in \LL_n}\ddd(\pi,V_2).  
\end{equation}
Let $\pi\in P^\mu_\bb(V_1)\cap P^\mu_\bb(V_2)\neq \emptyset$. By \eqref{scopazzo}, we have that $\ddd(\sigma,V)>\ddd(\pi,V)$. Therefore, $\sigma\notin P^\mu_\bb(V)$. By contrapositive, this yields that $P^\mu_\bb(V_1)\cap P^\mu_\bb(V_2)=P^\mu_\bb(V)$.
\end{proof}

\begin{proof}[Proof of Proposition \ref{prop:reinforcingwinners}]
Suppose first that $\beta_k = 0$, for all $2 \leq k<n$. Fix two profiles of voters~$V_1,V_2\in \mathbf{V}_n$ with $W^\mu_\bb(V_1)\cap W^\mu_\bb(V_2)\neq \emptyset$. Let $c\in W^\mu_\bb(V_1)\cap W^\mu_\bb(V_2)$. Then, for all $\sigma\in \LL_n$ with $c\in M(C,\sigma)$, we have $\sigma\in P^\mu_\bb(V_1)\cap P^\mu_\bb(V_2)$. Indeed note that, as $c\in W^\mu_\bb(V_1)\cap W^\mu_\bb(V_2)$, there exist preferences $\pi^1\in P^\mu_\bb(V_1)$ and $\pi^2\in P^\mu_\bb(V_2)$ such that $c\in M(C,\pi^1)$ and $c\in M(C,\pi^2)$. Given that~$\beta_2,\ldots,\beta_{n-1}=0$, we obtain, for $m = 1,2$,
\begin{align*}
\ddd\left(\pi^1,V_m\right)&=\sum_{j=1}^{|V_m|}\beta_n\mu\left(\{c \}\bigtriangleup M(C,v^j)\right)=\ddd(\pi^2,V_m).
\end{align*}
Thus, $\sigma\in P^\mu_\bb(V_1)\cap P^\mu_\bb(V_2)$ and by Proposition \ref{reinf_rank} we have that $P^\mu_\bb(V_1)\cap P^\mu_\bb(V_2)=P^\mu_\bb(V)$. This yields
\[
W^\mu_\bb(V_1)\cap W^\mu_\bb(V_2)=W^\mu_\bb(V).
\]

Conversely, without loss of generality we can assume that $\mu$ is the counting measure. We want to show that if $W_{\bb}$ satisfies reinforcing for winners, then $\beta_{k} = 0$, for all $k \in [2,n-1]$. To this end suppose that $\beta_k > 0$, for some $k \in [2,n-1]$. If $n = 3$, we can directly skip to Paragraph \emph{Final Computations} below. So, assume for the moment that $n \geq 4$. Consider the following two electorates $V_1,V_2\in \V$,
\begin{align*}
  V_1&=((1,2,3,4,\ldots,n),(2,3,1,4,\ldots,n),(3,1,2,4,\ldots,n))\\
  V_2&=((1,2,3,4,\ldots,n),(2,1,3,4,\ldots,n)),
\end{align*}
with $V$ denoting their concatenation, i.e. $V = V_1 \oplus V_2$. We first show that if $\sigma \in P_{\bb}(V_1)$ then it is without loss of generality to assume that $\{\sigma_1,\sigma_2,\sigma_3\} = \{1,2,3\}$.
 Let $\sigma \in \mathbb{S}_{n}$ and suppose~$\sigma_1 \notin \{1,2,3\}$. Let $k^{*} = \min \{i \in [n] \mid \sigma_{i} \in \{1,2,3\}\}$. We have $\sigma = (\sigma_1,\ldots,\sigma_{k^{*}-1},\sigma_{k^{*}},\sigma_{k^{*}+1},\ldots,\sigma_{n})$, and that~$\sigma_{i} \in \{1,2,3\}$ for some $i > k^{*}$. Let $\pi \in \mathbb{S}_n$ be defined as
$$\pi := (\sigma_{k^{*}},\sigma_1,\ldots,\sigma_{k^{*}-1},\sigma_{k^{*}+1},\ldots,\sigma_{n}).$$
Let $B \subseteq [n]$ be such that $\{\sigma_1,\sigma_{k^{*}}\} \subseteq B$ and $|B| = k$. Define the following collections of sets
\begin{enumerate}
    \item $\mathcal{S}_1 := \{A \subseteq [n] \mid \sigma_{k^{*}},\sigma_1 \in A\}$
    \item $\mathcal{S}_2 := \{A \subseteq [n] \mid \sigma_{k^{*}} \in A, \sigma_1 \notin A, \exists i \in [2,k^{*}-1],\; \textnormal{s.t.} \; \sigma_i \in A \},$
\end{enumerate}
and observe that $A \in 2^{[n]} \setminus (\mathcal{S}_1 \sqcup \mathcal{S}_2)$ if and only if $M(A,\sigma) = M(A,\pi)$. We have
\begin{equation}\label{DIO PORCONE GIULIO C'HA IL PITONE}
    \begin{aligned}
    \dd_{\bb}(\sigma,V_1) - \dd_{\bb}(\pi,V_1) &\overset{(a)}{=} \sum_{j=1}^{|V_1|} \sum_{A \in \mathcal{S}_1 \sqcup \mathcal{S}_2} \beta_{|A|}\left( |M(A,\sigma) \triangle M(A,v^{j})| - |M(A,\pi) \triangle M(A,v^{j})| \right) \\
    &\overset{(b)}{\geq} \sum_{j=1}^{|V_1|} \sum_{A \in \mathcal{S}_1} \beta_{|A|}\left( |M(A,\sigma) \triangle M(A,v^{j})| - |M(A,\pi) \triangle M(A,v^{j})| \right) \\
    &\overset{(c)}{=} \sum_{j=1}^{|V_1|} \sum_{A \in \mathcal{S}_1} \beta_{|A|}\left( 2 - |M(A,\pi) \triangle M(A,v^{j})| \right) \\
    &\overset{(d)}{\geq} \sum_{j=1}^{|V_1|} \beta_{k}\left( 2 - |M(B,\pi) \triangle M(B,v^{j})| \right) \overset{(e)}{\geq}  \beta_k (2-0) = 2\beta_k > 0,
    \end{aligned}
\end{equation}

where ($a$) holds because if $A \in 2^{[n]} \setminus (\mathcal{S}_1 \sqcup \mathcal{S}_2)$, then $M(A,\sigma) = M(A,\pi)$. ($b$) holds because, if $A \in \mathcal{S}_2$, then $M(A,\sigma) > 3 $, and $M(A,v^{j}) \leq 3$, for all $j \in [|V_1|]$. This implies that 
$|M(A,\sigma) \triangle M(A,v^{j})| - |M(A,\pi) \triangle M(A,v^{j})| = 2 - |M(A,\pi) \triangle M(A,v^{j})| \geq 0$, for all $j \in [|V_1|]$.
($c$) holds because for all subsets $A \in \mathcal{S}_1$, it holds $M(\sigma,A) = \sigma_1$ while $M(A,v^{j}) \neq \sigma_1$, for all $j \in [|V_1|]$, which, in turn, implies that $|M(A,\sigma) \triangle M(A,v^{j})| = 2$; ($d$) holds by definition of $B$; ($e$) holds because there is at least one voter $j$ such that $M(B,v^{j}) = M(B,\pi) = \sigma_{k^{*}}$. 

This shows that any optimal permutation is such that $\sigma_1 \in \{1,2,3\}$. There are two cases to be discussed. Let $\mathcal{K} := \{k \in [2,n-1] : \beta_k > 0\} $ be the set of $\beta$ that are positive.
\noindent
\paragraph{\emph{Case 1.}}{Suppose there exists $k \in \mathcal{K}$ such that $k \leq n-2$. One can remove $\sigma_1$ from the alternative set, restrict $v^j$ to $[n]\setminus\{\sigma_1\}$, and $\bb$ to $(\beta_2,\ldots,\beta_{n-1})$ (see Lemma \ref{split optimum}). By applying Lemma \ref{split optimum}, we conclude that any optimal permutation $\sigma \in \mathbb{S}_n$ is such that~$\{\sigma_1,\sigma_2\} \subseteq \{1,2,3\}$.

Furthermore, we can repeat the process by restricting $v^j$ to $[n]\setminus\{\sigma_1,\sigma_2\}$, and by truncating $\bb$ as $(\beta_2,\ldots,\beta_{n-2})$. By applying again Lemma \ref{split optimum} and repeating a third time the same arguments, we obtain that if $\sigma \in P_{\bb}(V_1)$ then~$\{\sigma_1,\sigma_2,\sigma_3\} = \{1,2,3\}$. Now, fix $\sigma \in P^{\mu}_{\bb}(V_1)$. Construct $\pi$ in the following way: $\pi_1 = \sigma_1$, $\pi_2 = \sigma_2$, $\pi_3 = \sigma_3$, and $\pi_i = i$, for $i \geq 4$. Then, we have that that~$\dd_{\bb}(\sigma,V_1) \geq \dd_{\bb}(\pi,V_1)$. Indeed, suppose that $A\in 2^{[n]}$, if $A\subseteq [4,n]$, for all $j\in [|V_1|]$
\[
|M(A,\sigma)\bigtriangleup M(A,v^j)|\geq |M(A,\pi)\bigtriangleup M(A,v^j)|=0.
\]
If $A\cap \{1,2,3\}\neq \emptyset$, then $M(A,\sigma)=M(A,\pi)$ and hence for all $j\in [|V_1|]$,
\[
|M(A,\sigma)\bigtriangleup M(A,v^j)|= |M(A,\pi)\bigtriangleup M(A,v^j)|.
\]
Therefore, there exists $\sigma^*\in P_\bb(V_1)$ such that $\sigma^{*}_i = i$ for $i \geq 4$ and $\sigma^*([3])=\left\lbrace 1,2,3\right\rbrace$.}
\noindent
\paragraph{\emph{Case 2.}}{Suppose that $\beta_{n-1} > 0$ and that $\beta_k = 0$ for $k \leq n-2$. The same steps underlying (\ref{DIO PORCONE GIULIO C'HA IL PITONE}) plus an application of Lemma \ref{split optimum} imply that if $\sigma \in P_{\bb}(V_1)$ then $\{\sigma_1,\sigma_2\} \subseteq \{1,2,3\}$. As all sets of size~$n-1$ contain at least one between $\sigma_1$ and $\sigma_2$, alternatives in positions from $3$ to $n$ are irrelevant when comparing $\sigma$ with $v^j$, for $j \in V_1$. Therefore, there exists $\sigma^*\in P_\bb(V_1)$ such that $\sigma^{*}_i = i$ for $i \geq 4$ and $\sigma^*([3])=\left\lbrace 1,2,3\right\rbrace$.}

Similar arguments apply to the other electorates $V_2$ and $V$. These steps imply that $W_\bb(V_1) \subseteq \{1,2,3\}$, $W_\bb(V_2) \subseteq \{1,2\}$, and $W_\bb(V)\subseteq \{1,2,3\}$. 

\paragraph{\emph{Final Computations}}\label{par: final computations}{ Simple algebra yields that $W_\bb(V_1)=\{1,2,3\}$, $W_\bb(V_2)=\{1,2\}$, and hence $W_\bb(V_1) \cap W_\bb(V_2) = \{1, 2\}$.
For $V$, we obtain 
 \begin{equation}\label{algebrina a pecorina}
\begin{aligned}
\dd_\bb((1,2,3,4,\ldots,n),V) &= 6\sum_{k=3}^{n} {\beta}_k \binom{n-3}{k-3} + 10\sum_{k=2}^{n-1} {\beta}_k \binom{n-3}{k-2} \\
\dd_\bb((2,1,3,4,\ldots,n),V) &= 6\sum_{k=3}^{n} {\beta}_k \binom{n-3}{k-3} + 12\sum_{k=2}^{n-1} {\beta}_k \binom{n-3}{k-2} \\
\dd_\bb((2,3,1,4,\ldots,n),V) &= 6\sum_{k=3}^{n} {\beta}_k \binom{n-3}{k-3} + 14\sum_{k=2}^{n-1} {\beta}_k \binom{n-3}{k-2}.
\end{aligned}
\end{equation}
Finally, as $\beta_k > 0$, for some $k \geq 2$, we conclude that $2 \notin W_\bb(V)$, which shows that $W_\bb$ does not satisfy reinforcing for winners.}
\end{proof}

\begin{proof}[Proof of Proposition \ref{prop:monotonicity}]
Let $V\in\mathbf{V}_n$ and $\sigma^*\in P^\mu_\bb(V)$ with $k\in M(C,\sigma^*)$. If all voters in the profile $V$ rank $k$ as their most preferred alternative, there is nothing to prove. Up to relabeling, we can assume that $k>1$ and that voter $v$ of the profile $V$ takes the form
\[
1>_v 2>_v \ldots >_v k-2>_v k-1>_v k >_v k+1>_v \ldots >_v n.
\]
Without loss of generality we can take an upranking $V^k$ obtained by substituting $v$ with $w:=t_{k-1,k} v$. If we show that $k\in W_\bb^\mu(V^k)$, then the proof is concluded by a straightforward induction argument. For all linear orders $\sigma\in \LL_n$, we have
\begin{itemize}
\item for all $i\in \left[1,k-2\right]\cup \left[k+1,n\right]$, $i^{\downarrow,\sigma}\cap i^{\downarrow,v}=i^{\downarrow,\sigma}\cap i^{\downarrow,w}$,
\item for $i=k-1$, $(k-1)^{\downarrow,\sigma}\cap (k-1)^{\downarrow,v}=(k-1)^{\downarrow,\sigma}\cap \left((k-1)^{\downarrow,w}\cup \left\lbrace k\right\rbrace\right)$ which in turn yields:
\[
\left| (k-1)^{\downarrow,\sigma}\cap (k-1)^{\downarrow,v} \right|=\left| (k-1)^{\downarrow,\sigma}\cap (k-1)^{\downarrow,w} \right|+\mathbf{1}_{(k-1)^{\downarrow,\sigma}}(k),
\]
\item for $i=k$, ${k}^{\downarrow,\sigma}\cap {k}^{\downarrow,v}={k}^{\downarrow,\sigma}\cap \left({k}^{\downarrow,w}\setminus\left\lbrace k-1\right\rbrace\right)$ which in turn yields:
\[
\left| {k}^{\downarrow,\sigma}\cap {k}^{\downarrow,v} \right|=\left| {k}^{\downarrow,\sigma}\cap {k}^{\downarrow,w} \right|-\mathbf{1}_{k^{\downarrow,\sigma}}(k-1).
\]
\end{itemize}
Fix a linear order $\sigma\in \LL_n$. The previous observations yield that
\begin{align*}
\sum_{i=1}^nf_\bb\left(\left| i^{\downarrow,\sigma}\cap i^{\downarrow,w} \right|\right)\mu(i)=&f_\bb\left(\left| k^{\downarrow,\sigma}\cap k^{\downarrow,w} \right|\right)\mu(k)+f_\bb\left(\left| (k-1)^{\downarrow,\sigma}\cap (k-1)^{\downarrow,w} \right|\right)\mu(k-1)\\
&+\sum_{i\in \left[1,k-2\right]\cup \left[k+1,n\right]}f_\bb\left(\left| i^{\downarrow,\sigma}\cap i^{\downarrow,v} \right|\right)\mu(i)\\
\leq& f_\bb\left(\left| k^{\downarrow,\sigma}\cap k^{\downarrow,w} \right|\right)\mu(k)+f_\bb\left(\left| (k-1)^{\downarrow,\sigma}\cap (k-1)^{\downarrow,v} \right|\right)\mu(k-1)\\
&+ \sum_{i\in \left[1,k-2\right]\cup \left[k+1,n\right]}f_\bb\left(\left| i^{\downarrow,\sigma}\cap i^{\downarrow,v} \right|\right)\mu(i)\\
=&f_\bb\left(\left| k^{\downarrow,\sigma}\cap k^{\downarrow,v} \right|\right)\mu(k)-f_\bb\left(\left| k^{\downarrow,\sigma}\cap k^{\downarrow,v} \right|\right)\mu(k)+f_\bb\left(\left| k^{\downarrow,\sigma}\cap k^{\downarrow,w} \right|\right)\mu(k)\\
&+\sum_{i\in \left[1,k-1\right]\cup \left[k+1,n\right]}f_\bb\left(\left| i^{\downarrow,\sigma}\cap i^{\downarrow,v} \right|\right)\mu(i)\\
=&\left(f_\bb\left(\left| k^{\downarrow,\sigma}\cap k^{\downarrow,w} \right|\right)-f_\bb\left(\left| k^{\downarrow,\sigma}\cap k^{\downarrow,v} \right|\right)\right)\mu(k)\\
&+\sum_{i=1}^nf_\bb\left(\left| i^{\downarrow,\sigma}\cap i^{\downarrow,v} \right|\right)\mu(i).
\end{align*}
The inequality above follows from the observations in the bullet points and Lemma \ref{DioCinghiale}. It is important to notice that whenever $k\in M(C,\sigma)$, the above inequality holds as equality. To ease the notation we let, for all linear orders $\sigma\in \LL_n$ and profile $U\in \V$,
\[
h^\mu_\bb(\sigma,U):=\sum_{j=1}^{|U|} \sum_{i = 1}^{n} \left[f_\bb\left( {i}^{\downarrow,\sigma} \right) - 2 f_\bb\left(\left|{i}^{\downarrow,\sigma} \cap {i}^{\downarrow,u^j}\right| \right) \right]\mu(i).\footnote{Observe that $\argmin_{\sigma \in \LL_n} \ddd(\sigma,V) = \argmin_{\sigma \in \LL_n} h^\mu_\bb(\sigma,V)$, and similarly for the new electorate we obtain after the upranking.}
\]
For all $\sigma\in \LL_n$, define $g(\sigma) := 2\left(f_\bb\left(\left| {k}^{\downarrow,\sigma} \cap {k}^{\downarrow,w}\right| \right) - f_\bb\left(\left| {k}^{\downarrow,\sigma} \cap {k}^{\downarrow,v}\right| \right) \right) \mu({k})$. The inequality provided above yields
\begin{equation}\label{TDalpha inequality}
h^\mu_\bb(\sigma,V^k)\geq h^\mu_\bb(\sigma,{V}) - g(\sigma)
\end{equation}
for all $\sigma\in \LL_n$. Consequently, it follows that
\begin{equation}\label{inequalitybestemmia}
\begin{aligned}
\min_{\sigma\in \LL_n} h^\mu_\bb(\sigma,V^k)&\geq \min_{\sigma\in \LL_n}\left[h^\mu_\bb(\sigma,{V}) -g(\sigma)\right]\geq \min_{\sigma\in \LL_n} h^\mu_\bb(\sigma,{V}) - \max_{\sigma\in \LL_n}g(\sigma).
\end{aligned}
\end{equation}

\begin{claim}\label{simple claim}
Let $\pi \in \argmin_\sigma h^\mu_\bb(\sigma,V) - g(\sigma) $ be such that $k\in M(C,\pi)$. Then, $\pi \in \argmin_{\sigma} h^\mu_\bb(\sigma,V^k) $.
\end{claim}
\begin{proof}
For all $\sigma\in \LL_n$, by \eqref{TDalpha inequality},
\[
h^\mu_\bb(\pi,V) - g(\pi) \leq h^\mu_\bb(\sigma,V) - g(\sigma) \leq h^\mu_\bb(\sigma, V^k).
\]
Given that $k >_{\pi} {j}$, for all $j\neq k$, we have $h^\mu_\bb(\pi,V) - g(\pi) = h^\mu_\bb(\pi,V^k)$.
\end{proof}

Now, observe that, by \eqref{inequalitybestemmia}, whenever the intersection of $\argmin_{\sigma} h^\mu_\bb(\sigma,V) $ and $\argmax_{\sigma} g(\sigma) $ is non-empty, any element in this intersection minimizes $h^\mu_\bb(\sigma,V) - g(\sigma)$. Further, we have the following
\begin{equation*}
\begin{aligned}
\left|{k}^{\downarrow,\sigma} \cap k^{\downarrow,w}\right| &\overset{(a)}{=} \left| {k}^{\downarrow,\sigma} \cap \left({k}^{\downarrow,v} \cup \{k-1\} \right) \right| \\
&= \left| \left( {k}^{\downarrow,\sigma} \cap {k}^{\downarrow,v} \right) \cup \left( {k}^{\downarrow,\sigma} \cap \{k-1\} \right) \right| \\
&\overset{(b)}{=} \left|  {k}^{\downarrow,\sigma} \cap {k}^{\downarrow,v}  \right| + \left| {k}^{\downarrow,\sigma} \cap \{k-1\} \right|,
\end{aligned}
\end{equation*}
\noindent
where ($a$) holds by construction of $w$ and ($b$) holds because ${k}^{\downarrow,v} \cap \{k-1\} = \emptyset$. Thus, we obtain
\begin{equation*}
\begin{aligned}
f_\bb\left(\left|{k}^{\downarrow,\sigma} \cap {k}^{\downarrow,w}\right| \right) - f_\bb\left(\left|{k}^{\downarrow,\sigma} \cap {k}^{\downarrow,v}\right| \right) &= f_\bb\left(\left|{k}^{\downarrow,\sigma} \cap {k}^{\downarrow,v}\right| + \left|{k}^{\downarrow,\sigma} \cap \{k-1\} \right| \right) \\
&-f_\bb\left(\left|{k}^{\downarrow,\sigma} \cap {k}^{\downarrow,v}\right| \right) \\
&\overset{(c)}{\leq} f_\bb\left(\left|{k}^{\downarrow,\sigma} \cap {k}^{\downarrow,v}\right| + 1 \right) - f_\bb\left(\left|{k}^{\downarrow,\sigma} \cap {k}^{\downarrow,v}\right| \right) \\
&\overset{(d)}{\leq} f_\bb\left(\left| {k}^{\downarrow,v}\right| + 1 \right) - f_\bb\left(\left| {k}^{\downarrow,v}\right| \right),
\end{aligned}
\end{equation*}
where ($c$) and ($d$) hold because of Lemmas \ref{DioMajale} and \ref{DioCinghiale}, respectively. Observe that the upper bound is achieved by any linear order $\sigma$ that ranks $k$ as first alternative. Thus, any linear order that ranks $k$ as its most preferred alternative maximizes $g$.

By assumption, $\sigma^{*}$ minimizes $\ddd(\cdot,V)$ (and so $h^\mu_\bb(\cdot,V)$) and it ranks $k$ as its most preferred alternative. Therefore, it minimizes $h^\mu_\bb(\cdot,V) - g(\cdot)$ (and so $\ddd(\cdot,V) - g(\cdot)$). By invoking Claim \ref{simple claim}, we conclude that $\sigma^{*}$ is a minimizer of $h^\mu_\bb(\cdot,V^k)$ and, in turn, of $\ddd(\cdot,V^k)$.
\end{proof}

\begin{proof}[Proof of Proposition \ref{prop:majority}]
Let $V=(v^1,\ldots,v^m)\in \V$ and $c\in [n]$, be such that $n(V,c)\geq 0$ and that $c\notin W^\mu_\bb(V)$. Without loss of generality we set $c=1$. Consider following sets
\[
S=\left\lbrace i\in [m]:c\in M\left([n],v^i\right) \right\rbrace\ \textnormal{and}\ T=[m]\setminus S.
\]
By assumption $| S|\geq | T|$. We have that for all $\omega \in P^\mu_\bb(V)$, candidate $c$ is not ranked first. Let~$\omega\in P^\mu_\bb(V)$, we can, without loss of generality, write it as
\[ 2 >_{\omega} \ldots >_{\omega} i >_{\omega} 1 >_{\omega} i+1 \ldots >_{\omega} n.\]
We aim to show that $\Id\in P^\mu_\bb(V)$ to reach a contradiction. Fix $A\subseteq [n]$ and $j\in S$. Suppose $1\in A$, it follows that $M(A,\Id)=M(A,v^j)$. Thus, we have that $\bigtriangleup_A(\Id,v^j)=\emptyset$. This yields,
\[
\bigtriangleup_A(\omega,\Id)\sqcup \bigtriangleup_A(\Id,v^j)=\bigtriangleup_A(\omega,v^j).
\]
If $1\notin A$, then $M(A,\Id)=M(A,\omega)$ and hence $\bigtriangleup_A(\omega,\Id)=\emptyset$. This yields,
\[
\bigtriangleup_A(\omega,\Id)\sqcup \bigtriangleup_A(\Id,v^j)=\bigtriangleup_A(\omega,v^j).
\]
By definition of $\ddd$, we have that
\begin{equation}\label{bisattino}
\ddd\left(\omega,v^j\right)=\ddd\left(\omega,\Id\right)+\ddd\left(\Id,v^j\right).
\end{equation}
By \eqref{bisattino} and the triangle inequality we have that
\begin{align*}
\sum_{j=1}^m\ddd\left(\omega,v^j\right)&=\sum_{j\in S}\ddd\left(\omega,v^j\right)+\sum_{j\in T}\ddd\left(\omega,v^j\right)\\
&\geq \sum_{j\in S}\left[\ddd\left(\omega,\Id\right)+\ddd\left(\Id,v^j\right)\right]+\sum_{j\in T}\left[\ddd\left(\Id,v^j\right)-\ddd\left(\omega,\Id\right)\right]\\
&\geq \left(| S|-| T|\right)\ddd\left(\omega,\Id\right)+\sum_{j=1}^m\ddd\left(\Id,v^j\right) \geq \sum_{j=1}^m\ddd\left(\Id,v^j\right),
\end{align*}
where the last inequality follows from observing that $| S|\geq | T|$. Since $\omega\in P^\mu_\bb(V)$ it follows that~$\Id\in P^\mu_\bb(V)$, but then we get $1\in W^\mu_\bb(V)$, a contradiction.
\end{proof}

\begin{proof}[Proof of Proposition \ref{prop:blockpareto}]
 Fix $k\leq n$ and $V=(v^1,\ldots,v^m)\in \V$ with $v^i(\left[k\right])=\left[k\right]$ for all~$i\in \left[m\right]$. Suppose there exists $v\in P_\bb(V)$ such that $v([k])\neq[k]$. We have that $|[k]\cap v((k,n])|=|(k,n]\cap v([k])|\neq0$. We now consider the transposition $t_{a,b}$ of $$a := M([k]\cap v((k,n]),v)\ \ \  \textnormal{and}\ \ \ b := m((k,n]\cap v([k]),v),$$
where we observe that, by definition, $a < b$, but $b >_v a$. Let $\tilde v:=t_{a,b}v$. We prove that for all $i\in[m]$, 
  $$\dd_\bb(\tilde v,v^i)< \dd_\bb(v,v^i).$$
  Notice that, by possibly relabeling the candidates, we can assume wlog that $v^i=\Id$. Let now $a'\in[k]$ be such that $v(a')=b$ and $b'\in(k,n]$ be such that $v(b')=a$. Clearly, we have $a' < b'$. To ease the notation we define $j^{\downarrow =} := j^{\downarrow} \cup \{j\}$, and, similarly, $j^{\uparrow =} := j^{\uparrow} \cup \{j\}$, for all $j\in [n]$. To better grasp the subsequent steps of the proof the reader can refer to the following illustration.
    \begin{figure}[ht!]
    \centering
    \scalebox{.7}{
    $
    \begin{array}{ccccccccccccc|ccccccccccc}
        \Id&=&1&\ldots&a-1&a&a+1&\ldots&a'-1&a'&a'+1&\ldots&k&k+1&\ldots&b'-1&b'&b'+1&\ldots&b-1&b&b+1&\ldots&n\\
        \tilde v&=&v_1&\ldots&v_{a-1}&v_a&v_{a+1}&\ldots&v_{a'-1}&a&v_{a'+1}&\ldots&v_k&v_{k+1}&\ldots&v_{b'-1}&b&v_{b'+1}&\ldots&v_{b-1}&v_b&v_{b+1}&\ldots&v_n\\
        v&=&v_1&\ldots&v_{a-1}&v_a&v_{a+1}&\ldots&v_{a'-1}&b&v_{a'+1}&\ldots&v_k&v_{k+1}&\ldots&v_{b'-1}&a&v_{b'+1}&\ldots&v_{b-1}&v_b&v_{b+1}&\ldots&v_n\\
    \end{array}
    $}
    \caption{Here we illustrate the three permutations in the case with $a<a'$ and $b>b'$.}
\end{figure}

 \noindent We can see that the following facts hold true.
\begin{enumerate}[label=(\roman*)]
  \item \label{itm:1} If $j\in [1,a)$, then $j^\downarrow$ contains both $a$ and $b$ and thus $|j^\downarrow\cap j^{\downarrow, v}|=|j^\downarrow\cap j^{\downarrow, \tilde v}|$.
  \item \label{itm:2} If $j\in(b,n]$, then $j^\downarrow$ contains neither $a$ nor $b$ and thus $|j^\downarrow\cap j^{\downarrow, v}|=|j^\downarrow\cap j^{\downarrow, \tilde v}|$.
  \item \label{itm:3} If $j=a$, then $a^{\downarrow,\tilde v}=(b^{\downarrow,v}\setminus\{a\})\cup\{b\}$. Notice that $a^\downarrow\cap a^{\downarrow, v}\subseteq a^\downarrow\cap a^{\downarrow,\tilde v}$. Moreover, 
  \begin{align*}
    \lvert (a^\downarrow\cap a^{\downarrow,\tilde v})\setminus(a^\downarrow\cap a^{\downarrow, v})\rvert&=\lvert(a^\downarrow\cap ((b^{\downarrow,v}\setminus\{a\})\cup\{b\}))\setminus(a^\downarrow\cap a^{\downarrow, v})\rvert\\
    &=\lvert(a^\downarrow\cap b^{\downarrow =,v})\setminus(a^\downarrow\cap a^{\downarrow, v})\rvert\\
    &=|a^\downarrow\cap b^{\downarrow =,v}\cap a^{\uparrow =,v}|\geq|\{b\}|=1.
  \end{align*}
  \item \label{itm:4} If $j\in(a,b)$, $j^\downarrow$ contains $b$ but not $a$. Moreover, $j^{\downarrow,v}=j^{\downarrow, \tilde v}$ iff $v^{-1}(j)<a'$ (i.e., $j >_v b$) or $v^{-1}(j)>b'$ (i.e., $a >_v j$). Notice that $v^{-1}(j)$ cannot be neither $a'$ nor $b'$. On the other hand, if $v^{-1}(j)\in (a',b')$, then $j^{\downarrow,\tilde{v}}=\{b\}\cup(j^{\downarrow, v}\setminus\{a\})$, and $$|j^\downarrow\cap j^{\downarrow\tilde v}|=|(j^\downarrow\cap j^{\downarrow v})\sqcup\{b\}|=1+|j^\downarrow\cap j^{\downarrow v}|.$$
  \item \label{itm:5} If $j=b$, $b^{\downarrow,\tilde v}=a^{\downarrow, v}$. Notice that $b^\downarrow\cap a^{\downarrow, v}\subseteq b^\downarrow\cap b^{\downarrow,v}$. Moreover, $$\lvert (b^\downarrow\cap b^{\downarrow,v})\setminus(b^\downarrow\cap a^{\downarrow, v})\rvert=|b^\downarrow\cap b^{\downarrow,v}\cap a^{\uparrow =,v}|.$$
  Notice that $b^\downarrow\cap b^{\downarrow,v}\cap a^{\uparrow =,v}\subsetneq a^\downarrow\cap b^{\downarrow =,v}\cap a^{\uparrow =,v}$ as $b$ belongs to the second but not to the first set and $b^{\downarrow}\subseteq a^{\downarrow}$.
\end{enumerate}
  Finally, we evaluate the difference between the two distances: 
  \begin{align*}
    \frac{\dd_\bb(v,\Id)-\dd_\bb(\tilde v,\Id)}2=&\sum_{j=1}^nf_\bb(|j^{\downarrow}\cap j^{\downarrow,\tilde v}|)-f_\bb(|j^{\downarrow}\cap j^{\downarrow,v}|)\\
    =&f_\bb(|a^\downarrow\cap a^{\downarrow,\tilde v}|)-f_\bb(|a^\downarrow\cap a^{\downarrow, v}|)-\left(f_\bb(|b^\downarrow\cap b^{\downarrow,v}|)-f_\bb(|b^\downarrow\cap b^{\downarrow,\tilde{v}}|)\right)\\
    &+\sum_{j\in (a,b)}(f_\bb(|j^\downarrow\cap j^{\downarrow,v}|+1)-f_\bb(|j^\downarrow\cap j^{\downarrow, v}|))\mathbf 1_{v^{-1}(j)\in (a',b')}\\
    {\geq}& f_\bb(|a^\downarrow\cap a^{\downarrow,\tilde v}|)-f_\bb(|a^\downarrow\cap a^{\downarrow, v}|)-\left(f_\bb(|b^\downarrow\cap b^{\downarrow,v}|)-f_\bb(|b^\downarrow\cap a^{\downarrow,v}|)\right)\\
    =&f_\bb(|a^\downarrow\cap a^{\downarrow, v}|+|a^\downarrow\cap b^{\downarrow =,v}\cap a^{\uparrow =,v}|)-f_\bb(|a^\downarrow\cap a^{\downarrow, v}|)\\
    &-\left(f_\bb(|b^\downarrow\cap a^{\downarrow,v}|+|b^\downarrow\cap b^{\downarrow,v}\cap a^{\uparrow =,v}|)-f_\bb(|b^\downarrow\cap a^{\downarrow,v}|)\right)\\
    \overset{(a)}{>}&f_\bb(|a^\downarrow\cap a^{\downarrow, v}|+|b^\downarrow\cap b^{\downarrow,v}\cap a^{\uparrow =,v}|)-f_\bb(|a^\downarrow\cap a^{\downarrow, v}|)\\
    &-\left(f_\bb(|b^\downarrow\cap a^{\downarrow,v}|+|b^\downarrow\cap b^{\downarrow,v}\cap a^{\uparrow =,v}|)-f_\bb(|b^\downarrow\cap a^{\downarrow,v}|)\right)\overset{(b)}{\geq}0,
  \end{align*}\normalsize
where ($a$) follows from $b^\downarrow\cap b^{\downarrow,v}\cap a^{\uparrow =,v}\subsetneq a^\downarrow\cap b^{\downarrow =,v}\cap a^{\uparrow =,v}$ and Lemma \ref{DioMajale} while ($b$) holds by combining Lemma \ref{DioCinghiale} and $|a^\downarrow\cap a^{\downarrow, v}|\geq |b^\downarrow\cap a^{\downarrow,v}|$.


\end{proof}

  \begin{proof}[Proof of Corollary \ref{coro:partitionwisepareto}]
Let $J=\{j\in [n]:\forall i,i'\in[m], v^i([j])=v^{i'}([j])\}$ this is non empty since~$n\in J$. Since $P_\bb$ has the blockwise Pareto winner property, for any $j\in J$, $A_j:=v^1([j])$ is such that for any $v\in P_\bb(V)$, $v([j])=A_j$. Moreover, for all $j\in J$ and ~$v\in P_\bb(V)$, $v((j,n])=A_j^{\mathsf{c}}$.

Now, given any sequence $(k_j)_{j=0}^\ell$ as in the definition of the partitionwise Pareto property, we have that for all $j\in[\ell]$, $k_j\in J$ and, by setting $A_0=\emptyset$, we have that $v^i((k_{j-1},k_j])=A_{k_j}\setminus A_{k_{j-1}}$ for all $i\in [m]$. But then, 
for any $v\in P_\bb(V)$, we have that
$$v((k_{j-1},k_j])=v((0,k_j]\setminus (0,k_{j-1}])=v((0,k_j])\setminus v((0,k_{j-1}])=A_{k_j}\setminus A_{k_{j-1}}.$$ 
\end{proof}

\begin{proof}[Proof of Proposition \ref{Condorcet_characterization_PW}]

\par\medskip
Here we prove that both \ref{it1:condch} and  \ref{it2:condch} imply \ref{it3:condch}. We start by showing that \ref{it2:condch} implies \ref{it3:condch}.

Assume that $n \geq 3$. Furthermore, without loss of generality, we can assume $\mu_1 \leq \mu_2 \leq \cdots \leq \mu_n$ with $\mu_i + \mu_ j > 0$, for all $i \neq j$.\footnote{Notice that this implies that while $\mu_1$ can be smaller or equal to zero, necessarily $\mu_j > 0$, for all $j \neq 1$.}
Consider the following profile of voters
\begin{equation}\label{condorcet_profile}V = (5(1,2,3,4,\ldots,n),4(3,2,1,4,\ldots,n),1(2,3,1,4,\ldots,n)).\footnote{When $n=3$ we understand that the corresponding permutations refer only to the first three positions.}
\end{equation}
Notice that $n_{1,2}(V)=0$. We show that unless ${\beta}_k = 0$, for all~$ k \geq 3$, $\mathrm{I} \in P^{\mu}_{\bb}(V)$ while any permutation~$\pi$ with $(2,1) \in \mathrm{adj}(\pi)$ does not belong to $P^{\mu}_{\bb}(V)$. We have 
\begin{equation}\label{right lower bound PORCA MADONNA}
\underset{\sigma \in \mathbb{S}_n}{\min} \sum_{j=1}^{|V|} \sum_{A\in \mathcal{P}_n} \beta_{|A|} \mu(M(A,\sigma) \triangle M(A,v^j)) \geq  \sum_{A\in \mathcal{P}_n} \beta_{|A|} \min\limits_{\sigma \in \mathbb{S}_n} \sum_{j=1}^{|V|} \mu(M(A,\sigma) \triangle M(A,v^j)).
\end{equation}

We first argue that, for all $A \in \mathcal{P}_n$, $\mathrm{I} \in \argmin_{\sigma}  \sum_{j=1}^{|V|}  \mu(M(A,\sigma) \triangle M(A,v^j)).$ Let $A \in \mathcal{P}_n$, we denote by $a = \min A=M(A,\Id)$. 

\noindent
Suppose $a \geq 3$, then, by construction, $M(A,v^j) = a$ for all $j\in [|V|]$. Therefore, any permutation $\sigma$ such that $M(A,\sigma) = a$ is optimal. Clearly, $\mathrm{I}$ is one of these permutations.

\noindent
Suppose for the rest of the argument that $a \leq 2$. It holds

\begin{table}[H]
    \centering
    \begin{tabular}{c|ccc}
        $M(A,v^j)$ & $\mathrm{I}$ & $(3,2,1,4,\ldots,n)$ & $(2,3,1,4,\ldots,n)$\\
        \hline
        $a = 1$ & $1$ & $1 + 2\mathbf{1}_{3 \in A} + \mathbf{1}_{2 \in A, 3 \notin A}$  & $1 + \mathbf{1}_{2 \in A} + 2\mathbf{1}_{3 \in A, 2 \notin A}$\\
        $a = 2$ & $2$ & $2 + \mathbf{1}_{3 \in A}$ & $2$
    \end{tabular}
    \label{tab:my_label}
\end{table}
\noindent
Suppose now that $a = 2$. If $3 \notin A$, then $M(A,v^j) = 2$ for all $j \in [|V|]$, so that $\mathrm{I}$ is optimal. On the other hand, if $3 \in A$, then $M(A,(3,2,1,4,\ldots,n)) = 3$. Notice that for this $A$, the value of~$\sum_{j=1}^{|V|}  \mu(M(A,\sigma) \triangle M(A,v^j))$ for each $\sigma\in \LL_n$, only depends on $x:=M(A,\sigma)$. In particular,
\begin{align*}
 \sum_{j=1}^{|V|}  \mu(M(A,\sigma) \triangle M(A,v^j))= 5 (\mu_x + \mu_2) \mathbf{1}_{x \neq 2 } + 4 (\mu_x + \mu_3) \mathbf{1}_{x \neq 3} + (\mu_x + \mu_2) \mathbf{1}_{x \neq 2}.
\end{align*}

\noindent As $\mu_1 \leq \mu_2 \leq \cdots \leq \mu_n$, and $\mu_i + \mu_j > 0$ for $i \neq j$, any $\sigma\in \LL_n$ such that $M(A,\sigma) = 2$ is optimal. 
Therefore, $\mathrm{I}$ is optimal.

For the remaining case, assume that $a = 1$. If $2 \notin A$ and $3 \notin A$, then $M(A,v^j) = 1$ for all $j \in [|V|]$. Therefore, $\mathrm{I}$ is optimal. Suppose, $2 \in A$ and $3 \notin A$. Then, the minimization problem becomes 
$$\underset{x \in A}{\min} \; 5 (\mu_x + \mu_1) \mathbf{1}_{x \neq 1 } + 4 (\mu_x + \mu_2) \mathbf{1}_{x \neq 2} + (\mu_x + \mu_2) \mathbf{1}_{x \neq 2}.$$

\noindent As $\mu_1 \leq \mu_2 \leq \cdots \leq \mu_n$ and $\mu_i + \mu_j > 0$, for $i \neq j$,  any permutation $\sigma$ such that $M(A,\sigma) = 1$ or $M(A,\sigma) = 2$ is optimal. Clearly, $\mathrm{I}$ is one of these permutations. Now, suppose $2 \notin A$ and $3 \in A$. Then, the minimization problem is
$$\underset{x \in A}{\min} \; 5 (\mu_x + \mu_1) \mathbf{1}_{x \neq 1 } + 4 (\mu_x + \mu_3) \mathbf{1}_{x \neq 3} + (\mu_x + \mu_3) \mathbf{1}_{x \neq 3}.$$

As $\mu_1 \leq \mu_2 \leq \cdots \leq \mu_n$ and $\mu_i + \mu_j > 0$, for $i \neq j$, any permutation $\sigma$ such that $M(A,\sigma) = 1$ or $M(A,\sigma) = 3$ is optimal. Clearly, $\mathrm{I}$ is one of these permutations. Finally, suppose $2 \in A$ and $3 \in A$. The minimization problem is
$$\underset{x \in A}{\min} \; 5 (\mu_x + \mu_1) \mathbf{1}_{x \neq 1 } + 4 (\mu_x + \mu_3) \mathbf{1}_{x \neq 3} + (\mu_x + \mu_2) \mathbf{1}_{x \neq 2}.$$

By using $\mu_1 \leq \mu_2 \leq \cdots \leq \mu_n$ and $\mu_i + \mu_j > 0$, for $i \neq j$, clearly no $x \geq 4$ (if any) can be optimal. Observe

\begin{equation}\label{CAZZO}
    \begin{aligned}
      &M(A,\sigma) = 1 \implies \sum_{j = 1}^{|V|} \mu\left(M(A,\sigma) \triangle M(C,v)\right) = 5\mu_1+\mu_2+4\mu_3  \\
      &M(A,\sigma) = 2 \implies \sum_{j = 1}^{|V|}  \mu\left(M(A,\sigma) \triangle M(A,v)\right) =  5\mu_1+9\mu_2+4\mu_3 \\
      &M(A,\sigma) = 3 \implies \sum_{j = 1}^{|V|}  \mu\left(M(A,\sigma) \triangle M(A,v)\right) = 5\mu_1+\mu_2+6\mu_3.
    \end{aligned}
\end{equation}
\noindent
Recall that, under our assumption that $\mu_1 \leq \mu_2 \leq \cdots \leq \mu_n$ and $\mu_i + \mu_j > 0$, for $i \neq j$, we have that, for $i \geq 2$, $\mu_i > 0$; in particular, $\mu_2 > 0$ and $\mu_3 > 0$. Therefore, any permutation such that~$M(A,\sigma) = 1$ is optimal. Clearly, $\mathrm{I}$ is optimal. By using Inequality (\ref{right lower bound PORCA MADONNA}), we conclude that $\mathrm{I} \in P^{\mu}_{\bb}(V)$. 

Finally, suppose that $\beta_k > 0$ for some $k \geq 3$ and let $\pi \in \mathbb{S}_n$ be any permutation such that~$2>_{\pi}1$. Let $A$ be a set size $k$ that contains $1$, $2$ and $3$. By construction, over $A$, $\pi$ does not choose $1$, and so, by (\ref{CAZZO}) cannot be optimal, i.e., $\pi \notin P^{\mu}_{\bb}(V)$. In particular, for $\pi$, with $(2,1)\in \textnormal{adj}(\pi)$, we have that $\pi\notin P^{\mu}_\bb(V)$.
\par\medskip
Using the same profile of voters \eqref{condorcet_profile} we can show that \ref{it1:condch} implies \ref{it3:condch}. We have that $1,2\in \mathcal{C}(V)$. Indeed, $n_{1,j}(V)\geq 0$ and $n_{2,j}(V)\geq 0$ for all $j\in [n]$. Though, following exactly the same steps as in the previous implication we have that unless ${\beta}_k = 0$, for all~$ k \geq 3$, $1 \in W^{\mu}_{\bb}(V)$ while $2\notin W^{\mu}_\bb(V)$.

\par\medskip
Now we prove that \ref{it3:condch} implies both \ref{it1:condch} and \ref{it2:condch}.
 Fix $c,d\in [n]$ and $\sigma,\pi \in \LL_n$ such that they differ only on comparing $c,d$, without loss of generality say $d >_{\sigma} c$ and $c >_{\pi} d$. Notice that for all $V\in \mathbf{V}_n$, given that $\beta_k=0$ for all $k\geq 3$,
\begin{align}
\ddd\left(\sigma,V\right)-\ddd\left(\pi,V\right)&=\sum_{j=1}^{|V|}\ddd\left(\sigma,v^j \right)-\ddd\left(\pi,v^j\right)\nonumber\\
  &=\beta_2\sum_{j=1}^{|V|}\mu\left(\bigtriangleup_{\left\lbrace c,d\right\rbrace}(\sigma,v^j)\right)-\mu\left(\bigtriangleup_{\left\lbrace c,d\right\rbrace}(\pi,v^j)\right)\nonumber\\
  &=\beta_2\left[\sum_{j:c >_{v^j} d}\mu\left(\bigtriangleup_{\left\lbrace c,d\right\rbrace}(\sigma,v^j)\right)-\sum_{j:d >_{v^j} c}\mu\left(\bigtriangleup_{\left\lbrace c,d\right\rbrace}(\pi,v^j)\right)\right]\nonumber\\
  &=\beta_2\left[\sum_{j:c >_{v^j} d}\mu\left(\left\lbrace c,d \right\rbrace\right)-\sum_{j:d >_{v^j} c}\mu\left(\left\lbrace c,d \right\rbrace\right)\right]\nonumber\\
  &=\beta_2n_{c,d}(V)\mu\left(\left\lbrace c,d \right\rbrace\right).  \label{culo2}
\end{align}

First we show that \ref{it3:condch} implies \ref{it1:condch}. Suppose that $W_{\bb}^{\mu}$ does not satisfy the Condorcet principle. Then, there exists $V\in \V$ and $c\in \mathcal{C}(V)$ such that for all $\sigma\in P^{\mu}_{\bb}(V)$, $c\notin M(C,\sigma)$. Therefore, for all $\sigma\in P^{\mu}_\bb(V)$ there exists $d\in C$ such that~$(d,c)\in \textnormal{adj}(\sigma)$. Let $\pi:=t_{d,c}\sigma$. Since $c\in \mathcal{C}(V)$, the previous equalities yield that
\[
\ddd\left(\sigma,V\right)-\ddd\left(\pi,V\right)=\beta_2n_{c,d}(V)\mu\left(\left\lbrace c,d \right\rbrace\right)\geq 0,
\]
and hence, $\pi \in P^{\mu}_{\bb}$. Repeating the same steps finitely many times, we have that there must exist some~$\omega\in P^{\mu}_{\bb}$ such that $c\in M(C,\omega)$, a contradiction. Therefore, $W^{\mu}_\bb$ satisfies the Condorcet principle.
\par\medskip
Now we show that \ref{it3:condch} implies \ref{it2:condch}. For the sake of contradiction, suppose that $P^\mu_\bb$ does not satisfy the Condorcet principle. Then, only two cases can happen. 
\begin{enumerate}
\item There exist $c,d\in C$ and $V\in \mathbf{V}_n$ such that $n_{c,d}(V)>0$ and some $\sigma\in P^\mu_\bb(V)$ with $(d,c)\in\mathrm{adj}(\sigma)$. By \eqref{culo2}, since $\beta_2>0$ and $\mu\in \mathcal{M}^{++}_n$, $\ddd(t_{d,c}\sigma,V)<\ddd(\sigma,V)$ contradicting the minimality of $\sigma$.
\item There exist $c,d\in C$ and $V\in \mathbf{V}_n$ such that $n_{c,d}(V)=0$, there is some $\sigma\in P^\mu_\bb(V)$ such that $(c,d)\in\mathrm{adj}(\sigma)$, but there is no $\pi\in P^\mu_\bb(V)$ such that $(d,c)\in\mathrm{adj}(\pi)$.
Though, by \eqref{culo2}, $\pi:=t_{c,d}\sigma$ satisfies $\ddd\left(\sigma,V\right)=\ddd\left(\pi,V\right)$ and hence $\pi\in P^\mu_\bb(V)$, a contradiction.
\end{enumerate}
Thus, $P^\mu_\bb$ satisfies the Condorcet principle.
\end{proof}

\subsection{Proofs in §\ref{section:computational}}

{
\begin{proof}[Proof of Theorem \ref{general approximation thm}]
Let $\sigma,\pi\in \LL_n$. Because of neutrality we can assume, without loss of generality, that $\pi = \Id$. We conveniently define $\lVert \sigma \rVert_{\bb} := \dd_{\bb}(\sigma,\Id)/2$ and $\lVert \sigma \rVert^{\textnormal{apx}}_{\bb} := \dd^{\textnormal{apx}}_{\bb}(\sigma,\Id)/2$. We start showing that $\dd_{\bb}^{\textnormal{apx}}(\sigma, \pi) \leq \dd_{\bb}(\sigma, \pi)$. First of all, notice that
\begin{equation*}
\dd_{\bb}^{\textnormal{apx}}(\sigma, \pi)=\sum_{j=1}^n f_\bb(|{j}^{\downarrow,\sigma}|)+f_\bb(|{j}^{\downarrow,\pi}|)-2\min\left( f_\bb(|{j}^{\downarrow,\sigma}|),f_\bb(|{j}^{\downarrow,\pi}|) \right).
\end{equation*}
Therefore, since $|A\cap B|\leq \min (|A|,|B|)$ and $f_\bb$ is increasing (Lemma \ref{DioMajale}), it holds that
  \begin{align*}
  \lVert \sigma \rVert_{\bb}-\lVert \sigma \rVert_{\bb}^{\textnormal{apx}}&=\sum_{j=1}^n\left[f_\bb(\min (|{j}^{\downarrow,\sigma}|),n-j)-f_\bb(|{j}^{\downarrow,\sigma}\cap\ \{{j+1},\ldots,{n}\}|) \right]\\
  &\geq\sum_{j=1}^n \left[f_\bb(\min(|{j}^{\downarrow,\sigma}|),n-j)-f_\bb(\min(|{j}^{\downarrow,\sigma}|,n-j)) \right]=0,
  \end{align*}
and this implies that $\dd_{\bb}(\sigma,
\pi)\geq \dd^{\textnormal{apx}}_{\bb}(\sigma,\pi).$ \\

\noindent
Now we show that $\dd_{\bb}(\sigma, \pi) \leq \gamma_{\bb} \dd_{\bb}^{\textnormal{apx}}(\sigma, \pi)$. If $\sigma=\pi$ the claim is trivial. Therefore suppose that $\sigma\neq \Id$. Denote $\sigma^0:=\sigma$, and consider the following procedure to transform $\sigma^{0}$ into the identity $\Id$.
\vspace{0.5cm}

\begin{algorithm}[H]\label{optimus prime}
{\small
\SetKwInput{Init}{Inputs}
\SetKwInput{Par}{Parameters}
 \Init{Permutation $\sigma^0$}
 Set $k = 0$\\
 \For{$i \in \{n,n-1,\ldots,1 \}$}{
 \While{$\sigma^{k}_{i} \neq {i}$}{
 \vspace{0.1cm}
$\sigma^{k+1} \leftarrow \sigma^k$\\
 ${\sigma^{k+1}_{\sigma^{k}_{i}}} \leftarrow {\sigma^{k}_{i}}$ and
 ${\sigma^{k+1}_{i}} \leftarrow {\sigma^{k}_{\sigma^{k}_{i}}}$\\
 $k \leftarrow k+1$

 }
 }\caption{Bottom-focus sorting}
 }
\end{algorithm}


Notice that Algorithm \ref{optimus prime} reaches the identity in $O(n^2)$ steps. Denote by $K$ the number of transformations necessary to convert $\sigma^{0}$ into $\Id$. Now, consider the transition from $\sigma^{k}$ to $\sigma^{k+1}$. Let $s$ be the greatest integer in $[n]$ such that $\sigma^k_s\neq s$. We have 

\begin{equation*}
\begin{aligned}
\sigma^{k}&=\left({\sigma^{k}_{1}},\cdots, {\sigma^{k}_{r-1}},{\sigma^{k}_{r}},\cdots,{\sigma^{k}_{s}},{\sigma^{k}_{s+1}}, \cdots,{\sigma^{k}_{n}}\right) \\
\sigma^{k+1}&=\left({\sigma^{k}_{1}},\cdots, {\sigma^{k}_{r-1}}, {\sigma^{k}_{s}},\cdots,{\sigma^{k}_{r}},{\sigma^{k}_{s+1}},\cdots,{\sigma^{k}_{n}}\right)
\end{aligned}
\end{equation*}

As, by design, the procedure works from the bottom to the top, and, by assumption, we are up-ranking ${\sigma^{k}_{s}}$, necessarily all alternatives in positions from $s+1$ to $n$ have already been matched. That is ${\sigma^{k}_{i}} = i$, for all $i \geq s+1$.\footnote{We understand that when $s = n$, no element has been previously matched.} As all elements from position $s+1$ to $n$ have been matched, it holds that, for all $i \leq s$, $\sigma^{k}_{i} \leq s$, and, in particular, $\sigma^{k}_{r} \leq s$. Furthermore, as we are swapping ${\sigma^{k}_{s}}$ with~${\sigma^{k}_{r}}$, we have $r = \sigma^{k}_{s}$ and $s > \sigma^{k}_{s}$. Finally, notice that $\sigma^{k}_{r} \neq r = \sigma^{k}_{s}$. 







Let $\Delta_k := \lVert\sigma^{k} \rVert_{\bb} - \lVert\sigma^{k+1} \rVert_{\bb}$ and $\Delta_{k}^{\textnormal{apx}} := \lVert\sigma^{k} \rVert_{\bb}^{\textnormal{apx}} - \lVert \sigma^{k+1} \rVert_{\bb}^{\textnormal{apx}}$. Then, we have $\lVert \sigma^{0} \rVert_{\bb} = \sum_{k=0}^{K-1} \Delta_k$ and $\lVert \sigma^{0} \rVert_{\bb}^{\textnormal{apx}} = \sum_{k=0}^{K-1} \Delta_{k}^{\textnormal{apx}}$, as $\lVert \Id \rVert_{\bb}^{\textnormal{apx}} = \lVert \Id \rVert_{\bb} = 0$. The following holds
\begin{equation}\label{apx equality}
\begin{aligned}
\Delta_{k}^{\textnormal{apx}} &= \sum_{j=1}^{n} \left[ f_\bb \left( \min \left( n-j, |{j}^{\downarrow,\sigma^{k+1}}| \right) \right) - f_\bb \left( \min \left( n-j, |{j}^{\downarrow,\sigma^{k}}| \right) \right) \right] \\
&\overset{(a)}{=} \sum_{j = r,\sigma^{k}_{r}} \left[ f_\bb \left( \min \left( n-j, |{j}^{\downarrow,\sigma^{k+1}}| \right) \right) - f_\bb \left( \min \left( n-j, |{j}^{\downarrow,\sigma^{k}}| \right) \right) \right] \\
&{=} f_\bb \left( \min \left( n-r, n-r \right) \right) - f_\bb \left( \min \left( n-r, n-s \right) \right) \\
&+ f_\bb \left( \min \left( n-\sigma^{k}_{r}, n-s \right) \right) - f_\bb \left( \min \left( n-\sigma^{k}_{r}, n-r \right) \right) \\
&\overset{(b)}{=} f_\bb \left( n-r \right) - f_\bb \left( n - \max \left(r,\sigma^{k}_{r} \right) \right) \geq 0,
\end{aligned}
\end{equation}
where ($a$) holds because $\sigma^{k}$ and $\sigma^{k+1}$ disagree only on candidates $r$ and ${\sigma^{k}_{r}}$, and ($b$) holds because, by construction, $s > r$ and, as argued above, $\sigma^{k}_{r} \leq s$. Further, note that, by applying Lemma \ref{DioMajale}, we conclude that $\Delta_{k}^{\textnormal{apx}} > 0$ if and only if $\sigma^{k}_{r} > r$.
For the non-approximated version the analysis is more delicate. We have
\begin{equation*}
\begin{aligned}
\Delta_{k} &= \sum_{j=1}^{n} \left[ f_\bb\left( \left| {j}^{\downarrow,\sigma^{k+1}} \cap \{ {j+1},\ldots,{n} \} \right| \right) - f_\bb\left( \left| {j}^{\downarrow,\sigma^{k}} \cap \{ {j+1},\ldots,{n} \} \right| \right) \right] \\
&\overset{(a)}{=} \sum_{j =r}^{s} \left[ f_\bb\left( \left| {\sigma^{k}_{j}}^{\downarrow,\sigma^{k+1}} \cap \{ {\sigma^{k}_{j}+1},\ldots,{n} \} \right| \right) - f_\bb\left( \left| {\sigma^{k}_{j}}^{\downarrow,\sigma^{k}} \cap \{ {\sigma^{k}_{j}+1},\ldots,{n} \} \right| \right) \right],
\end{aligned}
\end{equation*}
where ($a$) holds because only the alternatives in positions from $r$ to $s$ are affected by the transformation.
We handle the cases separately. 

\noindent
Case $ j = r$.
\begin{equation*}
\begin{aligned}
&f_\bb\left( \left| {\sigma^{k}_{r}}^{\downarrow,\sigma^{k+1}} \cap \{ {\sigma^{k}_{r}+1},\ldots,{n} \} \right| \right) - f_\bb\left( \left| {\sigma^{k}_{r}}^{\downarrow,\sigma^{k}} \cap \{ {\sigma^{k}_{r}+1},\ldots,{n} \} \right| \right) \\
&= f_\bb(|\{ {s+1},\ldots,{n}\}|) \\
&- f_\bb\left(\left| \left( \{{\sigma^{k}_{r+1}},\ldots,{\sigma^{k}_{s-1}}\} \cup \{ {\sigma^{k}_{s}}\} \cup \{ {s+1},\ldots,{n} \} \right) \cap \left( \{{\sigma^{k}_{r}+1},\ldots,s\} \cup \{{s+1},\ldots,{n} \} \right) \right| \right) \\
&= f_\bb(n-s) - f_\bb\left(n-s + \left| \{{\sigma^{k}_{r+1}},\ldots,{\sigma^{k}_{s-1}}, {\sigma^{k}_{s}}\} \cap \{{\sigma^{k}_{r}+1},\ldots,s\} \right| \right) \\
&\overset{(a)}{=} f_\bb(n-s) - f_\bb\left(n-s + \left| \{{\sigma^{k}_{r+1}},\ldots,{\sigma^{k}_{s-1}}, {r}\} \cap \{{\sigma^{k}_{r}+1},\ldots,s\} \right| \right) \\
&= f_\bb(n-s) - f_\bb\left(n-s + \left| \{{\sigma^{k}_{r+1}},\ldots,{\sigma^{k}_{s-1}}\} \cap \{{\sigma^{k}_{r}+1},\ldots,s\} \right| + \mathbf{1}(r \in \{{\sigma^{k}_{r}+1},\ldots,s\})\right)
\end{aligned}
\end{equation*}
where ($a$) holds because $\sigma^{k}_{s} = r$.

\noindent
Case $ j = s$. 
\begin{equation*}
\begin{aligned}
&f_\bb\left( \left| {\sigma^{k}_{s}}^{\downarrow,\sigma^{k+1}} \cap \{ {\sigma^{k}_{s}+1},\ldots,{n} \} \right| \right) - f_\bb\left( \left| {\sigma^{k}_{s}}^{\downarrow,\sigma^{k}} \cap \{ {\sigma^{k}_{s}+1},\ldots,{n} \} \right| \right) \\
&= f_\bb\left(\left| \left( \{{\sigma^{k}_{r+1}},\ldots,{\sigma^{k}_{s-1}}\} \cup \{ {\sigma^{k}_{r}}\} \cup \{ {s+1},\ldots,{n} \} \right) \cap \left( \{{\sigma^{k}_{s}+1},\ldots,s\} \cup \{{s+1},\ldots,{n} \} \right) \right| \right) \\
&- f_\bb(|\{ {s+1},\ldots,{n}\}|) \\
&= f_\bb\left(n-s + \left| \{{\sigma^{k}_{r+1}},\ldots,{\sigma^{k}_{s-1}}, {\sigma^{k}_{r}}\} \cap \{{\sigma^{k}_{s}+1},\ldots,s\} \right| \right) - f_\bb(n-s) \\
&\overset{(a)}{=} f_\bb\left(n-s + \left| \{{\sigma^{k}_{r+1}},\ldots,{\sigma^{k}_{s-1}}\} \cap \{{r+1},\ldots,s\} \right| + \mathbf{1}( 
 {\sigma^{k}_{r}} \in \{{r+1},\ldots,s\})\right) - f_\bb(n-s), 
\end{aligned}
\end{equation*}
where ($a$) holds because $\sigma^{k}_{s} = r$.

\noindent
Cases $ r < j < s$. 

\begin{equation*}
\begin{aligned}
&f_\bb\left( \left| {\sigma^{k}_{j}}^{\downarrow,\sigma^{k+1}} \cap \{ {\sigma^{k}_{j}+1},\ldots,{n} \} \right| \right) - f_\bb\left( \left| {\sigma^{k}_{j}}^{\downarrow,\sigma^{k}} \cap \{ {\sigma^{k}_{j}+1},\ldots,{n} \} \right| \right) \\
=& f_\bb\left(\left| \left( \{{\sigma^{k}_{j+1}},\ldots,{\sigma^{k}_{s-1}}\} \cup \{ {\sigma^{k}_{r}}\} \cup \{ {s+1},\ldots,{n} \} \right) \cap \left( \{{\sigma^{k}_{j}+1},\ldots,s\} \cup \{{s+1},\ldots,{n} \} \right)\right| \right) \\
&- f_\bb\left(\left| \left( \{{\sigma^{k}_{j+1}},\ldots,{\sigma^{k}_{s-1}}\} \cup \{ {\sigma^{k}_{s}}\} \cup \{ {s+1},\ldots,{n} \} \right) \cap \left( \{{\sigma^{k}_{j}+1},\ldots,s\} \cup \{{s+1},\ldots,{n} \} \right) \right| \right) \\
=& f_\bb\left(n-s + \left| \{{\sigma^{k}_{j+1}},\ldots,{\sigma^{k}_{s-1}}, {\sigma^{k}_{r}}\} \cap \{{\sigma^{k}_{j}+1},\ldots,s\} \right| \right) \\
&- f_\bb\left(n-s + \left| \{{\sigma^{k}_{j+1}},\ldots,{\sigma^{k}_{s-1}}, {\sigma^{k}_{s}}\} \cap \{{\sigma^{k}_{j}+1},\ldots,s\} \right| \right) \\
\overset{(a)}{=}& f_\bb\left(n-s + \left| \{{\sigma^{k}_{j+1}},\ldots,{\sigma^{k}_{s-1}}\} \cap \{{\sigma^{k}_{j}+1},\ldots,s\} \right| + \mathbf{1}({\sigma^{k}_{r}} \in \{{\sigma^{k}_{j}+1},\ldots,s\}) \right) \\
&- f_\bb\left(n-s + \left| \{{\sigma^{k}_{j+1}},\ldots,{\sigma^{k}_{s-1}}\} \cap \{{\sigma^{k}_{j}+1},\ldots,s\} \right| + \mathbf{1}({r} \in \{{\sigma^{k}_{j}+1},\ldots,s\}) \right),
\end{aligned}
\end{equation*}
where ($a$) holds because $\sigma^{k}_{s} = r.$
Taking into account all cases, we have established that

\begin{equation}\label{exact equality}
\begin{aligned}
\Delta_{k} &= \Delta_{k}^{r}+\sum_{j=r+1}^{s-1} \Delta_{k}^{j},
\end{aligned}
\end{equation}

where, to ease notation, we define for all $r<j<s$
\begin{equation*}
\begin{aligned}
 \Delta_{k}^{j} &:= f_\bb\left(n-s + | \{{\sigma^{k}_{j+1}},\ldots,{\sigma^{k}_{s-1}}\} \cap \{{\sigma^{k}_{j}+1},\ldots,s\} | + \mathbf{1}({\sigma^{k}_{r}} \in \{{\sigma^{k}_{j}+1},\ldots,s\}) \right) \\ 
 &- f_\bb\left(n-s + | \{{\sigma^{k}_{j+1}},\ldots,{\sigma^{k}_{s-1}}\} \cap \{{\sigma^{k}_{j}+1},\ldots,s\} | + \mathbf{1}({r} \in \{{\sigma^{k}_{j}+1},\ldots,s\}) 
\right) \\
 \Delta_{k}^{r} &:= f_\bb\left(n-s + | \{{\sigma^{k}_{r+1}},\ldots,{\sigma^{k}_{s-1}}\} \cap \{{r+1},\ldots,s\} | + \mathbf{1}({\sigma^{k}_{r}} \in \{{r+1},\ldots,s\}) \right) \\ 
 &- f_\bb\left(n-s + | \{{\sigma^{k}_{r+1}},\ldots,{\sigma^{k}_{s-1}}\} \cap \{{\sigma^{k}_{r}+1},\ldots,s\} | + \mathbf{1}({r} \in \{{\sigma^{k}_{r}+1},\ldots,s\}) \right).
\end{aligned}
\end{equation*}

\begin{claim}\label{sigma less than r}
If $\sigma^{k}_{r} < r$ then, for all $j \in \{r+1,\ldots,s-1\} \cup \{r\}$, $\Delta_k^{j} \leq 0$, and $\Delta_k \leq 0.$
\end{claim}

\begin{proof}
First of all, observe that, for all $j \in \{r+1,\ldots,s-1\}$, if ${\sigma^{k}_{r}} \in \{{\sigma^{k}_{j}+1},\ldots,s\}$, then also ${r} \in \{{\sigma^{k}_{j}+1},\ldots,s\}$. Therefore, $\Delta_{k}^{j}=0$. Further, as $\sigma^{k}_{r} < r$, we have that~$\mathbf{1}({\sigma^{k}_{r}} \in \{{r+1},\ldots,s\}) = 0$, and $\mathbf{1}({r} \in \{{\sigma^{k}_{r}+1},\ldots,s\}) = 1$. Therefore, it holds
\begin{equation*}
  \begin{aligned}
    \Delta_{k}^{r} &\leq f_\bb\left(n-s + \left| \{{\sigma^{k}_{r+1}},\ldots,{\sigma^{k}_{s-1}}\} \cap \{{\sigma^{k}_{r}+1},\ldots,s\} \right| \right) \\ 
 &- f_\bb\left(n-s + \left| \{{\sigma^{k}_{r+1}},\ldots,{\sigma^{k}_{s-1}}\} \cap \{{\sigma^{k}_{r}+1},\ldots,s\} \right| + 1 \right) \\
 &\overset{(a)}{\leq} 0,
  \end{aligned}
\end{equation*}
where ($a$) holds because $f_\bb$ is non-decreasing (Lemma \ref{DioMajale}) 
\end{proof}

By using Claim \ref{sigma less than r}, we have established that, whenever $\sigma^{k}_{r} < r$, $\Delta_k \leq 0$.

In the following, we discuss the case $\sigma^{k}_{r} > r$. The same arguments underlying Claim \ref{sigma less than r} yield

\begin{claim}\label{sigma larger than r}
If $\sigma^{k}_{r} > r$ then, for all $j \in \{r+1,\ldots,s-1\} \cup \{r\}$, $\Delta_k^{j} \geq 0$, and $\Delta_k \geq 0.$
\end{claim}

\noindent
Furthermore, if $\sigma^k_r>r$, for all $j \in \{r+1,\ldots,s-1\}$, it holds that 
\[ \neg \left( r < \sigma^{k}_{j}+1 \leq \sigma^{k}_{r} \right) \Longleftrightarrow \Delta_{k}^{j} = 0.\]

\noindent Further, by recalling that $\sigma^{k}_{s} = r$ (and so $\sigma^{k}_{j} \neq r$), we have
\[ r < \sigma^{k}_{j}+1 \leq \sigma^{k}_{r} \iff r+1 \leq \sigma^{k}_{j} \leq \sigma^{k}_{r}-1.\]
\noindent
Let $\mathcal J_{\sigma^k} := \{ j \in \{r+1,\ldots,s-1\} \mid r+1 \leq \sigma^{k}_{j} \leq \sigma^k_r-1\}$. To ease the notation we write $\mathcal J_k:=\mathcal J_{\sigma^k}$. We have shown that

\begin{equation}\label{what we accomplished}
  \begin{aligned}
  \Delta_k &= \Delta_{k}^{r}+\sum_{j \in \mathcal{J}_{k} } \Delta_{k}^{j}, \; \textnormal{with} \; \Delta_k^{j} \geq 0, \; \textnormal{for all } j \in \mathcal{J}_{k} \cup \{r\}.
  \end{aligned}
\end{equation}

\paragraph{\emph{Case 1.}}{Suppose $r < \sigma^{k}_r < s$. In such case we consider the permutation $\omega$ defined as follows 

\begin{equation}\label{conditions 1}
\begin{aligned}
\omega_{r} &= \sigma^{k}_{r} \; \textnormal{and} \; \omega_{s} = \sigma^{k}_s = r \\
\omega_{\sigma^{k}_{r}} &= s \\
\omega_{i} &= i, \; \textnormal{otherwise.} 
\end{aligned}
\end{equation}

}

\paragraph{\emph{Case 2.}}{Suppose $r < \sigma^{k}_r = s$. In such case we consider the permutation $\omega$ defined as follows 

\begin{equation}\label{conditions 2}
\begin{aligned}
\omega_{r} &= \sigma^{k}_{r} = s \; \textnormal{and} \; \omega_{s} = \sigma^{k}_s = r \\
\omega_{i} &= i, \; \textnormal{otherwise.} 
\end{aligned}
\end{equation}

}

Observe that the \emph{only} conditions we used to establish (\ref{what we accomplished}) are:

\begin{itemize}
\item $r < s$
\item $\sigma^{k}_i = i$, for all $i \geq s+1$
\item $\sigma^{k}_s = r$ 
\item $\sigma^{k}_{r} > r$,
\end{itemize}

and that for the rest $\sigma^{k}$ was arbitrary. In particular, notice that, regardless of whether we are under \emph{Case} $1$ or \emph{Case} $2$, we have that $\omega_i = i$, for all $i \geq s+1$, $\omega_s = r$ and $\omega_r = \sigma^{k}_r$. Therefore, under the assumption $\sigma^{k}_r > r$,  regardless of whether we are in \emph{Case} $1$ or \emph{Case} $2$, the same steps used to show (\ref{what we accomplished}) apply and yield 

\[ \Delta_{\omega} := \|\omega\|_{\bb}-\|\omega t_{r,s}\|_{\bb} = \Delta^{r}_{\omega} + \sum_{j \in \mathcal{J}_{\omega}} \Delta^{j}_{\omega},\]
\noindent
with $\Delta^{j}_{\omega} \geq 0$, for all $j \in \mathcal{J}_{\omega} \cup \{r\},$ where $\mathcal J_{\omega} := \{ j \in \{r+1,\ldots,s-1\} \mid r+1 \leq \omega_{j} \leq \omega_r-1\}$.

\begin{claim}\label{Diomaialino}
Under both \emph{Case} $1$ and \emph{Case} $2$, 
    \[\mathcal{J}_{\omega} = \{r+1,\ldots,\sigma^{k}_r-1\} \; \textnormal{and} \;
 |\mathcal{J}_{\omega}| \geq |\mathcal{J}_k|.\]
\end{claim}
\begin{proof}
First of all, under \emph{Case} $2$, it holds
\begin{equation*}
    \begin{aligned}
    |\mathcal{J}_{\omega}| &= |\{  j \in \{r+1,\ldots,s-1\} \mid r+1 \leq \omega_j \leq \omega_r - 1\}| \\
        &\overset{(a)}{=} |\{  j \in \{r+1,\ldots,s-1\} \mid r+1 \leq j \leq \sigma^{k}_r - 1\}| \\
        &\overset{(b)}{=} |\{ j \in \{r+1,\ldots,\sigma^{k}_r-1\} \mid r+1 \leq j \leq \sigma^{k}_r - 1\}| 
        = \sigma^{k}_r-r-1 \overset{(c)}{\geq} |\mathcal{J}_k|,
    \end{aligned}
\end{equation*}
where ($a$) holds because $\omega_r = \sigma^{k}_r$, and, for all $j \in \{r+1,\ldots,s-1\}$, $\omega_j = j$; ($b$) holds because $\sigma^{k}_r = s$; ($c$) holds because, for any $\tau \in \mathbb{S}_n$ such that $\tau_r = \sigma^{k}_r$,
$\mathcal{J}_{\tau} := \{j \in \{r+1,\ldots,s-1\} \mid r+1 \leq \delta_j \leq \sigma^{k}_r - 1\}$ has cardinality at most $\sigma^{k}_r-r-1$.

Under \emph{Case} $1$, we have
\begin{equation*}
    \begin{aligned}
    |\mathcal{J}_{\omega}| &= |\{  j \in \{r+1,\ldots,s-1\} \mid r+1 \leq \omega_j \leq \omega_r - 1\}| \\
        &\overset{(a)}{=} |\{  j \in \{r+1,\ldots,s-1\} \setminus \{\sigma^{k}_r\} \mid r+1 \leq j \leq \sigma^{k}_r - 1\}| + \mathbf{1}_{r+1 \leq \omega_{\sigma^{k}_r} \leq \sigma^{k}_r - 1}\\ 
        &\overset{(b)}{=} |\{  j \in \{r+1,\ldots,s-1\} \setminus \{\sigma^{k}_r\} \mid r+1 \leq j \leq \sigma^{k}_r - 1\}| + \mathbf{1}_{r+1 \leq s \leq \sigma^{k}_r - 1}\\
        &\overset{(c)}{=} |\{  j \in \{r+1,\ldots,s-1\} \setminus \{\sigma^{k}_r\} \mid r+1 \leq j \leq \sigma^{k}_r - 1\}| \\
        &\overset{(d)}{=} |\{  j \in \{r+1,\ldots,\sigma^{k}_r-1\} \mid r+1 \leq j \leq \sigma^{k}_r - 1\}| = \sigma^{k}_r-r-1 \geq |\mathcal{J}_k|,
    \end{aligned}
\end{equation*}

where ($a$) holds because $\omega_r = \sigma^{k}_r$, and, for $j \in \{r+1,\ldots,s-1\} \setminus \{\sigma^{k}_r\}$, $\omega_j = j$; ($b$) holds because~$\omega_{\sigma^{k}_r} = s$; ($c$) holds because $\sigma^{k}_r < s$; ($d$) holds because if $j \in [r+1,\sigma^{k}_r-1]$, then $j \in [r+1,s-1] \setminus \{\sigma^{k}_r\}$, as $r < \sigma^{k}_r < s$.
\end{proof}

\begin{claim}\label{Dioporcellino}
Under both \emph{Case} 1 and \emph{Case} 2, $\Delta_{k} \leq \Delta_{\omega}$.
\end{claim}

\begin{proof}
First of all, recall that we are assuming $\sigma^{k}_r > r$. 
We first deal with $\Delta^{r}_{k}$. On the one hand, notice that under both \emph{Case} $1$ and \emph{Case} $2$, we have $\Delta^{r}_{\omega} = f_{\bb}(n-r)-f_{\bb}(n-\sigma^{k}_r)$. On the other hand,  we have
\begin{equation}\label{la madonna puzza 1}
\begin{aligned}
 | \{{\sigma^{k}_{r+1}},\ldots,{\sigma^{k}_{s-1}}\} \cap \{{\sigma^{k}_r+1},\ldots,s\} | &\leq \min \{ |\{{\sigma^{k}_{r+1}},\ldots,{\sigma^{k}_{s-1}}\}|, |\{\sigma^{k}_{r}+1,\ldots,s\}| \} \\
 &\leq s-\sigma^{k}_r, 
\end{aligned}
\end{equation}

and 

\begin{equation}\label{la madonna puzza 2}
\begin{aligned}
 |\{\sigma^{k}_{r+1},\ldots,\sigma^{k}_{s-1}\} \cap \{r+1,\ldots,\sigma^{k}_r\}| &\overset{(a)}{=} |\{\sigma^{k}_{r+1},\ldots,\sigma^{k}_{s-1}\} \cap \{r+1,\ldots,\sigma^{k}_r-1\}| \\
 &\leq \min \{ |\{{\sigma^{k}_{r+1}},\ldots,{\sigma^{k}_{s-1}}\}|, |\{r+1,\ldots,\sigma^{k}_r-1\}| \} \\
 &\leq \sigma^{k}_r-r-1, 
\end{aligned}
\end{equation}

where ($a$) holds because $\sigma^{k}_r \notin \{\sigma^{k}_{r+1},\ldots,\sigma^{k}_{s-1}\}$. Together these inequalities yield

\begin{equation}\label{caso Delta r}
    \begin{aligned}
    \Delta^{r}_k &= f_{\bb}(n-s+|\{\sigma^{k}_{r+1},\ldots,\sigma^{k}_{s-1}\} \cap \{r+1,\ldots,s\}|+1) \\
    &- f_{\bb}(n-s+|\{\sigma^{k}_{r+1},\ldots,\sigma^{k}_{s-1}\} \cap \{\sigma^{k}_r+1,\ldots,s\}|) \\
    &= f_{\bb}(n-s+|\{\sigma^{k}_{r+1},\ldots,\sigma^{k}_{s-1}\} \cap \{\sigma^{k}_r+1,\ldots,s\}|+|\{\sigma^{k}_{r+1},\ldots,\sigma^{k}_{s-1}\} \cap \{r+1,\ldots,\sigma^{k}_r\}|+1) \\
    &- f_{\bb}(n-s+|\{\sigma^{k}_{r+1},\ldots,\sigma^{k}_{s-1}\} \cap \{\sigma^{k}_r+1,\ldots,s\}|) \\
    &\overset{(a)}{\leq} f_{\bb}(n-s+(s-\sigma^{k}_r) + |\{\sigma^{k}_{r+1},\ldots,\sigma^{k}_{s-1}\} \cap \{r+1,\ldots,\sigma^{k}_r\}|+1) \\
    &- f_{\bb}(n-s+(s-\sigma^{k}_r)) \\
    &\overset{(b)}{\leq} f_{\bb}(n-\sigma^{k}_r + (\sigma^{k}_r-r-1)+1) - f_{\bb}(n-\sigma^{k}_r) \\
    &= f_{\bb}(n-r)-f_{\bb}(n-\sigma^{k}_r) = \Delta^{r}_{\omega},
    \end{aligned}
\end{equation}
where ($a$) holds by applying Lemma \ref{DioCinghiale} and Inequality (\ref{la madonna puzza 1}), and ($b$) holds by applying Lemma \ref{DioMajale} and Inequality (\ref{la madonna puzza 2}). Thus, under both \emph{Case} $1$ and \emph{Case} $2$, $\Delta_{k}^{r} \leq \Delta_{\omega}^{r}$. If $\mathcal{J}_k=\emptyset$, then since $\Delta^j_{\omega}\geq 0$ for all $j\in \mathcal{J}_{\omega}$, inequality \eqref{caso Delta r} yields $\Delta_k=\Delta^r_{k}\leq \Delta^r_{\omega}\leq \Delta_{\omega}$.

Therefore, in the remaining part of the argument, we assume that~$\mathcal{J}_{k} \neq \emptyset$. For all $j \in \mathcal{J}_{k}$, we have

\begin{equation}\label{j upper bound}
    \begin{aligned}
    |\{\sigma^{k}_{j+1},\ldots,\sigma^{k}_{s-1}\} \cap \{\sigma^{k}_{j}+1,\ldots,s\}| &\leq \min(|\{\sigma^{k}_{j+1},\ldots,\sigma^{k}_{s-1}\}|,|\{\sigma^{k}_{j}+1,\ldots,s\}|) \\
    &= \min(s-(j+1),s-\sigma^{k}_j) \leq s-j-1.
    \end{aligned}
\end{equation}

We now show that under both \emph{Case} $1$ and \emph{Case} $2$,
\begin{equation}\label{easy cases}
   \forall j \in \mathcal{J}_{k} \cap \mathcal{J}_{\omega}, \Delta^{j}_{k} \leq \Delta^{j}_{\omega}. 
\end{equation}
If $\mathcal{J}_{k} \cap \mathcal{J}_{\omega} = \emptyset$, then there is nothing to show, so suppose the intersection is non-empty. Let $j \in \mathcal{J}_{k} \cap \mathcal{J}_{\omega}$. On the one hand, notice that, under both \emph{Case} $1$ and \emph{Case}~$2$, we have $\Delta^{j}_{\omega} = f_{\bb}(n-j)-f_{\bb}(n-(j+1))$. Therefore, we obtain

\begin{equation}\label{j good case}
    \begin{aligned}
    \Delta^{j}_{k} &= f_{\bb}(n-s+|\{{\sigma^{k}_{j+1}},\ldots,{\sigma^{k}_{s-1}}\} \cap \{{\sigma^{k}_{j}+1},\ldots,s\}| + 1) \\
    &- f_{\bb}(n-s +|\{{\sigma^{k}_{j+1}},\ldots,{\sigma^{k}_{s-1}}\} \cap \{{\sigma^{k}_{j}+1},\ldots,s\}|) \\
    &\overset{(a)}{\leq} f_{\bb}(n-s+(s-j-1) + 1) - f_{\bb}(n-s + (s-j-1) ) \\
    &= f_{\bb}(n-j) - f_{\bb}(n-j-1) = \Delta^{j}_{\omega},
    \end{aligned}
\end{equation}
where ($a$) holds by Lemma \ref{DioCinghiale} and Inequality (\ref{j upper bound}).

By combining (\ref{easy cases}) and Inequality (\ref{caso Delta r}), we obtain that for both \emph{Case} $1$ and \emph{Case}~$2$, if $\mathcal{J}_{k} \subseteq \mathcal{J}_{\omega}$ then $\Delta_{k} \leq \Delta_{\omega}$.

We conclude the analysis of \emph{Case} $2$, by showing that $\mathcal{J}_{k} \subseteq \mathcal{J}_{\omega}$. Indeed, suppose that $j\in \mathcal{J}_k$, but~$\omega_j\notin \left[r+1,\omega_r-1\right]$. If $\omega_j<r+1$, then we have two possibilities: either $\omega_j=r$ which implies $j=s$, or $\omega_j<r$, which implies $\omega_j=j<r+1$. In both cases we have a contradiction, as~$j\in [r+1,s-1]$. If $\omega_j>\omega_r-1$, then $\omega_j>\sigma^k_r-1$ and we have two possibilities: either~$\omega_j=\sigma^k_r$ which implies $j=r$, or $\omega_j>\sigma^k_r=s$ which implies $j=\omega_j>s$. In both cases we obtain a contradiction. Thus, under \emph{
Case} $2$, it holds $\mathcal{J}_{k} \subseteq \mathcal{J}_{\omega}$.

Under \emph{Case} $1$, it might happen that $\mathcal{J}_{k} \setminus \mathcal{J}_{\omega} \neq \emptyset$. Suppose $\mathcal{J}_{k} \setminus \mathcal{J}_{\omega} \neq \emptyset$. Since by Claim \ref{Diomaialino} we have that $|\mathcal{J}_{\omega}| \geq |\mathcal{J}_{k}|$, there exists an injective map $\rho: \mathcal{J}_{k} \setminus \mathcal{J}_{\omega} \to \mathcal{J}_{\omega} \setminus \mathcal{J}_k$. Indeed, notice that 
$$|\mathcal{J}_{\omega} \setminus \mathcal{J}_{k}| = |\mathcal{J}_{\omega} \setminus (\mathcal{J}_{k} \cap \mathcal{J}_{\omega})| \overset{(a)}{\geq} |\mathcal{J}_{k} \setminus (\mathcal{J}_{k} \cap \mathcal{J}_{\omega})| = |\mathcal{J}_{k} \setminus \mathcal{J}_{\omega}| \overset{(b)}{>} 0, $$

where ($a$) holds because $|\mathcal{J}_{\omega}| \geq |\mathcal{J}_{k}|$ and we are subtracting from both the same set $\mathcal{J}_{k} \cap \mathcal{J}_{\omega}$; ($b$) holds under our assumption $\mathcal{J}_{k} \setminus \mathcal{J}_{\omega} \neq \emptyset$.
Furthermore, recall that by Claim \ref{Diomaialino}, we have $\mathcal{J}_{\omega} = [r+1,\sigma^{k}_r-1]$,  which, in turn, implies that $\mathcal{J}_{\omega} \setminus \mathcal{J}_{k} \subseteq [r+1,\sigma^{k}_r-1]$.

As $\mathcal{J}_k \subseteq [r+1,s-1]$ and $\sigma^{k}_r < s$, if $j \in \mathcal{J}_{k} \setminus \mathcal{J}_{\omega}$ necessarily $j \geq \sigma^{k}_r > \rho(j)$. Let~$j \in \mathcal{J}_{k} \setminus \mathcal{J}_{\omega}$ and $z:= \rho(j)$. We show that $\Delta^{j}_{k} \leq \Delta^{z}_{\omega}$. Notice that $\Delta^{z}_{\omega} = f_{\bb}(n-z)-f_{\bb}(n-(z+1))$. On the one hand, we have 
\begin{equation*}
    \begin{aligned}
    |\{\omega_{z+1},\ldots,\omega_{s-1}\} \cap \{\omega_{z}+1,\ldots,s\}| &= |\{{z+1},\ldots,{\sigma^{k}_{r}-1}, s, {\sigma^{k}_{r}+1},\ldots,{s-1}\} \cap \{{z+1},\ldots,s\}|  \\
    &= s-z-1 > s-j-1,
    \end{aligned}
\end{equation*}
where we understand that if $z = \sigma^{k}_r - 1$ then $$|\{\omega_{z+1},\ldots,\omega_{s-1}\} \cap \{\omega_{z}+1,\ldots,s\}| = |\{s, {\sigma^{k}_{r}+1},\ldots,{s-1}\} \cap \{\sigma^{k}_r,\ldots,s\}| = s - \sigma^{k}_r.$$ On the other hand, we obtain 
\begin{equation}\label{troietta}
    \Delta^{j}_{k} \overset{(a)}{\leq} f_{\bb}(n-j) - f_{\bb}(n-j-1) \overset{(b)}{\leq} f_{\bb}(n-z)-f_{\bb}(n-z-1) = \Delta^{z}_{\omega},
\end{equation}

where ($a$) holds by using Inequality (\ref{j upper bound}), and ($b$) holds by applying Lemma \ref{DioCinghiale} and the fact that~$z < j$. By combining (\ref{caso Delta r}), (\ref{easy cases}), and (\ref{troietta}), we conclude that, also under \emph{Case} $1$, $\Delta_{k} \leq \Delta_{\omega}$.
\end{proof}
We can now proceed to compute an upper bound. The following hold
\begin{equation}\label{a mirimeo piaxe l'oseo}
  \begin{aligned}
  \Delta_k &\overset{(a)}{\leq} \Delta_{\omega} \\
  &=\Delta_{\omega}^{r}+\sum_{j \in \mathcal{J}_{\omega}} \Delta_{\omega}^{j} \\
  &\overset{(b)}{=} f_\bb(n-r) - f_\bb(n-\sigma^{k}_{r}) +\sum_{j=r+1}^{\sigma^{k}_{r}-1} \left[ f_\bb(n-j) - f_\bb(n-j-1) \right] \\
  &\overset{(c)}{=}f_\bb(n-r-1) - f_\bb(n-\sigma^{k}_{r})+ \sum_{j=r}^{\sigma^{k}_{r}-1} \left[ f_\bb(n-j) - f_\bb(n-j-1) \right]  \\
  &= f_\bb(n-r)-f_\bb(n-\sigma^{k}_{r}) + f_\bb(n-r-1) - f_\bb(n-\sigma^{k}_{r}), 
  \end{aligned}
\end{equation}

where ($a$) holds by applying Claim (\ref{Dioporcellino}); ($b$) holds by applying Claim (\ref{Diomaialino}), and ($c$) holds by adding and subtracting $f_\bb(n-r-1)$. Since $\sigma^k_r>r$, we have that $\Delta^{\textnormal{apx}}_k>0$ and 



\begin{equation}\label{DIOCANEEEEEEEEE}
\begin{aligned}
  \frac{\Delta_k}{\Delta_{k}^{\textnormal{apx}}}
  &\overset{(a)}{\leq} \frac{f_\bb(n-r)-f_\bb(n-\sigma^{k}_{r}) + f_\bb(n-(r+1)) - f_\bb(n-\sigma^{k}_{r})}{\Delta_{k}^{\textnormal{apx}}} \\
  &= \frac{f_\bb(n-r)-f_\bb(n-\sigma^{k}_{r}) + f_\bb(n-(r+1)) - f_\bb(n-\sigma^{k}_{r})}{f_\bb(n-r)-f_\bb(n-\sigma^{k}_{r})} \\
  &= 1 + \frac{f_\bb(n-(r+1)) - f_\bb(n-\sigma^{k}_{r})}{f_\bb(n-r)-f_\bb(n-\sigma^{k}_{r})} \\
  &\overset{(b)}{\leq} 1 + \frac{f_\bb(n-(r+1))}{f_\bb(n-r)} \leq \gamma_\bb,
\end{aligned}
\end{equation}
where ($a$) holds by using (\ref{a mirimeo piaxe l'oseo}), and ($b$) follows observing that $x\mapsto \frac{a-x}{b-x}$ is decreasing in $[0,b)$ for all $a<b$. Finally, we conclude the proof considering all the cases analyzed above. In particular, we define 
$$\mathcal{K} := \left\{ k \in \{0,\ldots,K-1\} \mid \sigma^{k+1} = \sigma^{k} t_{r,s}, \; \textnormal{for some indices} \; s > r \; \textnormal{and } \sigma^{k}_{r} > r\right\}.$$
When $k\in \mathcal{K}$, by Claim \ref{sigma larger than r}, we have that $\Delta_k\geq 0$ and $\Delta_k^{\textnormal{apx}}>0$. If $k\notin \mathcal{K}$, then by Claim \ref{sigma less than r} we have that $\Delta_k\leq 0$ and $\Delta_{k}^{\textnormal{apx}}=0$. Moreover notice that $\mathcal{K}\neq \emptyset$, otherwise, by invoking Claim \ref{sigma less than r}, we would obtain that $\Delta_{k} \leq 0$, for all $0 \leq k \leq K-1$. This, in turn, would imply that $\lVert \sigma^{0} \rVert_{\beta} = \sum_{k=0}^{K-1} \Delta_{k} = 0$, which is absurd given that $\sigma^{0} \neq \Id$ and $\beta_2 > 0.$ Thus, $\mathcal{K}\neq \emptyset$. Furthermore, notice that since $\beta_2>0$ and $\sigma^0\neq \Id$, we have $
\lVert \sigma^{0} \rVert_{\beta}^{\textnormal{apx}}>0.$ These observations yield

\begin{equation*}
\begin{aligned}
\frac{\lVert \sigma^{0} \rVert_{\beta}}{ \lVert \sigma^{0} \rVert_{\beta}^{\textnormal{apx}}} = \frac{\sum_{k=0}^{K-1} \Delta_k}{ \sum_{k=0}^{K-1} \Delta_k^{\textnormal{apx}} } \overset{(a)}{=} \frac{\sum_{k=0}^{K-1} \Delta_k}{\sum_{k \in \mathcal{K}} \Delta_k^{\textnormal{apx}}} \overset{(b)}{\leq} \frac{\sum_{k \in \mathcal{K}} \Delta_k}{ \sum_{k \in \mathcal{K}} \Delta_k^{\textnormal{apx}} } \overset{(c)}{\leq} \frac{\sum_{k \in \mathcal{K}} \gamma_\bb \Delta_k^{\textnormal{apx}}}{ \sum_{k \in \mathcal{K}} \Delta_k^{\textnormal{apx}} } \leq \gamma_\bb.
\end{aligned}
\end{equation*}
\noindent
where ($a$) holds because, when $k \in \mathcal{K}^{c}$, $\Delta_k^{\textnormal{apx}} = 0$; ($b$) holds because, when $k \in \mathcal{K}^{c}$, $\Delta_k \leq 0$, and ($c$) holds by Inequality (\ref{DIOCANEEEEEEEEE}).
\end{proof}
}

\begin{proof}[Proof of Corollary \ref{corollary approximation}]
Let $\sigma^{\textnormal{apx}} \in P^{\textnormal{apx}}_\bb(V)$ and $\sigma^{*} \in P_\bb(V)$. The following inequalities hold
\begin{equation*}
\begin{aligned}
\min\limits_{\sigma\in \LL_n}\dd_\bb(\sigma,V)&\leq \sum_{j = 1}^{|V|} \dd_\bb(\sigma^{\textnormal{apx}},v^j)\overset{(a)}{\leq} \gamma_\bb \sum_{j = 1}^{|V|} \dd_\bb^{\textnormal{apx}}(\sigma^{\textnormal{apx}},v^j)\\
&\overset{(b)}{\leq} \gamma_\bb \sum_{j = 1}^{|V|} \dd_\bb^{\textnormal{apx}}(\sigma^{*},v^j) \overset{(c)}{\leq} \gamma_\bb \sum_{j = 1}^{|V|} \dd_\bb(\sigma^{*},v^j)=\gamma_\bb\min\limits_{\sigma\in \LL_n}\dd_\bb(\sigma,V),
\end{aligned}
\end{equation*}
where ($a$) and ($c$) hold by using Theorem \ref{general approximation thm}, and ($b$) holds as, by construction, $\sigma^{\textnormal{apx}}$ is a minimizer of $\sum_{j = 1}^{|V|} \dd_\bb^{\textnormal{apx}}(\cdot,v^j)$.
\end{proof}

\begin{proof}[Proof of Lemma \ref{hungarian lemma}]
Let $(v^j)_{j=1}^{m}\in \V$.
Construct the following bipartite graph $G = ([n],[n],E)$, where each vertex $c \in [n]$ corresponds to an alternative in $[n]$, and each vertex $p \in [n]$ corresponds to a possible ranking. Edges~$e = \{c,p\} \in E$ are interpreted as ``alternative $c$ is in position $p$''. We weight each edge $e = \{c,p\} \in E$ by
\[ w_{c,p} = \sum_{j = 1}^{|V|} \big | f_\bb(n-p) - f_\bb(|c^{\downarrow,v^j}|) \big |\]
Observe that we can compute all weights $\{w_{c,p}\}_{c,p}$ in time $O(m n^3)$. Indeed, by using Pascal's formula, we can pre-compute $\{f_\bb(0),\ldots,f_\bb(n-1)\}$ in time $O(n^2)$. Clearly, the number of edges $\{c,p\}$ is $n^2$. Fix $\{c,p\}$ and recall that we have a profile of $m$ preferences. For each preference, we need to compute $|c^{\downarrow, v^j}|$ and this can be done in linear time. Therefore, in total, it takes $O(m n^3)$ to compute all weights.

Further, note that each perfect matching is associated with a linear order $\sigma$, such that $|c^{\downarrow, \sigma} | = n-p$, whenever the edge $\{c,p\}$ belongs to the matching. Particularly, by solving for a minimum weight perfect matching, we obtain $ \sigma^{\textnormal{apx}} \in \LL_n$ that minimizes $\omega \mapsto \sum_{j = 1}^{|V|} \dd^{\textnormal{apx}}_\bb(\omega,v^j)$. Finally, we can compute a minimum weight perfect matching by relying on the Hungarian algorithm whose running time is $O(n^3)$. Thus, the overall running time is $O(m n^3)$.
\end{proof}
\begin{proof}[Proof of Theorem \ref{theoremungulato}]
Combine Theorem \ref{general approximation thm} with Corollary \ref{corollary approximation} and Lemma \ref{hungarian lemma}.
\end{proof}

\begin{proof}[Proof of Corollary \ref{corollariocinghiale}]
For all $\sigma,\pi\in \LL_n$,
\begin{equation*}
\begin{aligned}
\dd^{\textnormal{apx}}_{\mu,\bb}(\sigma,\pi) \overset{(a)}{\leq} \ddd(\sigma,\pi) &\overset{(b)}{\leq} U \dd_\bb(\sigma,\pi) \\
&\overset{(c)}{\leq} \gamma_\bb U \dd^{\textnormal{apx}}_\bb (\sigma,\pi) \\
&\overset{(d)}{\leq} \gamma_\bb \frac{U}{u} \dd^{\textnormal{apx}}_{\mu,\bb} (\sigma,\pi),
\end{aligned}
\end{equation*}

where ($a$) holds by applying the same arguments underlying Theorem \ref{general approximation thm}; ($b$) holds because $\mu$ is bounded above, ($c$) holds by applying Theorem \ref{general approximation thm}, and ($d$) because $\mu$ is bounded below.
\end{proof}

\begin{proof}[Proof of Theorem \ref{PTAS}]
Fix $\epsilon > 0$. Let $n \in \mathbb{N}$, $V = (v^j)_{j=1}^{m}\in \V$, $\bb_{2:n}$ and $\mu^{(n)}$ satisfying the assumption of Theorem \ref{PTAS}. To ease notation, we denote $\bb = \bb_{2:n}$ and $\mu = \mu^{(n)}$.  Suppose there exist $q \in \{0,\ldots,n-2\}$ alternatives $({\sigma_{1}},\ldots,{\sigma_{q}})$ such that, for all $k \leq q$,
\[ \left| \left\{ j \in [m] \mid \forall a \in C \setminus \{{\sigma_{1}},\ldots,{\sigma_{k}}\},\ {\sigma_{k}} >_{v^j} a \right\}\right| > \frac{m}{2}.\footnote{We understand that, when $q = 0$, there is no majoritarian alternative.}\]
Then, by invoking Proposition \ref{prop:majority}, Lemma \ref{split optimum} and Remark \ref{remark for majority} repeatedly, we conclude that there exists an optimal ranking $\sigma^{*}$ that assigns alternatives $({\sigma_{1}},\ldots,{\sigma_{q}})$ to the first $q$ positions. Denote by $\bar{\sigma}$ the ranking produced by Algorithm \ref{MyopicTop}. Observe that, by construction, 
\begin{equation}\label{equazione demmerda}
  \ddd(\sigma^{*},V;1,q) = \ddd(\bar{\sigma},V;1,q) \; \textnormal{and} \; \ddd(\sigma^{*},V;q+1,q+K) \geq \ddd(\bar{\sigma},V;q+1,q+K).
\end{equation}
Suppose that $\ddd(\sigma^{*},V) = 0$. As $\beta_2 > 0$, this can happen if and only if $v^{j} = \sigma^*$, for all $j\in [m]$. The while-loop in Algorithm \ref{MyopicTop} ensures that, by construction, $\bar{\sigma} = \sigma^{*}$. Therefore, in the following, we can assume $\ddd(\sigma^{*},V) > 0$. It holds
\begin{equation*}
\begin{aligned}
\frac{\ddd(\bar{\sigma},V)}{\ddd(\sigma^{*},V)} &= \frac{\ddd(\bar{\sigma},V;1,q) + \ddd(\bar{\sigma},V;q+1,q+K) + \ddd(\bar{\sigma},V;q+K+1,n) }{\ddd(\sigma^{*},V)} \\
&\overset{(a)}{\leq} \frac{\ddd(\sigma^{*},V;1,q) + \ddd(\sigma^{*},V;q+1,q+K) + \ddd(\bar{\sigma},V;q+K+1,n) }{\ddd(\sigma^{*},V)} \\
&= \frac{\ddd(\sigma^{*},V;1,q+K) + \ddd(\bar{\sigma},V;q+K+1,n)}{\ddd(\sigma^{*},V)} \\
&= \frac{\ddd(\sigma^{*},V)-\ddd(\sigma^{*},V;q+K+1,n) + \ddd(\bar{\sigma},V;q+K+1,n)}{\ddd(\sigma^{*},V)} \\
&\leq 1 + \frac{|\ddd(\sigma^{*},V;q+K+1,n)|}{\ddd(\sigma^{*},V)} + \frac{ \ddd(\bar{\sigma},V;q+K+1,n)}{\ddd(\sigma^{*},V)} \\
&\overset{(b)}{\leq} 1 + \frac{|\ddd(\sigma^{*},V;q+K+1,n)|}{\ddd(\sigma^{*},V)} + \frac{2 m U\sum_{i = q+K+1}^{n} f_\bb(n-i)}{\ddd(\sigma^{*},V)} \\
&\overset{(c)}{\leq} 1 + \frac{|\ddd(\sigma^{*},V;q+K+1,n)|}{\frac{u m }{2} \phi_{q+1}} + \frac{2 m U\sum_{i = q+K+1}^{n} f_\bb(n-i)}{\frac{u m }{2} \phi_{q+1}} \\
&\overset{(d)}{\leq} 1 + \frac{4 m U\sum_{i = q+K+1}^{n} f_\bb(n-i)}{\frac{u m }{2} \phi_{q+1}}+\frac{2 m U\sum_{i = q+K+1}^{n} f_\bb(n-i)}{\frac{u m }{2} \phi_{q+1}} \\
&= 1 + \frac{12 U \sum_{i = q+K+1}^{n} f_\bb(n-i)}{u \phi_{q+1}} \\
&= 1 + \frac{12 U \sum_{i = q+K+1}^{n} \sum_{j=2}^{n-i+1}\beta_j\binom{n-i}{j-1}}{u \phi_{q+1}}= 1 + \frac{12 U \sum_{j=2}^{n}\beta_j\sum_{i=q+K+1}^{n-j+1}\binom{n-i}{j-1}}{u \phi_{q+1}}\\
&=1+\frac{12 U \sum_{j=2}^{n}\beta_j\sum_{r=j-1}^{n-K-q-1}\binom{r}{j-1}}{u \phi_{q+1}}=1+\frac{12 U \sum_{j=2}^{n}\beta_j\sum_{r=0}^{n-K-q-j}\binom{j-1+r}{j-1}}{u \phi_{q+1}}\\
&=1+\frac{12 U \sum_{j=2}^{n}\beta_j\binom{n-K-q}{n-K-q-j}}{u \phi_{q+1}}= 1+\frac{12 U \sum_{j=2}^{n}\beta_j\binom{n-K-q}{n-K-q-j}}{u \sum_{j=2}^n\beta_j\binom{n-(q+1)-1}{j-2}},\\
&=1 + \frac{12 U\sum_{j=2}^{n-(q+K)} \beta_{j} \binom{n-(q+K)}{j}}{u \sum_{j=0}^{n-(q+2)} \beta_{j+2} \binom{n-(q+2)}{j}}.
\end{aligned}
\end{equation*}
where ($a$) holds by using the relations in (\ref{equazione demmerda}); ($b$) holds by applying Lemma \ref{lemmino per upper boundino}; ($c$) holds by applying Lemma \ref{lemmino per lower boundino}, and ($d$) holds by applying the following argument:
\begin{equation*}
  \begin{aligned}
  |\ddd(\sigma^{*},V;q+K+1,n)| &= \bigg|\sum_{j \in V} \sum_{i=q+K+1}^{n} f_\bb(n-i)\mu(\sigma_{i}) + f_\bb(n-i) \mu(v^j_{i}) \\
  &- 2 f_\bb \left( \left| {\sigma_{i}}^{\downarrow,\sigma} \cap {\sigma_{i}}^{\downarrow,v^j}\right| \right) \mu(\sigma_{i} ) \bigg| \\
  &\overset{(d1)}{\leq} \sum_{j \in V} \sum_{i=q+K+1}^{n} f_\bb(n-i)\mu({\sigma_{i}}) + f_\bb(n-i) \mu(v^j_{i}) \\ 
  &+ 2 f_\bb \left( \left| {\sigma_{i}}^{\downarrow,\sigma} \cap {\sigma_{i}}^{\downarrow,v^j}\right| \right) \mu(\sigma_{i} ) \\
  &\overset{(d2)}{\leq} 4m U \sum_{i=q+K+1}^{n} f_\bb(n-i),
  \end{aligned}
\end{equation*}
where ($d1$) holds by applying twice the triangular inequality, and ($d2$) holds by observing that $f_\bb \left( \left| {\sigma_{i}}^{\downarrow,\sigma} \cap {\sigma_{i}}^{\downarrow,v^j}\right| \right) \leq f_\bb(n-i)$ and $\mu_i \leq U$.
We have established 
\begin{equation*}
\begin{aligned}
\frac{\ddd(\bar{\sigma},V)}{\ddd(\sigma^{*},V)} &\leq 1 + \frac{12 U\sum_{j=2}^{n-(q+K)} \beta_{j} \binom{n-(q+K)}{j}}{u \sum_{j=0}^{n-(q+2)} \beta_{j+2} \binom{n-(q+2)}{j}} \\
&\leq 1 + \frac{12U}{u}\sup_{t \geq K \vee 2} \frac{\sum_{j=0}^{t-K} \beta_{j} \binom{t-K}{j}}{ \sum_{j=0}^{t-2} \beta_{j+2} \binom{t-2}{j}} \\
&\leq 1 + \epsilon,
\end{aligned}
\end{equation*}
for $K \geq g \left(\frac{12 U }{u \epsilon} \right)$. Therefore, by choosing $K = g \left(\frac{12 U }{u \epsilon} \right)$, we have
\[ \ddd(\bar{\sigma},V) \leq (1+\epsilon)\ddd(\sigma^{*},V). \]


Finally, on the one hand, observe that the while loop and the computation of $\{f_\bb(0),\ldots,f_{\bb}(n-1)\}$ can be done in polynomial time. On the other hand, we can minimize $\ddd(\cdot,V;q+1,q+K+1)$ by searching over all possible cases; when there are no majority alternatives, there are $O\left( K! \binom{n}{K} \right) = O \left( n^K\right)$ cases to be checked. Evaluating each case can be done in time $O(m K n)$, and so the overall running time of the algorithm turns out to be $O(n^{K+1}m K)$.
\end{proof}

\subsection{Computations}\label{comps}
Suppose $\beta_j = (\alpha-1)^{j}$ for $\alpha >1$ and $j\in [2,n]$. Pick $K \in \mathbb{N}$. Then, we have
\begin{equation*}
\begin{aligned}
h(t) := \frac{\sum_{j=2}^{t-K} (\alpha-1)^{j} \binom{t-K}{j}}{\sum_{j=0}^{t-2} (\alpha - 1)^{j+2} \binom{t-2}{j}} &= \frac{\alpha^{t-K} -(\alpha-1)(t-K)-1}{(\alpha-1)^2 \alpha^{t-2}} \\
&= \frac{\alpha^{-K} -\alpha^{-t}(\alpha-1)(t-K)-\alpha^{-t}}{(\alpha-1)^2 \alpha^{-2}}
\end{aligned}
\end{equation*}

Now, $\sup_{t \geq K \vee 2} h(t) = \sup_{t \geq K \vee 2}\frac{\alpha^{-K} -\alpha^{-t}(\alpha-1)(t-K)-\alpha^{-t}}{(\alpha-1)^2 \alpha^{-2}} = \frac{\alpha^{-K}}{\left(\frac{\alpha-1}{\alpha} \right)^2}$, as $-\alpha^{-t}(\alpha-1)(t-K)-\alpha^{-t} \leq 0$. Therefore, $\frac{\alpha^{-K}}{\left(\frac{\alpha-1}{\alpha} \right)^2} \leq \epsilon$ if and only if $K \geq \log_{\alpha} \left( \frac{\alpha^2}{(\alpha-1)^2 \epsilon}\right)$.

Suppose $\beta_j = j+1$ and $j\in [2,n]$. Pick $K \in \mathbb{N}$. Then, we have
\begin{equation*}
\begin{aligned}
h(t) := \frac{\sum_{j=2}^{t-K} (j+1) \binom{t-K}{j}}{\sum_{j=0}^{t-2} (j+3) \binom{t-2}{j}} &= \frac{(2^{t-K-1}-2)(t-K) + 2^{t-K}-1}{2^{t-3}(t+4)} \\
&= \frac{(t-K)(2^{-K-1} - 2 \times 2^{-t}) + 2^{-K} - 2^{-t}}{2^{-3}(t+4)} \\
&\overset{(a)}{\leq} \frac{(t-K) 2^{-K-1} + 2^{-K}}{2^{-3}(t+4)} \overset{(b)}{\leq} 4 \times 2^{-K},
\end{aligned}
\end{equation*}

where ($a$) holds because $t \geq K$, and ($b$) holds because $K \geq 1$. In particular, note that $\lim_{t \to \infty} h(t) = 4 \times 2^{-K}$. Therefore, $\sup_{t \geq K \vee 2} h(t) = 4 \times 2^{-K}$, which, in turn, implies that $\sup_{t \geq K \vee 2} h(t) \leq \epsilon$ if and only if $K \geq \log_2(4/\epsilon)$.

Finally, suppose $\beta_j = 1+(-1)^{j}$ for all $j\in [2,n]$. Pick $K \in \mathbb{N}$. Then, we have
\begin{equation*}
\begin{aligned}
h(t) := \frac{\sum_{j=2}^{t-K} (1+(-1)^j) \binom{t-K}{j}}{\sum_{j=0}^{t-2} (1+(-1)^{j+2}) \binom{t-2}{j}} &= \frac{\sum_{j=2}^{t-K} (1+(-1)^j) \binom{t-K}{j}}{\sum_{j=0}^{t-2} (1+(-1)^{j}) \binom{t-2}{j}} = \frac{\sum_{j=2}^{t-K} \binom{t-K}{j} + \sum_{j=2}^{t-K} (-1)^j \binom{t-K}{j}}{\sum_{j=0}^{t-2} \binom{t-2}{j} + \sum_{j=0}^{t-2} (-1)^{j} \binom{t-2}{j}} \\
&= \frac{(2^{t-K}-(t-K)-1) + \delta_{t-K} - (1-(t-K))}{2^{t-2} + \delta_{t-2}} \\
&\leq \frac{2^{t-K}}{2^{t-2}} = 4 \times 2^{-K}.
\end{aligned}
\end{equation*}

In particular, note that $\lim_{t \to \infty} h(t) = 4 \times 2^{-K}$. Therefore, $\sup_{t \geq K \vee 2} h(t) = 4 \times 2^{-K}$, which, in turn, implies that $\sup_{t \geq K \vee 2} h(t) \leq \epsilon$ if and only if $K \geq \log_2(4/\epsilon)$.


\end{document}